\documentclass[11pt]{article}

\usepackage[colorlinks=true,linkcolor=blue,citecolor=blue]{hyperref}
\usepackage{amsmath, amssymb, amsthm}
\usepackage{mathtools}
\usepackage{array}
\usepackage{float}
\usepackage[capitalize]{cleveref}
\crefname{equation}{}{}
\usepackage{fullpage}
\usepackage{graphics}
\usepackage{pifont}
\usepackage{tikz}
\usepackage[T1]{fontenc}

\DeclareMathOperator*{\argmin}{arg\,min}
\usetikzlibrary{arrows.meta}
\usepackage{environ}
\usepackage{framed}
\usepackage{url}
\usepackage[linesnumbered,ruled,vlined]{algorithm2e}
\usepackage[noend]{algpseudocode}
\usepackage[labelfont=bf]{caption}
\usepackage{cite}
\usepackage{framed}
\usepackage[framemethod=tikz]{mdframed}
\usepackage{appendix}
\usepackage{graphicx}
\usepackage{tcolorbox}
\usepackage{thm-restate}
\allowdisplaybreaks[1]

\newcommand\remove[1]{}

\newtheorem{theorem}{Theorem}
\newtheorem{lemma}{Lemma}[section]
\newtheorem{fact}[lemma]{Fact}
\newtheorem*{lemma*}{Lemma}

\newtheorem*{corollary*}{Corollary}
\newtheorem{claim}[lemma]{Claim}

\theoremstyle{definition}

\newtheorem*{theorem*}{Theorem}
\newtheorem{definition}[lemma]{Definition}

\newtheorem*{rem*}{Remark}

\newtheorem{obs}[lemma]{Observation}

\newcommand\R{\mathbb{R}}

\newcommand\E{\mathbb{E}}

\renewcommand{\bf}{\textbf}
\newcommand{\eps}{\epsilon}
\renewcommand{\O}{\widetilde{O}}
\newcommand{\pe}{\preceq}

\renewcommand{\l}{\langle}
\renewcommand{\r}{\rangle}
\newcommand{\bs}{\backslash}
\newcommand{\assign}{\leftarrow}

\renewcommand{\forall}{\mathrm{\text{ for all }}}

\newcommand{\diag}{\mathrm{diag}}

\newcommand{\new}{\mathrm{new}}

\newcommand{\mSC}{\mathbf{SC}}

\newcommand{\step}{\mathrm{step}}
\newcommand{\solve}{\mathrm{solve}}
\newcommand{\loc}{{\mathrm{(loc)}}}
\newcommand{\chk}{{\mathrm{(chk)}}}
\newcommand{\Recenter}{\textsc{Recenter}}
\newcommand{\rw}{\mathsf{RW}}
\newcommand\mL{\boldsymbol{L}}
\newcommand\mR{\boldsymbol{R}}

\renewcommand\deg{\boldsymbol{deg}}
\newcommand\zzero{\boldsymbol{0}}
\newcommand\oone{\boldsymbol{1}}

\newcommand{\pr}[2]{\mathsf{Pr}_{#1}\left[#2\right]}
\newcommand{\expec}[2]{\mathbb{E}_{#1}\left[#2\right]}

\newcommand{\err}{\mathrm{err}}
\newcommand{\cerr}{\mathrm{cerr}}
\newcommand{\expand}{\mathrm{(exp)}}

\newcommand\that{\widehat{t}}
\newcommand\muhat{\widehat{\mu}}

\newcommand\cchi{\boldsymbol{\chi}}

\newcommand\pphi{\boldsymbol{\phi}}

\newcommand\ppi{\boldsymbol{\pi}}
\newcommand\ppitil{\widetilde{\boldsymbol{\pi}}}

\newcommand\rrhotil{\widetilde{\boldsymbol{\rho}}}

\newcommand\dd{\boldsymbol{d}}
\newcommand\ddhat{\widehat{\boldsymbol{d}}}
\newcommand{\ddi}{\dd[i]}
\newcommand\ff{\boldsymbol{f}}
\newcommand\ffhat{\widehat{\boldsymbol{f}}}
\newcommand\fftil{\widetilde{\boldsymbol{f}}}
\newcommand\ffbar{\overline{\boldsymbol{f}}}
\renewcommand\gg{\boldsymbol{g}}
\newcommand\gghat{\widehat{\boldsymbol{g}}}
\newcommand\pp{\boldsymbol{p}}
\newcommand\qq{\boldsymbol{q}}
\newcommand\rr{\boldsymbol{r}}
\newcommand\rrhat{\widehat{\boldsymbol{r}}}
\newcommand{\Shat}{\widehat{S}}

\newcommand\uu{\boldsymbol{u}}
\newcommand\xxtil{\widetilde{\boldsymbol{x}}}

\newcommand\vv{\boldsymbol{v}}
\newcommand\xx{\boldsymbol{x}}
\newcommand\yy{\boldsymbol{y}}
\newcommand{\Ghat}{\widehat{G}}
\newcommand{\mRhat}{\widehat{\mR}}

\renewcommand{\sc}{\mathrm{(sc)}}

\newcommand{\abs}[1]{\left| #1 \right|}
\newcommand{\norm}[1]{\left\| #1 \right\|}
\newcommand{\bg}{\gg}
\newcommand{\cE}{\mathcal{E}}

\newcommand\Chat{\widehat{C}}
\newcommand\Ehat{\widehat{E}}
\newcommand{\init}{\mathrm{(init)}}
\newcommand{\fin}{\mathrm{(fin)}}
\newcommand{\DynamicSC}{\textsc{DynamicSC}}
\newcommand{\safe}{\mathrm{(safe)}}
\newcommand\prhit[3]{p_{#2}^{#3}\left(#1\right)}
\newcommand{\ndd}{N}
\newcommand{\change}{\mathrm{change}}
\newcommand{\betachk}{\beta_{\Checker}}
\newcommand{\betaloc}{\beta_{\Locator}}

\foreach \x in {A,...,Z}{
	\expandafter\xdef\csname m\x\endcsname{\noexpand\mathbf{\x}}
}

\DeclareMathOperator\polylog{polylog}

\newif\ifrandom
\randomtrue

\newcommand{\defeq}{\stackrel{\mathrm{\scriptscriptstyle def}}{=}}

\newcommand{\poly}{{\mathrm{poly}}}
\newcommand{\Checker}{\textsc{Checker}}
\newcommand{\Locator}{\textsc{Locator}}
\newcommand{\TrivialLocator}{\textsc{TrivialLocator}}

\newcommand{\vol}{{\hbox{\bf vol}}}

\crefname{algocf}{Algorithm}{Algorithms}

\newcommand{\tomato}{\top}

\begin{document}

\title{
Fully Dynamic Electrical Flows:\\
Sparse Maxflow Faster Than Goldberg-Rao
}

\author{
Y\v{u} Gao
\thanks{
Georgia Institute of Technology,
\texttt{ygao380@gatech.edu}}
\and
Yang P. Liu
\thanks{Stanford University,
\texttt{yangpliu@stanford.edu}}
\and
Richard Peng
\thanks{Georgia Institute of Technology \&  University of Waterloo,
\texttt{rpeng@cc.gatech.edu}}
}

\hypersetup{pageanchor=false}

\pagenumbering{gobble}
\begin{titlepage}
\clearpage\maketitle
\thispagestyle{empty}

\begin{abstract}
We give an algorithm for computing exact maximum flows on graphs with $m$ edges and
integer capacities in the range $[1, U]$
in $\widetilde{O}(m^{\frac{3}{2} - \frac{1}{328}} \log U)$ time.\footnote{We use $\widetilde{O}(\cdot)$ to suppress logarithmic factors in $m$.}
For sparse graphs with polynomially bounded integer capacities,
this is the first improvement over the
$\widetilde{O}(m^{1.5} \log U)$ time bound from [Goldberg-Rao JACM `98].

Our algorithm revolves around dynamically maintaining the augmenting
electrical flows at the core of the interior point method based
algorithm from [M\k{a}dry JACM `16].
This entails designing data structures that, in limited settings, return edges with large electric energy in a graph undergoing resistance updates.
\end{abstract}
\end{titlepage}

\hypersetup{pageanchor=true}

\newpage
\addtocontents{toc}{\protect\setcounter{tocdepth}{2}}
\tableofcontents

\newpage
\pagenumbering{arabic}
\section{Introduction}
\label{sec:intro}

The maxflow problem asks to route the maximum amount of flow between two vertices in a graph such that the flow on any edge is at most its capacity. The efficiency of this problem is well-studied and has numerous applications
in scheduling, image processing, and network science \cite{CLRS,GT14}. The main result of this paper is a faster exact maxflow algorithm on sparse directed graphs in the weakly polynomial setting, where the runtime depends logarithmically on the capacities.
\begin{theorem}
\label{thm:main}
There is an algorithm that on a graph $G$ with $m$ edges
and integer capacities in $[1, U]$
computes a maximum flow between vertices $s, t$ in time
$\O(m^{\frac32 - \frac{1}{328}} \log U).$
\end{theorem}
In sparse graphs with polynomially large capacities, this is the first improvement over the classical $O(m^{3/2}\log m \log U)$ time algorithm of Goldberg-Rao \cite{GR98}, which represented the culmination of a long line of work starting from the work of Hopcroft-Karp \cite{HK73} for bipartite matchings and Karzanov and Even-Tarjan for unit capacity maxflow \cite{K73,ET75}. Improving over this exponent of 3/2 for graph optimization problems has been intensively studied over the past decade via combinations of continuous optimization and discrete tools.

This line of work was initiated by Christiano-Kelner-M\k{a}dry-Spielman-Teng \cite{CKMST11} who gave a $\O(m^{4/3}\eps^{-O(1)})$ time algorithm for $(1+\eps)$-approximate maxflow. This has since been improved to $m^{1+o(1)}\eps^{-O(1)}$ \cite{S13,KLOS14} and the focus shifted to achieving improved $\eps$ dependencies \cite{S18,ST18} and exact solutions. Towards this, two breakthrough results were the $\O(m^{10/7}U^{1/7})$ time algorithm of M\k{a}dry \cite{M13,M16} which broke the $3/2$ exponent barrier on unweighted graphs, and the $\O(m\sqrt{n} \log U)$ time algorithm of Lee-Sidford \cite{LS19}, which was an improvement for any dense graph. Since then, these results respectively have been improved to yield algorithms that run in time $m^{4/3+o(1)}U^{1/3}$ \cite{LS20_STOC,KLS20} and $\O((m+n^{3/2})\log U)$ \cite{BLNPSSSW20,BLLSSSW20:arxiv}. However, the $3/2$ exponent of Goldberg-Rao \cite{GR98} remained the state-of-the-art on sparse capacitated graphs.

Classical approaches to solving maxflow use augmenting paths to construct the final flow.
Our algorithm, as well as the recent improvements above, instead computes the maxflow using a sequence of \emph{electric flows}. For resistances $\rr \in \R^E_{\ge0}$, the $s$-$t$ electric flow is the one that routes one unit from $s$ to $t$ while minimizing the quadratic energy:
\[
\min_{\ff \text{ routes one } s\text{-}t \text{ unit}} \sum_{e \in E} \rr_e\ff_e^2.
\]
Electric flows are induced by vertex potentials, and correspond to solving a linear system in the graph Laplacian. Motivated by this connection with scientific computing, two decades of work on combinatorial preconditioners led to the breakthrough result by Spielman-Teng~\cite{ST11} that Laplacian systems and electrical flows can be computed to high accuracy in $\O(m)$ time.

Our algorithm, as well as the recent faster runtimes for dense
graphs~\cite{BLNPSSSW20,BLLSSSW20:arxiv}, are built upon the
dynamic processes view of flow augmentations~\cite{GN80,ST83} that provided much
impetus for the study of dynamic graph data structures.
In this view, the final flow is obtained via a sequence of flow
modifications, and dynamic tree data structures such as link-cut trees~\cite{ST83, GN80} are designed to allow for sublinear time
identification and modification of edges that limit flow progress. Concretely, the maxflow is built using a sequence of $\O(\sqrt{m})$ electric flows on graphs with slowly changing resistances. This corresponds to the celebrated interior point method (IPM henceforth) which shows that linear programs can be solved using $\O(\sqrt{m})$ slowly changing linear system solves \cite{K84,Vaidya89}. To implement this framework we design data structures that on a graph with dynamic changing resistances:
\begin{itemize}
\item Identify all edges $e$ with at least an $\eps^2$ fraction of the total electric energy in the electric flow $\ff$ on a graph with resistances $\rr$: \begin{align*} \rr_e\ff_e^2 \ge \eps^2 \sum_{e \in E} \rr_e\ff_e^2. \end{align*}
\item Estimate the square root of energy or flow value of an edge up to an additive error of $\pm \eps/10 \cdot \sqrt{\sum_{e \in E} \rr_e\ff_e^2},$ i.e. a $\eps/10$ fraction of the square root of the total electric energy.
\end{itemize}

Finally, we leverage this data structure along with several modifications to the outer loop to achieve our main result \cref{thm:main}.

\subsection{Key Algorithmic Pieces}
\label{subsec:keyalgo}

At a high level, our algorithm implements an IPM which augments $\O(\sqrt{m})$ electric flows by building a data structure that detects large energy edges in an $s$-$t$ electric flow on a dynamic graph.
In addition to this, our algorithm requires several modifications to the IPM. First, our data structure requires properties specific to $s$-$t$ electric flows to achieve its guarantees. Consequently, we are forced to design an IPM that only augments via $s$-$t$ electric flows. On the other hand, a standard IPM alternates between routing electric flows and routing additional electric circulations every step. Second, our data structures are randomized and thus their outputs may affect future inputs when applied within the IPM. This requires delicately modifying our algorithm to bypass this issue. We now give more detailed descriptions of each piece.

\paragraph{Locating high energy edges in $s$-$t$ electric flows.} Our data structures for dynamic electric flows are based on the interpretation of electrical flow as random walks on the graph \cite{doyleS84}, which has been used previously for dynamic effective resistances \cite{DGGP19,CGHPS20}. In our setting we wish to detect edges with at least $\eps^2$ fraction of the $\ell_2$ electric energy. To achieve this, we use a spectral vertex sparsifier, which approximates the electric flow and potentials on this smaller set of \emph{terminal} vertices. We use this sparsifier as well as additional random walks to maintain the result of an $\ell_2$ heavy hitter sketch on the electric flow vector. This allows us to approximately maintain a short sketch vector and thus recover the large entries of the electric flow vector.

Our data structure has several subtleties which affect its interaction with the outer loop.
First, it is essential that the electrical flows maintained are $s$-$t$ to ensure additional stability in our algorithms.
$s$-$t$ electrical flows have additional, sharper, upper bounds on vertex potentials and flow values on edges, which do not hold for electrical flows with more general demands.
Secondly, we only maintain an approximate $\ell_2$ heavy hitter sketch but argue that this suffices for detection of large energy edges (\cref{lem:heavyhitter}).

\paragraph{IPM with $s$-$t$ flows.} We must modify the IPM outer loop to interact with our dynamic electric flow data structure described above which fundamentally uses properties specific to $s$-$t$ flows. The standard IPM \cite{Ren88} which uses electric flows to solve maxflow \cite{DS08,M16,LS20_STOC} has both a \emph{progress} phase where an $s$-$t$ electric flow is augmented, and a \emph{centering} phase where electric circulations are added to slightly fix the flow.

We modify the IPM to only use $s$-$t$ electric flows to make more than $1/\sqrt{m}$ progress before we pay $\O(m)$ time to center using electric circulations. We leverage two key properties of the method to achieve this. First, we argue that damping the step size of the IPM causes errors to accumulate more slowly. This allows us to use several $s$-$t$ electric flow steps (maintained in sublinear time by data structures) as opposed to flows with general demands before a centering step. Also, to argue this formally we use the fact that the resistances are multiplicatively stable to within a polynomial factor of the number of steps of standard size $1/\sqrt{m}$.

\paragraph{Randomness in data structures and adaptivity.} Because we are applying randomized data structures inside an outer loop, their outputs may affect future inputs. In the literature, this is referred to as an \emph{adaptive adversary}. On the other hand, our data structures na\"{i}vely only work against \emph{oblivious adversaries}, where the inputs are independent of the outputs and randomness of the data structure.

We handle these issues by carefully designing our data structures and outer loop to not leak randomness between components, instead of making our data structures deterministic or work against adaptive adversaries in general. We start by breaking the data structure into a \Locator~and a \Checker, based on ideas from \cite{FMPSWX18}. The \Locator~returns a superset that contains all edges with large energy with high probability, and the \Checker~independently estimates the energies of those edges to decide whether to update them. This way, the randomness of \Locator~does not affect its inputs. However, the outputs of \Checker~may affect its inputs. Now, we leverage that the sequence of flows encountered during the IPM outer loop are almost deterministic, and there are only a few iterations between deterministic instances. This way, we can use a separate \Checker~for each of these iterations before resetting every \Checker~to the deterministic instance.

\subsection{Heuristic Runtime Calculation}
\label{subsec:heuristic}

The following key properties of the IPM outer loop are necessary to understand why a sublinear time data structure suffices to achieve a $m^{3/2-\Omega(1)}$ time algorithm for capacitated maxflow.
\begin{enumerate}
    \item Computing electric flows on graphs whose resistances are within $1\pm\gamma$ of the true resistances suffices to make $1/\sqrt{m}$ progress (for some parameter $\gamma = \widetilde{\Omega}(1)$).
    \item The resistances change slowly multiplicatively throughout the course of the algorithm. In fact, at most $\O(T^2\gamma^{-2})$ edges have their resistances change by at least $1\pm\gamma$ multiplicatively over $T$ steps of the method for any $\gamma$ (\cref{lemma:l2change}). In particular, over all $\O(\sqrt{m})$ iterations of electric flow computation, each edge's resistance changes $\O(1)$ times on average.
    \item The resistance of an edge is approximately the inverse of its residual capacity squared. This way, an edge's resistance changes significantly if the electric energy of the edge is large in the computed electric flows.
\end{enumerate}

If we have a data structure which detects edges with large energies in $m^{1-\eta}$ amortized time per edge for some constant $\eta > 0$, then we can leverage it along with the above facts to design the following algorithm.
We take steps in batches of size $k$, after which we pay $\O(m)$ time to fix and recenter our flow to find the true underlying resistances. During each batch, we use the data structure to detect all edges whose resistance changed by more than $1+\widetilde\Omega(1)$ multiplicatively, and return their resistances.

Now we estimate the runtime of this algorithm. The cost of recentering is $O(m) \cdot \O(m^{1/2}/k) = \O(m^{3/2}/k),$ as there are $\sqrt{m}$ total steps and we recenter every $k$ iterations. Also, by the second item above that at most $\O(k^2)$ edges have their resistances change significantly during a batch, so the data structure takes $\O(m^{1-\eta}k^2)$ time per batch. The total time used by the data structure is therefore $\O(m^{1-\eta}k^2) \cdot \O(\sqrt{m}/k) = \O(m^{3/2-\eta}k).$ Taking $k = m^{\eta/2}$ gives a final runtime of $\O(m^{3/2-\eta/2}),$ which is less than $m^{3/2}$ as desired.
The tradeoffs in our algorithms are significantly
higher and more complicated in reality: we have higher exponents on the batch size $k$ due to compounding errors in the method,
and we have additional layers of intermediate rebuilds.
Nonetheless, the final tradeoffs by which we obtain Theorem~\ref{thm:main} are still similar in spirit.

\subsection{Related Work and Discussion}
\label{subsec:related}

There is a long history of work on the maximum flow problem, as well as work related to each of our key algorithmic pieces in \cref{subsec:keyalgo}: dynamic graph data structures, IPMs in the context of data structures, and random and adaptivity in data structures.

Our discussion below focuses on algorithms 
whose capacity dependence is logarithmic (weakly polynomial). The weakly polynomial setting also is equivalent to the setting where the edge capacities are positive real numbers, and we wish to compute an $\eps$-approximate solution in runtime depending on $\log(1/\eps)$.
In the strongly polynomial setting, where the algorithm runtime has no capacity dependence,
following early work of \cite{Kar74,GN80,ST83},
the best known maxflow runtime is $O(mn)$ and $O(n^2/\log n)$ when $m = O(n)$ \cite{Orlin13, KRT94}.


\paragraph{Maxflow Algorithms}

Network flow problems are widely
studied in operations research,
theoretical computer science,
and optimization~\cite{GT14}.
Among the many variants, the capacitated maxflow problem
captures key features of both combinatorial graph algorithms and numerical optimization routines.
As a result, it has an extensive history starting from the work
of Dinic and Edmonds-Karp \cite{Dinitz70, EK73}.
The seminal work by Edmonds-Karp~\cite{EK73} presented two
algorithms: an $O(n^2m)$ strongly polynomial time algorithm
by finding shortest augmenting paths, and an $O(m^2 \log{U})$
weakly polynomial time algorithm based on finding bottleneck
shortest paths.
Improving these algorithms provided motivation for dynamic tree 
data structures~\cite{GN80}, dual algorithms~\cite{GT88},
and numerical primitives such as scaling~\cite{GR98}.
These progress culminated in a runtime of
$\O(\min(m^{3/2}, mn^{2/3})\log U)$: for more details,
we refer the reader to the review by Goldberg and
Tarjan~\cite{GT14}.

In the two decades since Goldberg-Rao~\cite{GR98}, all improvements on the exact maximum flow problem rely on
continuous optimization techniques.
These include the $\O(m\sqrt{n}\log U)$ runtime of Lee-Sidford \cite{LS19}, and several results culminating in a $\O(m/\eps)$ runtime for $\eps$-approximate maxflow on undirected graphs \cite{CKMST11,S13,KLOS14,P16,S17,ST18}.
Additionally, a line of work \cite{M13,M16,LS20_STOC} achieving a $m^{4/3+o(1)}$ runtime in uncapacitated graphs \cite{KLS20} by using weight changes and $\ell_p$-norm flows \cite{KPSW19} to eliminate high energy edges, as opposed to our approach of using data structures to detect them. Recently, approaches that combine interior point methods (IPMs) with graphical data structures achieved a $\O((m + n^{3/2}) \log U)$ runtime for maxflow \cite{BLNPSSSW20,BLLSSSW20:arxiv}.
In this way, the bound of Golberg-Rao \cite{GR98} has been improved in higher error approximate settings (on undirected graphs), for uncapacitated graphs, and for dense capacitated graphs.
However, our result is the first to show an improvement for exact maxflow in the weakly polynomial parameter regime central to the line of work spanning from Edmonds-Karp~\cite{EK73} to
Goldberg-Rao~\cite{GR98}:
sparse directed graphs with polynomially bounded capacities.

\paragraph{Data Structures for IPMs.}
Starting from early work of Karmarkar \cite{K84} and Vaidya \cite{Vaidya89},
several results leverage the fact that the linear systems resulting from IPMs are slowly changing, and that only approximate solutions are needed to implement the method. In this way, data structures for efficiently maintaining the inverse of dynamically changing linear systems have been used to speed up IPMs for linear programming \cite{LS15,CLS19,B20a,B20b,BLSS20,BLNPSSSW20} and recently semidefinite programming \cite{JKLPS20}. Additionally, our algorithm uses the fact the multiplicative change in resistances is at most polynomial in the number of steps taken. While this type of result was previously known\footnote{Personal communication with Yin Tat Lee and Aaron Sidford \cite{LSpersonal}, also similar in spirit to \cite[Lemma 67]{LS19}.}, we are not aware of other IPM analyses that use this fact.

In the graphical setting of maxflow, this corresponds to dynamically maintaining electric flows in a graph with changing resistances.
Our result is heavily motivated by the recent~\cite{BLNPSSSW20}
and its follow-up \cite{BLLSSSW20:arxiv}
which obtained $\O(m + n^{1.5})$ type running times
for flow problems. The flow-based version of these results use dynamic sparsification algorithms to maintain approximate electric flows in $\O(n)$ instead of $\O(m)$ time per iteration.
Additionally these works required several other techniques to achieve their runtimes, including robust central paths/different measures of centrality, and weighted barriers. While we do not use these pieces in our algorithm, we are optimistic that understanding how to apply these techniques could improve the runtime of our method.

Also, $\ell_2$ heavy hitters are used in \cite{BLNPSSSW20,BLLSSSW20:arxiv} and our algorithms; however, we open up the standard statement of $\ell_2$ heavy hitter \cite{KNPW11} to prove that an approximate matrix-vector product suffices to implement the heavy hitter data structure (\cref{lem:heavyhitter}). Critically, we treat the heavy hitter sketch matrix as demands on which we compute electric flows which allows for interaction with random walks and spectral vertex sparsification.

\paragraph{Dynamic graph data structures.} The data structures we use to make sublinear time steps in interior
point methods broadly belong to data structures maintaining approximations to optimization problems in dynamically changing
graphs~\cite{OR10,GP13,BS15,BHN16,ADKKP16,HKN18,FG19,CGHPS20}.
Our maintenance of electrical flows is most directly related to 
dynamic effective resistance data structures~\cite{GHP17,GoranciHP18,DGGP19,CGHPS20}.
In particular, they heavily rely on dynamic vertex sparsifiers,
which by itself has also received significant attention in
data structures~\cite{PSS19,Goranci19:thesis,JS20:arxiv}.
In particular, our sublinear runtime comes in part from maintaining a spectral vertex sparsifier onto a smaller vertex subset.

\paragraph{Adaptivity and randomness.} Our data structures are randomized, and are accessed in an adaptive manner:
queries to it may depend on its own output.
While there has been much recent work on making randomized sparsification
based data structures more resilient against such adaptive
inputs~\cite{NS17,W17,NSW17,SW19,CGLNPS19,BBG20},
our approach at a high level bypasses most of these issues because the
(non-robust, unweighted) central path of IPMs is a fixed object.
In this way, our randomized data structures are essentially
pseudo-deterministic \cite{GG11,GGR13}: while the algorithm is randomized, the output is the same with high probability.
Additionally, the top-level interactions of our randomized components involve calling one data structure inside another to hide randomness. This has much in common with the randomized approximate min-degree algorithm from~\cite{FMPSWX18}.

\subsection{General Notation and Conventions}
\label{subsec:notation}

We use plaintext to denote scalars,
bold lower case for vectors, and bold upper case for matrices.
A glossary of variables and parameters is given in Appendix~\ref{appendix:VAR}.
We will use the $\widehat{\cdot}$ notation to denote
a later, modified, copy of a variable.
As our update steps are approximate, we will also use the 
$\widetilde{\cdot}$ notation to denote approximate/error carrying
versions of true variables.

We use $\O(\cdot)$ to suppress logarithmic factors in $m$ and $\widetilde{\Omega}(\cdot)$ to suppress the inverse logarithmic factors in $m$.
We let $\zzero, \oone \in \R^n$ denote the all zeroes/ones vectors respectively. For vectors $\xx, \yy \in \R^n$ we let $(\xx \circ \yy)_i \defeq \xx_i\yy_i.$
When context is clear, we also use $\frac{\xx}{\yy}$ to  denote 
the entry-wise division of two vectors, that is
$\left(\frac{\xx}{\yy}\right)_i
\defeq
\frac{\xx_i}{\yy_i}.$
We use $\abs{\xx}$ and $\abs{\mY}$ to denote the entry-wise absolute values of vector $\xx$ and matrix $\mY$.

Instead of tracking explicit constants in our parameters, we sometimes use $c$ and $C$ to denote sufficiently small (respectively large) absolute constants. E.g., for a parameter $k > 1$, we write $\eps = ck^{-6}$ to denote that there is a constant $c$ where $\eps = ck^{-6},$ and we will set $c$ later to be sufficiently small. $c$ and $C$ may denote different constants in different places. We use ``with high probability'' or ``w.h.p.'' to mean with probability at least $1-n^{-10}$.

We say that a symmetric matrix $\mM \in \R^{n \times n}$ is positive semidefinite (psd) if $\xx^\top\mA\xx \ge 0$ for all $\xx \in \R^n.$ For psd matrices $\mA, \mB$ we write $\mA \pe \mB$ if $\mB - \mA$ is psd. For positive real numbers $a, b$ we write $a \approx_{\gamma} b$ to denote $\exp(-\gamma)b \le a \le \exp(\gamma)b.$
For psd matrices $\mA, \mB$ we write $\mA \approx_{\gamma} \mB$ if $\exp(-\gamma)\mB \pe \mA \pe \exp(\gamma)\mB$.

\subsection{Organization of Paper}
\label{subsec:organization}

The remainder of the paper is organized as follows. In \cref{sec:overview2} we elaborate on each major piece of our algorithm introduced in \cref{subsec:keyalgo}: dynamic electric flow data structures, our modified IPM outer loop, and handling of randomness and adaptive adversaries. Then in \cref{sec:Flows} we give the linear algebraic formulation of the maximum flow problem. We then introduce the key notion of electric flows and its relationship with linear systems and random walks.

The remainder of the paper is organized as follows. In \cref{sec:checker} we first build a dynamic spectral vertex sparsifier and apply it to build a \Checker~data structure for estimating flows on edges. In \cref{sec:Locator} we extend this to build a \Locator~data structure for heavy hitters of electric flows, i.e. detecting edges in the electric flow with large energies. In \cref{sec:pathcorrect} we formally give the interior point method setup and argue that we can use the above \Checker~and \Locator~data structures to give a \textsc{RecenteringBatch}~procedure that makes more than $m^{-1/2}$ progress in amortized $\O(m)$ time. We additionally show several stability bounds that are essential for analyzing the runtime. In \cref{sec:proofmain} we explain how to trade off all parameters to formally argue our main result (\cref{thm:main}).

Finally, the appendix contains several omitted proofs in \cref{sec:proofs}, and a table of variables, notations, and parameters are given in \cref{appendix:VAR}.
\section{Overview of Approach}
\label{sec:overview2}

In this section we elaborate on the key pieces of our approach described in \cref{subsec:keyalgo}: dynamic electrical flow data structures (\cref{sec:overviewlocator}), an interior point method for maxflow using this data structure (\cref{sec:overviewipm}), and how to handle issues with randomness and adaptive adversaries (\cref{sec:overviewadaptive}).

\subsection{Overview of \Locator~for Dynamic Electric Flows}
\label{sec:overviewlocator}

Recall the dynamic electric flow problem we solve. For a graph $G = (V, E)$ with changing resistances $\rr \in \R^E_{\ge0}$ such that the energy of the electric flow $\ff$ is at most $1$ always, i.e. $\sum_{e \in E} \rr_e\ff_e^2 \le 1$, return a set of at most $\O(\eps^{-2})$ edges $S \subseteq E$ that contains all edges with energy at least $\eps^2$, i.e. $\rr_e\ff_e^2 \ge \eps^2$ for $e \in S$. We wish to solve this in amortized sublinear time per resistance update.

At a high level, our approach is based on the vertex sparsification view towards data structures. In this view, we achieve sublinear runtimes by maintaining an object onto a smaller subset of terminal vertices $C \subseteq V$ that approximately preserve the desired property in our data structure. For example, in our setting we will leverage \emph{spectral vertex sparsifiers} that maintain the electrical properties of the graph onto the set of terminals, such as pairwise effective resistances. Alternatively, this can be viewed as maintaining the spectral properties of the inverse of the graph Laplacian (and is known as the \emph{Schur complement}).
In our algorithms, the set $C$ will increases in size throughout our data structure to ensure that edge changes happen within $C$. Hence the focus is on maintaining properties onto $C$ while new vertices are added to it throughout the algorithm.

We detect edges with large electric energies by first setting up a linear $\ell_2$ heavy hitter sketch \cite{KNPW11} against the energy vector $\mR^{1/2}\ff$, where $\mR$ is the diagonal matrix of resistances and $\ff$ is the electric flow. We then approximately maintain the sketch using random walks and spectral vertex sparsifiers. At a high level, an $\ell_2$ heavy hitter sketch works by estimating the total $\ell_2$-norm / energy of various edge subsets using Johnson-Lindenstrauss sketches up to accuracy $\eps$. In this way, for $\O(\eps^{-2})$ sketch vectors $\qq \in \{-1, 0, 1\}^m$, we must maintain the quantity $\l \qq, \mR^{1/2}\ff \r$. Now we relate the electric flow to the electric potentials $\pphi$ using Ohm's law: for any edge $e = (u, v)$ we have $\ff_e = (\pphi_u - \pphi_v)/\rr_e.$ Written algebraically, this is $\ff = \mR^{-1}\mB\pphi$ where $\mB$ is the (unweighted) edge-vertex incidence matrix of the graph $G$. Plugging this into our previous formula gives us
\[ \l \qq, \mR^{1/2}\ff \r = \l \qq, \mR^{-1/2}\mB\pphi \r = \l \mB^\top\mR^{-1/2}\qq, \pphi \r. \]
For simplicity we now let $\dd \defeq \mB^\top\mR^{-1/2}\qq.$ Intuitively, the sketch vector $\qq$ is inducing a demand $\dd$ on the vertices which we now want to dot against the vertex potentials $\pphi$.

Now our goal is to use a smaller set of terminal vertices $C$ to estimate the quantity $\l \dd, \pphi \r.$ We achieve this by leveraging the fact that we can recover potentials outside $C$ by \emph{harmonically extending} the potentials restricted to $C$: $\pphi_C$. Precisely, the potential $\pphi_v$ at vertex $v \neq s,t$ is the average of its neighbors, weighted proportional to inverse resistances. Equivalently, starting a random walk at a vertex $v \notin C$ and taking exit edges proportional to inverse of resistances is a martingale (preserves mean) on the potentials. In this way we can write $\pphi = \mathcal{H}\pphi_C$ where $\mathcal{H} \in \R^{V(G) \times C}$ is this extension operator. Hence \[ \l \dd, \pphi \r = \l \dd, \mathcal{H}\pphi_C \r = \l \mathcal{H}^\top\dd, \pphi_C \r. \]
To compute this final quantity we must maintain $\mathcal{H}^\top\dd$ and $\pphi_C$ efficiently in sublinear time. For the former, given our random walk interpretation of $\mathcal{H}$, we may interpret $\mathcal{H}^\top\dd$ as the vector given by ``projecting'' $\dd$ onto the terminal set $C$ via random walks, and we write $\ppi^C(\dd) \defeq \mathcal{H}^\top\dd$ (\cref{def:proj}). In other words, the demand vector $\dd$ is distributed onto $C$ based on the probabilities that random walks from vertices $v$ hit $C$ for the first time. This interpretation of $\mathcal{H}^\top\dd$ allows us to build random walks to simulate the changes to this vector under the terminal set $C$ growing in size. For the latter, we maintain $\pphi_C$ by using the approximate spectral vertex sparsifier of \cite{DGGP19} which approximately maintains the Laplacian inverse on $C$ and hence the potentials. This construction is also based on running random walks from edges outside $C$ until they hit $C$.

We briefly elaborate on how resistance updates affect the terminal set $C$ and the random walks we maintain. We start by initializing $C$ to be a random set of size $\beta m$. (The reader can imagine $\beta = m^{-0.01}$ so that $|C|$ is sublinear.) We run random walks from each edge or vertex until it hits $C$. These walks are short, specifically visiting $\O(\beta^{-1})$ distinct vertices with high probability, because $C$ was chosen to be $\beta m$ random vertices.
Now, in general when the resistance of an edge $e = (u, v)$ is changed we add both endpoints $u, v$ of $e$ to $C$. Now the edge $e$ will be contained fully inside $C$ so we can directly perform the resistance change. However we must update our random walks due to $C$ changing. To do this we shortcut each random walk we computed to when it hit the larger set $C$ and update the necessary properties. A depiction of this process is given in \cref{fig:SC}.
\begin{figure}[!ht]
    \centering
    \includegraphics[width=\textwidth]{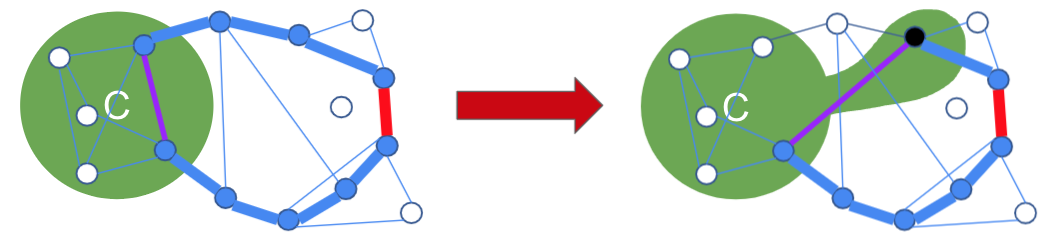}
    \caption{Shortcutting random walks from the red edge to a terminal set $C$ when under the insertion of the black vertex to $C$. \label{fig:SC}}
\end{figure}

To conclude, we describe some difficulties with the approach described above, specifically pertaining to maintaining the projected demand $\mathcal{H}^\top\dd = \ppi^C(\dd)$. The first concern is that entries of $\dd = \mB^\top\mR^{-1/2}\qq$ are too large if some edge $e$ has resistance $\rr_e$ close to $0$ (as then $\rr_e^{-1/2}$ is large). We handle this with the observation that edges with small resistances cannot have large energies in an $s$-$t$ electric flow (\cref{lemma:badres}), so we can restrict our heavy hitter sketch to edges with sufficiently large resistances. Also, na\"{i}vely estimating the projection $\ppi^C(\dd)$ with random walks from each vertex accumulates too much variance because the demand vector $\dd$ is dense. Instead we exactly compute $\ppi^C(\dd)$ by solving a linear system to start, and we estimate the \emph{change} in this vector under insertions to $C$ by locally sampling random walks from the inserted vertex (\cref{fact:projchange}). Finally, we periodically recalculate this vector to ensure that error does not accumulate.

\subsection{Overview of Interior Point Method}
\label{sec:overviewipm}

In this section we formalize the outer loop that our algorithm uses to argue that $\O(\sqrt{m})$ approximate electric flow computations suffice to compute a maxflow. We assume that the graph $G$ is undirected \cite{Lin09,M13} and that we know the optimal maxflow value $F^*$ by a standard binary search reduction. Given this, the \emph{central path} is a sequence of flows $\ff(\mu)$ for $\mu \in (0, F^*]$ defined by the minimizers of a logarithmic barrier potential:
\begin{align}
\ff\left(\mu \right)
\defeq
\argmin_{\mB^\top \ff = \left(F^{*} - \mu \right) \cchi_{st}}V\left(\ff\right) \enspace \text{ for } \enspace V(\ff) \defeq \sum_{e \in E} -\log(\uu_e - \ff_e)-\log(\uu_e+\ff_e). \label{eq:logbarrier}
\end{align}
Note that for $\mu = F^*$ that $\ff(F^*) = \textbf{0}$, the zero flow. Starting there, the goal of our algorithm is to follow this central path by slowly decreasing $\mu$ towards $0$ while computing the flows $\ff(\mu)$ along the way. While $\mu$ never equals $0$ exactly, the flows $\ff(\mu)$ approach the maximum flow as $\mu$ approaches $0$. We want to emphasize that the sequence of flows encountered by the algorithm along the way is \emph{deterministic} in this sense, as the minimizer of the convex problem \eqref{eq:logbarrier} is unique.

We remark that this central path (which is adapted from \cite{M16}) differs from the more standard central path used to solve mincost flow with cost $\mathbf{c}^\top \ff$. While this version can also work by setting $\mathbf{c}$ as a large negative cost on an $s$-$t$ edge, we choose to work with our formulation because the intuition that we are augmenting by $s$-$t$ electric flows is useful for our data structure based approach.

Now consider trying to decrease the path parameter $\mu$ to $\muhat < \mu$ starting from the current central path flow $\ff(\mu)$. Then we wish compute a flow $\Delta\ff$ which routes $\mu-\muhat$ units from $s$ to $t$ such that adding $\Delta\ff$ to our current flow $\ff(\mu)$ gets to the minimizer of \eqref{eq:logbarrier} for $\muhat$, i.e. $\ff(\muhat) = \ff(\mu) + \Delta\ff.$ While directly computing $\Delta\ff$ exactly is more difficult, one can show that up to a first order approximation, $\Delta\ff$ is given by the electric flow that routes $\mu-\muhat$ units from $s$ to $t$, with resistances given by $\rr_e = (\uu_e-\ff(\mu)_e)^{-2} + (\uu_e+\ff(\mu)_e)^{-2}.$

To handle the approximations induced by using $s$-$t$ electric flows, we require another fact: if we are able to calculate a flow $\fftil$ that is ``close'' to $\ff(\muhat)$ on all edges, then we can compute $\ff(\muhat)$ exactly using $\O(m)$ additional time (by computing additional electric circulations). Here, $\fftil$ is close to $\ff(\muhat)$ if all residual capacities differ by at most a multiplicative $1.1$ factor (\cref{lemma:recenter}). Now, one can show that if $\muhat = (1-c/\sqrt{m})\mu$ for a small constant $c$, $\Delta\fftil$ is the electric $s$-$t$ flow routing $\mu-\muhat$ units (even with approximate resistances), and $\fftil = \ff(\mu) + \Delta\fftil$, then $\fftil$ is close to $\ff(\muhat)$. Thus, this gives a method that terminates in $\O(\sqrt{m})$ iterations and $\O(m^{3/2})$ time.

\begin{figure}[!ht]
    \centering
    \begin{tikzpicture}[scale=0.85]
    \node[anchor=west] at(0, 9.6){$\ff(\mu)$};
    \draw (0,9.6) to[out = 270, in = 180] (12,0);
    \draw (0,9.6) [dashed,->]to (2, 7.2);
    \draw (2,7.2) [double, ->] to (0.5, 6.8);
    \node[anchor=east] at(0.5, 6.8){$\ff((1 - 1/\sqrt{m})\mu)$};
    \draw (0.5,6.8) [dashed,->]to (0.7, 4);
    \draw (0.7,4) [double, ->] to (1.8, 4.4);
    \node[anchor=west] at(1.8, 4.4){$\ff((1 - 2/\sqrt{m})\mu)$};
    \draw (1.8,4.4) [dashed,->]to (5, 3.2);
    \draw (5,3.2) [double, ->] to (4.3, 2.2);
    \node[anchor=east] at(4.3, 2.2){$\ff((1 - 3/\sqrt{m})\mu)$};
    \draw (4.3,2.2) [dashed,->]to (8, 1.7);
    \draw (8,1.7) [double, ->] to (7.7, 0.6);
    \node[anchor=east] at(7.7, 0.6){$\ff((1 - 4/\sqrt{m})\mu)$};
    \draw (7.7,0.6) [dashed,->]to (12, 0.8);
    \draw (12,0.8) [double, ->] to (12, 0);
    \node[anchor=west] at(12, 0){$\ff((1 - 5/\sqrt{m})\mu)$};
    \draw (0,9.6) [dotted, ->] to (0.1, 8.9);
    \draw (0.1, 8.9) [dotted, ->] to (0.3, 8.2);
    \draw (0.3, 8.2) [dotted, ->] to (0.5, 7.5);
    \draw (0.5, 7.5) [dotted, ->] to (0.9, 6.8);
    \draw (0.9, 6.8) [dotted, ->] to (1.2, 6.2);
    \draw (1.2, 6.2) [dotted, ->] to (1.6, 5.6);
    \draw (1.6, 5.6) [dotted, ->] to (2.0, 5.0);
    \draw (2.0, 5.0) [dotted, ->] to (2.5, 4.4);
    \draw (2.5, 4.4) [dotted, ->] to (2.9, 3.9);
    \draw (2.9, 3.9) [dotted, ->] to (3.5, 3.4);
    \draw (3.5, 3.4) [dotted, ->] to (4.1, 3.0);
    \draw (4.1, 3.0) [dotted, ->] to (4.6, 2.6);
    \draw (4.6, 2.6) [dotted, ->] to (5.3, 2.2);
    \draw (5.3, 2.2) [dotted, ->] to (6.2, 1.9);
    \draw (6.2, 1.9) [dotted, ->] to (6.9, 1.6);
    \draw (6.9, 1.6) [dotted, ->] to (7.6, 1.3);
    \draw (7.6, 1.3) [double, ->] to (7.7, 0.6);
    \draw (7.7, 0.6) [dotted, ->] to (8.6, 0.3);
    \draw (8.6, 0.3) [dotted, ->] to (9.6, -0.1);
    \draw (9.6, -0.1) [dotted, ->] to (10.7, -0.3);
    \draw (10.7, -0.3) [dotted, ->] to (11.9, -0.5);
    \end{tikzpicture}
    
    \caption{Our algorithm is split into $\O(\sqrt{m}/k)$ \emph{batches} of steps. Each batch is split into $\O(k^4)$ smaller steps to reduce error so that we can recenter in $\O(m)$ time at the end of each batch (\cref{lemma:recenter}). Each small step implements an approximate electric flow using data structures. \label{fig:ipm}.
    In the diagram, the dashed line is standard path following which recenters (depicted by the double arrows) after each step of size $1/\sqrt{m}$. The dotted line is the finer steps which takes more total steps but only recenters after $4/\sqrt{m}$ progress.}
\end{figure}
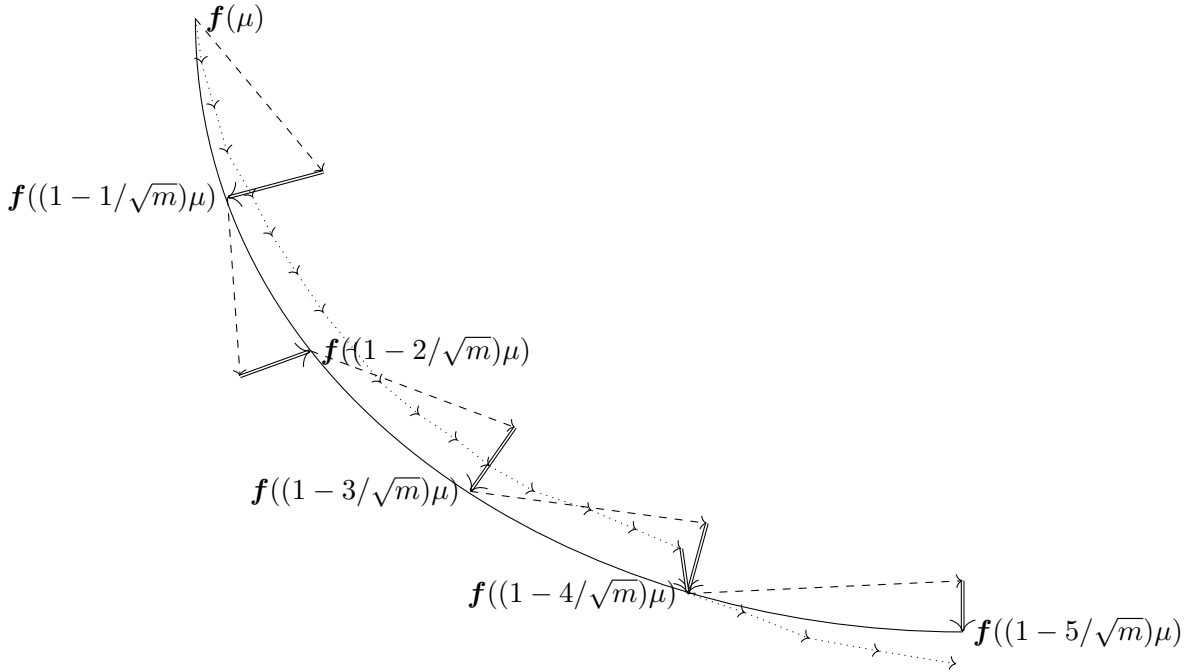

To achieve a $m^{3/2-\Omega(1)}$ time maxflow algorithm using IPMs we must be able to to decrease $\mu$ to $\muhat = (1-k/\sqrt{m})\mu$ for some $k = m^{\Omega(1)}$ in $\O(m)$ amortized time. This would achieve a $\O(m^{3/2}/k)$ time algorithm. Directly adding the electric flow routing $\mu-\muhat = k\mu/\sqrt{m}$ units from $s$ to $t$ accumulates too much error. We instead split this step into a batch of smaller steps, each which is an electric flow routing $(\muhat-\mu)/k^4 = \mu/(k^3\sqrt{m})$ units. We show in \cref{sec:pathcorrect} that because the electric flow is the first order approximation to the change in the central path, and because residual capacities are stable within a $O(k^2)$ factor during the step (\cref{lemma:stable}), that this sufficiently reduces error.

Now our method approximately implements each of the smaller steps in the batch using the data structure described in \cref{sec:overviewlocator}. We would like to emphasize again that even though the flows encountered during the small steps within a batch are randomized, we can pay $\O(m)$ at the end of each batch to move our flow back to the exact minimizer of \eqref{eq:logbarrier} so that it is deterministic. A depiction of the batches, splits into small steps, and recentering is given in \cref{fig:ipm}.

\subsection{Overview of Handling of Randomness}
\label{sec:overviewadaptive}
In this section we explain how to adapt our IPM outer loop and data structures to ensure that randomness in the data structures used to produce outputs does not affect the distribution of future inputs to itself. To this end, let us recall our setup described in the above \cref{sec:overviewlocator,sec:overviewipm}. We have a heavy hitter data structure which returns a set $S$ of edges that contains all edges with an $\eps^2$ fraction of the energy, and estimates their energies up to additive error.

Our first step towards addressing the randomness issue is to decouple the data structure. We split it into two parts: the first part which returns edges with large energies (\Locator), and a separate part which estimates again the energies of returned edges (\Checker). Our reasons for doing this are twofold -- it both helps with reasoning about where randomness arises in the algorithm, and provides mild runtime improvements.

\begin{figure}[!ht]
    \centering
    \includegraphics[width=\textwidth]{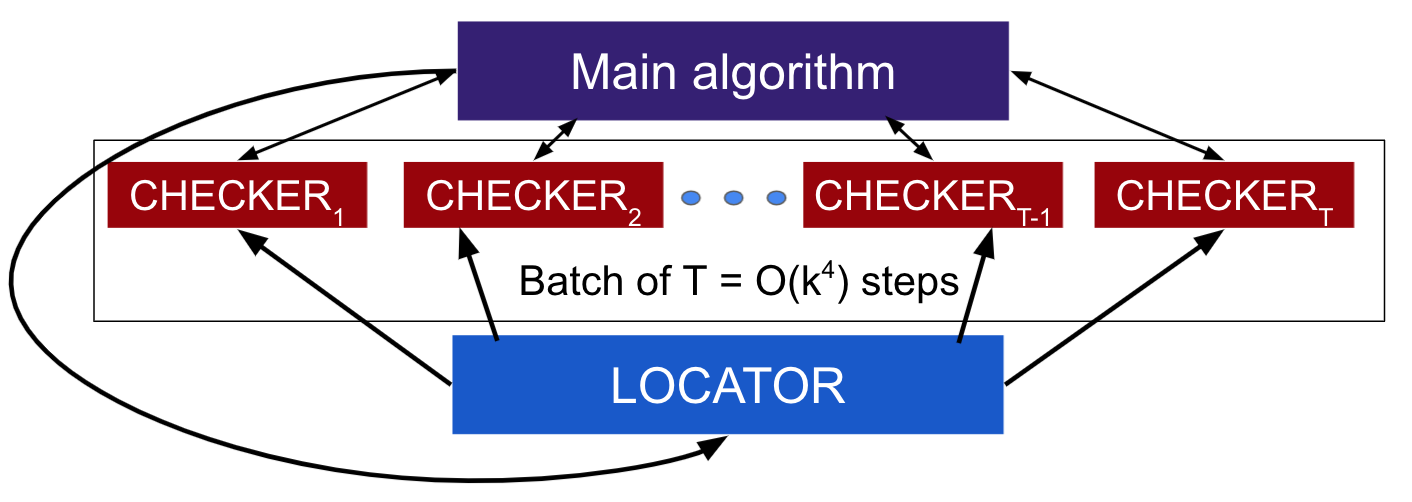}
    \caption{Within a batch of steps, during the $i$-th step the \Locator~passes an edge subset to the $i$-th \Checker. This communicates with the outside loop which passes updates to later \Checker s as well as the \Locator. \label{fig:adaptivity}}
\end{figure}

In this new setup the \Locator~corresponds to the heavy hitter, and returns a set of at most $O(\eps^{-2})$ edges that contains all edges with at least $\eps^2/10$ fraction of the electric energy. This set is fed to the \Checker~which independently estimates the amount of electric flow on that edge for each edge in the set. The \Checker~wishes to accept any edge with at least $\eps^2$ fraction of the electric energy and to estimate its flow value. Due to our IPM setup, these data structures are used within an outer loop consisting of $\O(\sqrt{m}/k)$ batches of steps, each which is split into $k^4$ smaller steps. After each batch, the algorithm perfectly moves back to the minimizer of \eqref{eq:logbarrier} in $\O(m)$ time and updates resistances. For each smaller step within the batch, we call \Locator~and \Checker~together to find edges with large energies / flows, and hence must have their resistances updated.
A depiction of the interactions between the \Locator, \Checker~data structures, and the algorithmic outer loop is given in \cref{fig:adaptivity}.

Note that the flow that we maintain is deterministically equal to the minimizer of \eqref{eq:logbarrier} at the start and end of each batch, so we can essentially update both the data structures deterministically. Hence we focus on ensuring the property that our data structures outputs do not affect futures inputs or states during the $k^4$ steps within a batch. We first describe why \Locator~can be assumed to be against oblivious adversaries, i.e. inputs are independent of the randomness. To understand why this is the case, consider the algorithm that does not use the \Locator~data structure at all, and instead uses the \Checker~to independently estimate the flow on every single edge and decides whether it believes the edge to have high energy. Clearly this algorithm is valid. We argue that using the \Locator~data structure simulates this algorithm that checks every edge. Indeed, we may set the thresholds for \Locator~so that any edge that \Checker~decides to update with non-negligible probability is included in the set of edges \Locator~returns with high probability. In this way, the outputs of \Locator~do not affect its future inputs as long as \Checker~is checking each edge independently.

While this explains why \Locator~may operate against oblivious adversaries, the same is unclear for \Checker. Indeed, different flow value estimates for an edge $e$ affect whether the resistance of $e$ is changed, and this can affect the internal state of \Checker~itself even during the same batch. To handle this, we actually construct $k^4$ independent \Checker~data structures, which we call $D^{\chk}_i$ for $i \in [k^4]$, one for each small step within a batch. We use $D^{\chk}_i$ to handle the set of edges that \Locator~returns at small step $i$ out of $k^4$. After this step, we stop updating $D^{\chk}_i$ until the end of the batch. At that time, we roll back all changes made to $D^{\chk}_i$ during the batch and then deterministically update its state to the new exact minimizer flow we compute. In this way, we can argue that the outputs of $D^{\chk}_i$ can only affect inputs of $D^{\chk}_j$ for $j > i$, so no $D^{\chk}_i$ has inputs affecting itself. In this way, we may assume that each $D^{\chk}_i$ actually operates against oblivious adversaries.

\section{Preliminaries: Maxflow and Electrical Flows}
\label{sec:Flows}
We start by formally defining the maxflow problem,
the electrical flow subroutine, and key objects for
representing both problems.
We will use $G = (V, E)$ to denote graphs,
$\uu$ to denote edge capacities,
and $\ff$ to denote flows.
We will also use
$\deg_v$ to denote the combinatorial/unweighted degree
of vertex $v$ in $G$, that is,
$\deg_v=|\{e\in E\mid e\ni v\}|$.

\subsection{Maxflow}

One can reduce directed maxflow to undirected maxflow with linear time overhead \cite{Lin09,M13}, so we assume our graph $G=(V, E)$ is undirected throughout. $m$ will be the number $|E|$ of edges and $n$ will be the number $|V|$ of vertices. We assume $m\le n^2$.

We also assume that $G$ is connected and has at least two vertices and one edge. Thus, each vertex of $G$ has at least one edge incident to it.
By standard capacity scaling techniques \cite{AO91} we may assume that $U = \poly(m)$ throughout this paper. Also, we assume we know the optimal numbers of units $F^*=\poly(m)$, as our algorithm works for any underestimate.
Furthermore, our algorithm actually works for general demand maxflows, as we can add a super source $s$ and super sink $t$ to accumulate to positive (respectively negative) demands on vertices.

We can then formalize the decision version of maxflow via
linear algebra.
Define $\mB$ to be the edge-vertex incidence matrix of $G$:
\[
\mB \in \R^{E \times V}
\qquad
\mB_{eu} =
\begin{cases}
1 & \text{if $u$ is the head of $e$}\\
-1 & \text{if $u$ is the tail of $e$}\\
0 & \text{otherwise}
\end{cases}
\]
and $\cchi_{st}$ to be the indicator vector with $-1$ at
source $s$, $1$ at sink $t$ and $0$ everywhere else.
Routing $F^{*}$ units of flow from $s$ to $t$
then becomes finding $\ff \in \R^{E}$ such that
\[
\mB^{\tomato} \ff = F^{*} \cchi_{st}
\qquad \text{and} \qquad
-\uu \leq \ff \leq \uu.
\]

\subsection{Electrical Flows}

Electrical flows are $\ell_2$-minimization analogs of maxflow,
and underlie all interior point method oriented approaches to
high-accuracy maxflow~\cite{DS08,M13,LS19,M16,CMSV17}.
 
We use the term \textit{demand vector} for any vector $\dd$ such that $\dd\in \mathbb{R}^n$ and  $\oone^\top \dd=0$. We let $\rr\in\mathbb{R}^E$ be the vector of resistances: $\rr_e$ denotes the resistance of edge $e$. For a demand vector $\dd$, and the vector of resistances $\rr$,
the electrical flow problem is
\[
\min_{\ff: \mB^{\tomato} \ff = \dd}
\sum_{e} \rr_e \ff_e^2.
\]
Here the energy function can be further abbreviated
using the norm notation: by letting $\mR$ denote the
diagonal matrix with $\rr$ on the diagonal, the energy
can be written as $\norm{\ff}_{\mR}^2$.
The quadratic minimization nature of this problem means its
solution, or the optimal electrical flow, has a linear
algebraic closed form, specifically
\[
\ff 
=
\mR^{-1} \mB \left( \mB^{\tomato} \mR^{-1} \mB \right)^{\dag} \dd,
\]
where $\dag$ denotes the Moore-Penrose pseudoinverse.
The matrix $\mB^\top\mR^{-1}\mB$ is important on its own, and is known as the graph Laplacian matrix, $\mL = \mB^{\tomato} \mR^{-1} \mB$.
Laplacian systems can be solved to high accuracy in nearly linear time \cite{ST04, KMP10, KMP11, KOSZ13, CKMPPRX14, KLPSS16, KS16}.
The resulting solution vector on the vertices also have natural
interpretations as voltages that induce the electrical flow~\cite{doyleS84}.
Specifically, for the voltages
\[
\pphi
=
\mL^{\dag} \dd
=
\left( \mB^{\tomato} \mR^{-1} \mB \right)^{\dag} \dd
\]
the flow is given by \emph{Ohm's Law}:
\[
\ff_{e}
= \frac{\pphi_{u} - \pphi_{v}}{\rr_e}
\qquad
\forall e = (uv).
\]
Both this flow, and the voltages, can be computed to
high accuracy in nearly-linear time using
Laplacian solvers~\cite{ST04}.
\begin{theorem} 
\label{thm:lap}
Let $G$ be a graph with $n$ vertices and $m$ edges. Let $\rr \in \R_{>0}^E$ denote edge resistances. For any demand vector $\dd$ and $\eps>0$ there is an algorithm which computes in $\O(m \log \eps^{-1})$ time \emph{potentials} $\pphi$ such that $\|\pphi - \pphi^*\|_{\mL} \le \eps\|\pphi^*\|_\mL$, where $\mL = \mB^\top \mR^{-1}\mB$ is the Laplacian of $G$, and $\pphi^* = \mL^\dagger \dd$ are the true potentials determined by the resistances $\rr$.
\end{theorem}

Critical to our data structures are the intuition of electrical
flows as random walks.
Specifically, that the unit electrical flow from $s$ to $t$ is the
expected trajectory of the random walk from $s$ to $t$, with cancellations,
where from vertex $u$ we go to $v \sim u$ with probability 
\[
\frac{\rr_{uv}^{-1}} {\sum_{w \sim u} \rr_{uw}^{-1}}
\]
where the reciprocal of resistances, conductance,
plays a role analogous to the weight of edges.
Many of our intuitions and notations have overlaps with the electrical
flow based analyses of sandpile processes~\cite{DFGX18}.
For a more systematic exposition, we refer the reader to the excellent monograph by Doyle and Snell~\cite{doyleS84}.
\section{Dynamic Schur Complements and Checking High Energy Edges}
\label{sec:checker}

The main goal of this section is to show the following procedure for supporting electrical flow queries on edges for dynamically changing graphs.
\begin{theorem}
\label{thm:checker}
There is a \textsc{Checker} data structure supporting the following operations with the given runtimes against oblivious adversaries, for constants $0 < \beta, \eps < 1$,
\begin{itemize}
    \item $\textsc{Initialize}(G, \rr^\init \in \R^{E(G)}_{>0}, \eps, \beta).$ Initializes the data structure with a graph $G$ where edge $e$ has resistance $\rr_e$. Runtime: $\O(m\beta^{-4}\eps^{-4}).$
    \item $\textsc{PermanentUpdate}(e, \rr^\new_e \in \R_{>0}).$ Update $\rr_e \assign \rr^\new_e$ . Runtime: amortized $\O(\beta^{-2}\eps^{-2})$.
    \item $\textsc{TemporaryUpdate}(e, \rr^\new_e >0).$ Update $\rr_e \assign \rr^\new_e$. Runtime: Worst case $\O\left((K\beta^{-2}\eps^{-2})^2\right)$ for $K$ \textsc{TemporaryUpdate}s that are not rolled back. All \textsc{TemporaryUpdate}s should be rolled back before the next \textsc{PermanentUpdate}.
    \item $\textsc{Rollback}().$ Rolls back the last \textsc{TemporaryUpdate} if exists. Costs the same time as the original operation.
    \item $\textsc{Check}(e).$ Let $\cE$ be the energy of a unit $s$-$t$ electric flow $\Delta\ff$. If the energy of edge $e$ is at least $\eps^2\cE$, edge $e$ must be accepted. If edge $e$ is accepted, additionally returns a real number $\bg_e$ satisfying
    \begin{align} \left|\rr_e^{1/2}(\bg_e-\Delta\ff_e)\right| \le \eps\sqrt{\cE}/10. \label{eq:chkguarantee} \end{align}
    If the energy of $e$ is at most $\eps^2\cE/2,$ edge $e$ must not be accepted.
    Runtime: worst case $\O\left(\left(\beta m+\left(K\beta^{-2}\eps^{-2}\right)^2\right)\eps^{-2}\right)$ where $K$ is the number of $\textsc{TemporaryUpdate}$ operations that are not rolled back.
    \textbf{Additionally, the output of $\textsc{Check}(e)$ is independent of any previous calls to $\textsc{Check}.$}
\end{itemize}
Finally, the probability that all calls to $\textsc{Check}(e)$ return valid outputs is at least $1-n^{-10}.$ The total number of \textsc{PermanentUpdate} and \textsc{TemporaryUpdate} that are not rolled back should not exceed $\beta m$.
\end{theorem}
Our approach is based on that of \cite{DGGP19} which builds dynamic spectral vertex sparsifiers or Schur complements, which we introduce below in \cref{sec:dynamicspectralsparsifier}. However we give a self-contained exposition here because we must adapt various guarantees of the algorithm for our setting, in large part to deal with randomness and adaptivity issues. This also explains the occurrence of \textsc{PermanentUpdate} and \textsc{TemporaryUpdate} in our data structure guarantee: calls to \textsc{PermanentUpdate} intuitively result from deterministic data structure changes, while \textsc{TemporaryUpdate} is for changes which we later wish to undo. This allows us to more carefully control the randomness in \Checker.
We formally construct the dynamic Schur complement data structure in \cref{subsec:dynamicsc} and apply it to show the \Checker~(\cref{thm:checker}) in \cref{subsec:checkersub}.

\subsection{Preliminaries for Dynamic Schur Complements}
\label{sec:dynamicspectralsparsifier}

We introduce various properties of spectral vertex sparsifiers, or Schur complements, that we require throughout this section and the next. Additionally, we introduce dynamic expander decompositions and spectral sparsification which we need within our dynamic Schur complement data structure.

\subsubsection{Schur Complements}

Our data structures use approximate spectral vertex sparsifiers for graph Laplacians, also known as Schur complements.
For a non-empty subset of vertices $C \subseteq V$,
with the rest of the vertices denoted as $F = V \setminus C$, let $\mL_{FF}, \mL_{CF}, \mL_{CC}$ denote the respective blocks of the Laplacian $\mL$ induced by the rows/columns corresponding to $F$ and $C$. The \emph{Schur complement} $\mSC(\mL, C)$  of $\mL$ onto $C$
is given by
\[
\mSC\left(\mL, C\right)
\defeq
\mL_{CC} - \mL_{CF}\mL_{FF}^{-1}\mL_{FC}.
\]
Note that $\mL_{FF}$ is full rank by the assumption that the
graph is connected. It is known that $\mSC(\mL, C)$ is also a graph Laplacian.
When context is clear, we use $\mSC(G, C)$ instead of $\mSC(\mL, C)$.

Such smaller graphs are significant because they directly provide
the solution to the Laplacian system on the subset $C$.
For vector $\xx \in \R^n$, we define the vector
$\xx_C \in \R^C$ as shorter vector formed by the entries in $C$.
The following property of Schur complements is critical to
all dynamic data structures for effective resistances to
date~\cite{Goranci19:thesis,GoranciHP18,DGGP19}, and is also
used throughout elimination based linear systems
solvers~\cite{KLPSS16,KS16,CKKPPRS18}.
\begin{fact}
\label{fact:scsubset}
For any $\mL$, strict subset $C$, and vector $\xx \in \R^n$ such
that $\xx_{V \setminus C} = \zzero$ for $i \notin C.$ Then
$(\mL^\dagger\xx)_C = \mSC(\mL, C)^\dagger (\xx_C).$
\end{fact}

In order to obtain sublinear time, these data structures take a more
local interpretation of electrical flows.
This is done by considering the Schur complement as the result of random
walking the edges of $G$ until both endpoints are in $C$.
In order to formalize this, we first need to define the probability
of a random walk hitting $C$ at some vertex $v \in C$.
Here we follow from the notation from~\cite{S18,SRS18}. A variant of this notation was also central to~\cite{DFGX18}.

\begin{definition}
\label{def:hit}
Let $G=(V, E, \rr)$ be a graph with resistances $\rr$,
and $C \subseteq V$ be a subset of vertices.
For vertices $u \in V$, $v \in C$, define $p^C_v(u)$
to be the probability that a random walk (picked proportional to resistance inverses) starting at $u$
reaches $v$ before any vertex in $C\setminus \{v\}$.
\end{definition}
A folklore result that has become increasingly important
in elimination based algorithms~\cite{KLPSS16,CKKPPRS18,DGGP19} is that the Schur complement is the expectation of such random walks.
\begin{fact}
For a graph $G$ and vertex subset $C$, the weight (inverse resistance)
of an edge $(c_1, c_2)$ in $\mSC(G, C)$ with $c_1, c_2 \in C$ is given by
\[
\sum_{e = \left(u, v\right) \in E\left( G \right)}
p^{C}_{c_1}\left( u \right) \cdot
p^{C}_{c_2}\left( v \right) \cdot
\rr_e^{-1}.
\]
\end{fact}
While we do not use this directly, its intuition is critical for our
sublinear time access to portions of electrical flows.
Specifically, operator approximations of Schur complements can be obtained
by sampling slight variants of such walks.
\begin{lemma}[Schur complement approximation, \cite{DGGP19} Theorem 3.1]
\label{lemma:scapprox}
Let $G = (V, E, \rr)$ be an undirected, weighted multigraph with a subset of vertices $T$. Furthermore, let $\eps \in (0, 1)$ and let $\rho = P\eps^{-2}\log n$. Let $H$ be an initially empty graph with vertices $T$, and for each edge $e = (u, v) \in E(G)$ repeat the following procedure $\rho$ times.
\begin{enumerate}
    \item Simulate a random walk from $u$ until it hits $T$ at $t_1$.
    \item Simulate a random walk from $v$ until it hits $T$ at $t_2$.
    \item Combine these random walks (along with edge $e = (u, v)$) to form a walk $W$.
    \item Add edge $(t_1, t_2)$ to $H$ with resistance $\rho\sum_{e \in W} \rr_e.$
\end{enumerate}
The resulting graph $H$ satisfies $\mL(H) \approx_{\eps} \mSC(G, T)$ with probability at least $1 - n^{-\Omega(P)}.$
\end{lemma}

Note that if both endpoints of an edge $e$ are in $T$, then the sampling process above from $e$ creates $\rho$ copies of $e$ and each of them has resistance $\rho \rr_e$. The sum of the copies is equivalent to the original edge $e$ itself with its original resistance $\rr_e$. So in this case, we may simply add $e$ to $H$ instead of the $\rho$ sampled copies of it. 

For a walk $W$, we say that its \emph{resistive length} is $\sum_{e \in W} \rr_e$.

Furthermore, if a walk is short, in that it does not visit more than
$\ell$ distinct vertices, it can be sampled, along with its resistance
length, to accuracy $1 \pm \eps$ in $\poly(\ell\eps^{-1})$ time.
Augmenting $C$ with random vertices then in turn gives an effective
way for keeping all the walks short.
Specifically, adding each vertex independently with probability $\beta$
to $C$ ensures that a particular path meets $C$ in $\O(\beta^{-1})$ steps,
and in~\cite{DGGP19} this was shown to also hold when the addition
to $C$ is performed first.
These short works can in turn be sampled locally.

\begin{lemma}[Efficient random walks, \cite{DGGP19} Lemma 4.9]
\label{lemma:effwalks}
There is an algorithm that
given a graph $G$ with all resistances in $[U_1, U_2]$,
a starting vertex $u$, a length $\ell$, and error $\epsilon > 0$,
samples a random walk from $u$ until it reaches $\ell$ distinct vertices.
Additionally, for any vertex $v$ among these $\ell$ distinct vertices, the algorithm returns a multiplicative $\exp(\eps)$ approximation to the resistive length of the walk from $u$ to first time it hits $v$. The total runtime is $\O(\ell^4\eps^{-2}\log^{O(1)}(U_2/U_1))$.
\end{lemma}

The running time and accuracy of our data structures depend on the expected number of these random walks that pass through each vertex. We refer to this as the \emph{congestion} that the random walks incur on vertices.
In \cref{subsec:proofofreducecongestion} we show that $C$ can be efficiently
augmented with the highly congested vertices to reduce the
congestion on vertices in $V\setminus C$. At a high level, we take the $O(\beta m)$ vertices with largest expected congestion and add them to $C$. By a Chernoff bound we can estimate congestion to a factor of two accuracy by random sampling.

\begin{lemma}
\label{lemma:reducecongestion}
There is a routine $\textsc{CongestionReductionSubset}$ that
given an undirected weighted graph $G$ with degrees $\deg$,
and a parameter $\beta$,
$\Chat = \textsc{CongestionReductionSubset}(G, \beta)$
returns in $\O(m\beta^{-2})$ time $\Chat$ with
$|\Chat| \leq O(\beta m)$ such that for all vertices $v\in V\setminus \Chat$ we have
\[
\sum_{u\in V\setminus \Chat, u\neq v} \deg_u \cdot
p^{\Chat \cup \left\{v \right\}}_{v} \left( u \right)
\leq
\O\left( \beta^{-2} \right).
\]
\end{lemma}

A similar congestion reduction procedure is also used
the distributed Schur-complement based Laplacian solver
in~\cite{FGLPSY20:arxiv}).

\subsubsection{Dynamic Expander Decompositions and Spectral Sparsification}

Our algorithm for dynamic spectral sparsification which we need within the dynamic Schur complement is based
on expander decomposition based spectral sparsification~\cite{ST11}.
Specifically, independently sampling an edge $e = (uv)$ contained in an expander with probability $p_e = \O(\deg_u^{-1} + \deg_v^{-1})$ produces spectral
approximations with high probability.
The dynamic maintenance of decompositions of graphs into expanders
was introduced by Nanongkai and Saranurak~\cite{NS17}
and Wulff-Nilsen~\cite{W17} for dynamic connectivity and minimum
spanning tree data structures.
We cite a more recent formulation that is tailored more explicitly
towards graph sparsification.

\begin{lemma}[Dynamic expander decomposition, \cite{BBG20}]
\label{lemma:dynamicexpander}
Given a unweighted graph $G$ and conductance parameter $\phi \le \log^{-4}n$, there is an algorithm which uses $\O(m\phi^{-1})$ preprocessing time, and against an adaptive adversary maintains an expander decomposition of $G$ into subgraphs $G_1, G_2, \dots, G_\ell$ that are $\phi$-expanders under insertions / deletions in amortized $\O(\phi^{-2})$ time, and $\sum_{i=1}^\ell |V(G_i)| = O(n \log^2 n).$
\end{lemma}
By \emph{maintaining} an expander decomposition, we mean that the algorithm outputs which edges move in and out of which expanders. Maintaining an expander decomposition allows us to sample a spectral sparsifier in $\O(n)$ as opposed to $\O(m)$ time when queried.

To support rollbacks with worst case runtime, we will need the following theorem of \cite{SW19}.

\begin{lemma}[Expander pruning, \cite{SW19}]
\label{lemma:expanderpruning}
Let $G=(V,E)$ be a $\phi$-expander with $m$ edges. There is a deterministic
algorithm with access to adjacency lists of $G$ such that, given
an online sequence of $K\le\phi m/10$ edge deletions in $G$, can
maintain a \emph{pruned set} $P\subseteq V$ such that the following
property holds. Let $G_i$ and $P_{i}$ be the graph $G$ and the set $P$ after the $i$-th deletion.
We have, for all $i$, 
\begin{enumerate}
\item $P_{0}=\emptyset$ and $P_{i}\subseteq P_{i+1}$,
\item \label{item:cutsize} $\vol(P_{i})\le8i/\phi$ and $|E(P_{i},V-P_{i})|\le4i$, and
\item $G_i\{V-P_{i}\}$ is a $\phi/6$-expander. 
\end{enumerate}
The total time for updating $P_{0},\dots,P_{K}$ is $O(K\log m/\phi^{2})$. 
\end{lemma}

In the following lemma, we combine \cref{lemma:dynamicexpander} and \cref{lemma:expanderpruning} to get a dynamic expander decomposition algorithm that supports rollbacks. For each update that will not be rolled back (permanent update), we simply forward the update to the dynamic expander decomposition (\cref{lemma:dynamicexpander}). For each update that will be rolled back (temporary update), we prune the updated edge from the expander it is currently in by \cref{lemma:expanderpruning}. Temporary updates are slower than permanent updates but their runtime is worst case bounded. 

\begin{lemma}
\label{lemma:dynamicexpanderwithpruning}
Given a unweighted graph $G$ and conductance parameter $\phi \le \log^{-4}n$, there is an algorithm which uses $\O(m\phi^{-1})$ preprocessing time and against an adaptive adversary maintains an expander decomposition of $G$ into subgraphs $G_1, G_2, \dots, G_t$ that are $\phi/6$-expanders that supports the following two types of updates:
\begin{enumerate}
    \item Permanent update: Insert / delete an arbitrary edge in amortized $\O(\phi^{-2})$ time.
    \item Temporary update: Insert / delete an arbitrary edge. $K$ consecutive temporary updates cost worst case $\O(K^2+K/\phi)$ time.
    \item Rollback: Rollback the last temporary update.
\end{enumerate} Temporary updates must be rolled back before the next permanent update. The expander decomposition always satisfies $\sum_{i=1}^t |V(G_i)| = \O(n \log^2 n+K^2+K/\phi)$ where $K$ is the number of updates of temporary updates that are not yet rolled back.  
\end{lemma}
\begin{proof}
For the permanent updates, we may use \cref{lemma:dynamicexpander} directly because the temporary updates are guaranteed to be rolled back before the next permanent update. For a sequence of temporary updates to edge $e_1, \ldots, e_K$, we delete the edges $e_i$ in order from their respective expanders by \cref{lemma:expanderpruning}. Let $P_K$ be the union of the pruned sets of all expanders after deleting $e_1,\ldots, e_K$. By \cref{item:cutsize} of \cref{lemma:expanderpruning} and summing over the expanders, at most $O(K^2)$ edges are incident to the pruned vertices $P_K$. Each of these edges can be viewed as an expander by itself. 

In an expander $G_i$ with $m_i$ edges, \cref{lemma:expanderpruning} can be used to delete at most $\phi m
_i/10$ edges. To delete $K_i>\phi m_i/10$ edges from this expander, we delete all remaining edges in it in $O(m_i)$ time and remove this empty expander. (And view the deleted edges as expanders by themselves.) In this case, deleting $K_i$ edges cost at most $m_i=O(K_i/\phi)$ time. 

Combining the two cases above, $K$ consecutive temporary updates costs $\O(K^2+K/\phi)$ time and $G$ is decomposed into expanders with total size $\O(n \log^2 n+K^2+K/\phi)$.
\end{proof}
The dynamic expander decomposition in \cref{lemma:dynamicexpanderwithpruning} implies the following dynamic spectral edge sparsifier:
\begin{lemma}[Dynamic edge sparsifier]
\label{lemma:dynamicsparsifier}
There is an algorithm which given a graph $G$ with polynomially bounded resistances, preprocesses $G$ in $\O(m)$, and against an adaptive adversaries supports the following operations:
\begin{enumerate}
    \item Permanent update: Insert / delete an arbitrary edge in amortized $\O(1)$ time.
    \item Temporary update: Insert / delete an arbitrary edge. $K$ consecutive temporary updates cost worst case $\O(K^2)$ time in total.
    \item Rollback: Rollback the last temporary update.
    \item Output: Output an $(1+\eps)$-spectral sparsifier of $G$ in time $\O((n+K^2)\eps^{-2})$ at any time where $K$ is the number of updates of temporary updates that are not yet rolled back. 
\end{enumerate}
Temporary updates must be rolled back before the next permanent update.
\end{lemma}
\begin{proof}
We maintain $O(\log m)$ instances of the dynamic expander decomposition
given in~\cref{lemma:dynamicexpanderwithpruning} with $\phi = \Theta\left(\log^{-4} n\right)$
to maintain a dynamic expander decomposition for edges
with resistances in $[2^i, 2^{i+1}]$
for $-O(\log m) \le i \le O(\log m).$ Temporary and permanent updates are forwarded respectively.
When a $(1+\eps)$-spectral sparsifier query is received,
do the following:
on each expander and each vertex $v$ in it,
we uniformly sample the neighboring edges of $v$
so that we keep $\O(\eps^{-2})$ edges in expectation.
\end{proof}

\subsection{Dynamic Schur Complement}
\label{subsec:dynamicsc}

In this section we prove the following result about maintaining a dynamic Schur complement under terminal additions, deletions, and resistance changes, with an initial set of \emph{safe terminals} that is our choice.

\begin{theorem}[Dynamic Schur complement]
\label{thm:dynamicsc}
There is a \DynamicSC~data structure that against an oblivious adversary supports the following operations for any parameters $0 < \beta, \eps < 1$,
on a graph $G$ with dynamic terminal set $C$.
\begin{itemize}
    \item $\textsc{Initialize}(G, \rr^\init \in \R^{E(G)}_{>0}, C^\init, \eps, \beta).$
    Initializes a graph $G$ with resistances $\rr^\init$ and a set of safe terminals $C^\safe$. Sets the terminal set $C \assign C^\safe \cup C^\init.$ Runtime: $\O(m\beta^{-4}\eps^{-4})$.
    
    \item $\textsc{PermanentAddTerminals}(\Delta C \subseteq V(G)).$ Adds all vertices in the set $\Delta C$ as terminals. Runtime: $\O(|\Delta C|\beta^{-2}\eps^{-2})$ amortized.
    
    \item $\textsc{TemporaryAddTerminals}(\Delta C \subseteq V(G)).$ Adds all vertices in the set $\Delta C$ as (temporary) terminals. Runtime: Worst case $\O\left(\left(K\beta^{-2}\eps^{-2}\right)^2\right)$ if the \textsc{TemporaryAddTerminals} operations that are not rolled back add $K$ terminals in total. All \textsc{TemporaryAddTerminals} operations should be rolled back before the next \textsc{PermanentAddTerminals}.
    
    \item $\textsc{Update}(Z, \rr^\new \in \R^Z_{>0}).$ Under the guarantee that for all edges $e \in Z$, both endpoints of $e$ are terminals, updates $\rr_e \assign \rr_e^\new$ for all $e \in Z$. Runtime: Worst case $\O(|Z|)$.
    
    \item $\textsc{SC}().$ Returns a $(1+\eps)$-approximation of $\mSC(G, C)$ for the current terminal set $C$ and resistances. Runtime: Worst case $\O\left(\left(\beta m+\left(K\beta^{-2}\eps^{-2}\right)^2\right)\eps^{-2}\right)$ where $K$ is the number of $\textsc{TemporaryAddTerminals}$ operations that are not rolled back.
    
    \item $\textsc{Rollback}().$ Rollback the last \textsc{Update}, \textsc{PermanentAddTerminals} or \textsc{TemporaryAddTerminals} if exists. Costs the same time as the original operation.
\end{itemize}
The data structure succeeds with high probability. 
Furthermore for the same initial graph, and same updated state of terminals, the Schur complement generated has the same distribution.
\end{theorem}

Now we give the algorithm and pseudocode for \cref{thm:dynamicsc}
in~\cref{algo:dynamicsc,algo:dynamicsc2}.

\begin{algorithm}[!ht]
\caption{\DynamicSC. \label{algo:dynamicsc}}
\SetKwProg{Globals}{global variables}{}{}
\SetKwProg{Proc}{procedure}{}{}
\SetKwProg{Guarantee}{guarantee:}{}{}
\Globals{}{
    $\eps, \beta$: approximation quality and size of terminals. \\
    $C, C^\init, C^\safe$: current terminal set, initial terminal set, and safe terminals. \\
    $\rr, \rr^\init$: current resistances and original resistances. \\
    $P_e \assign \emptyset$ for $e \in E(G)$: set of paths for $e \in E(G).$ \\
    $D^\expand$ -- instances of dynamic expander decomposition for dynamic spectral sparsification for $\phi = \Theta(\log^{-4} n)$, as in \cref{lemma:dynamicexpanderwithpruning} and \cref{lemma:dynamicsparsifier}.
}
\Proc{$\textsc{Initialize}(G, \rr^\init \in \R^{E(G)}_{>0}, C^\init, \eps, \beta)$}{
    $C^\mathrm{(congest)} \assign \textsc{CongestionReductionSubset}(G, \beta)$. \tcp{\cref{lemma:reducecongestion}}
    Let $C^\mathrm{(sample)}$ be a random subset of $\beta m$ vertices. \label{line:csample} \\
    $C^\safe \assign C^\mathrm{(congest)} \cup C^\mathrm{(sample)}$ \label{line:makecsafe} \\
    For each edge $e \in E(G)$, sample $\rho$ random walks from $e$ and approximate resistive lengths up to $\exp(\eps/10)$ as in \cref{lemma:effwalks} with $\ell=\Omega(\beta^{-1}\log m)$. The algorithm fails if any random walk does not hit $C$. Let these walks be $P_e$.
    \label{line:realwalks} \\
    $C \assign C^\safe \cup C^\init$, and shortcut walks $w \in \cup_e P_e$ to where they hit $C$. \\
    Pass the updated edge and resistance on $C$ to $D^\expand$, using \cref{lemma:scapprox} to decide the resistances based on the total resistive length of the walk. \label{line:makefirstsc} \\
    \Return $C^\safe$. \label{line:returnsafe} \tcp{Returns safe terminals.}
}
\Proc{$\textsc{TemporaryAddTerminals}(\Delta C \subseteq V(G)).$}{
    $C \assign C \cup \Delta C.$ \\
    Shortcut walks $w \in \cup_e P_e$ to where they hit $C$. \label{line:shortcutadd2} \\
    Pass necessary edge and resistance updates from shortcutting walks to $D^\expand$ as \textbf{temporary} updates, using \cref{lemma:scapprox} to decide the resistances. \label{line:addupdatesc2} \tcp{We have $\exp(\eps/10)$ approximate resistive lengths, as we used \cref{lemma:effwalks} in line \ref{line:realwalks}.}
}
\Proc{$\textsc{PermanentAddTerminals}(\Delta C \subseteq V(G)).$}{
    \Guarantee{All \textsc{TemporaryAddTerminals} are already rolled back.}{}
    $C \assign C \cup \Delta C.$ \\
    Shortcut walks $w \in \cup_e P_e$ to where they hit $C$. \label{line:shortcutadd1} \\
    Pass necessary edge and resistance updates from shortcutting walks to $D^\expand$ as \textbf{permanent} updates, using \cref{lemma:scapprox} to decide the resistances. \label{line:addupdatesc1} \tcp{We have $\exp(\eps/10)$ approximate resistive lengths, as we used \cref{lemma:effwalks} in line \ref{line:realwalks}.}
}
\end{algorithm}

\begin{algorithm}[!ht]
\caption{\DynamicSC. \label{algo:dynamicsc2}}
\SetKwProg{Globals}{global variables}{}{}
\SetKwProg{Proc}{procedure}{}{}
\SetKwProg{Guarantee}{guarantee:}{}{}
\Globals{}{
    $\eps, \beta$: approximation quality and size of terminals. \\
    $C, C^\init, C^\safe$: current terminal set, initial terminal set, and safe terminals. \\
    $\rr, \rr^\init$: current resistances and original resistances. \\
    $P_e \assign \emptyset$ for $e \in E(G)$: set of paths for $e \in E(G).$ \\
    $D^\expand$ -- instances of dynamic expander decomposition for dynamic spectral sparsification for $\phi = \Theta(\log^{-4} n)$, as in \cref{lemma:dynamicexpanderwithpruning} and \cref{lemma:dynamicsparsifier}.
}
\Proc{$\textsc{Update}(Z, \rr^\new \in \R^Z_{>0}).$}{
    \Guarantee{All endpoints of edges in $Z$ are terminals.}{}
    \For{$e \in Z$}{$\rr_e \assign \rr^\new_e$, and pass changes to $D^\expand$. \label{line:updateupdatesc}}
}
\Proc{$\textsc{SC}().$}{
    Sample $\widetilde{\mSC}$ using \cref{lemma:dynamicsparsifier} and the expander decomposition maintained by $D^\expand.$ \label{line:samplesc} \\
    \Return $\widetilde{\mSC}.$
}
\Proc{$\textsc{Rollback}().$}{
    Rollback the last operation by undo the changes made by the operation in reverse order.
}
\end{algorithm}

At a high level, our algorithm maintains a set of random walks onto a randomly chosen set of terminals of size $\beta m$. Additionally, we add some other terminals to reduce the congestion of the random walks to ensure that our runtimes for shortcutting walks (not including passing the updates to the dynamic edge sparsifier $D^\expand$) in \textsc{PermanentAddTerminals} and \textsc{TemporaryUpdates} are worst case, instead of amortized. \textsc{PermanentAddTerminals} (resp. \textsc{TemporaryAddTerminals}) passes the updates caused by shortcutting walks to $D^\expand$ (\cref{lemma:dynamicsparsifier}) as permanent (resp. temporary) updates. Because only temporary updates has worst case time bound in $D^\expand$, only the runtime of \textsc{TemporaryAddTerminals} (not \textsc{PermanentAddTerminals}) is worst case.

To maintain an approximate Schur complement, we use \cref{lemma:scapprox} and the random walks we have chosen. Additionally, we obtain $\exp(\eps/10)$ approximate resistive lengths for the walks using \cref{lemma:effwalks}, and this only affects the approximation quality by $\exp(\eps/10)$ as well. As $C^\mathrm{(sample)}$ contains $\beta m$ random vertices and each random walk we chose visits $\Omega(\beta^{-1} \log m)$ distinct vertices, all of the random walks hit $C$ with high probability. Under terminal additions, we can shorten these walks using a binary tree (which should support updates with worst case runtime), and plug the resulting resistive lengths from \cref{lemma:scapprox} into the dynamic edge sparsifier in \cref{lemma:dynamicsparsifier}.

We now analyze Algorithm \DynamicSC~(\cref{algo:dynamicsc,algo:dynamicsc2}) to prove \cref{thm:dynamicsc}.

\begin{proof}[Proof of \cref{thm:dynamicsc}]
We first show correctness, then analyze the runtime.
\paragraph{Correctness.} We first check that $C^\safe$ indeed has $O(\beta m)$ vertices. Indeed, $C^\mathrm{(sample)}$ has at most $\beta m$ vertices and $C^\mathrm{(congest)}$ has $O(\beta m)$ vertices by \cref{lemma:reducecongestion}.
Hence $|C^\safe| \le O(\beta m).$ Additionally, because the walks have at least $\Omega(\beta^{-1} \log m)$ vertices, all walks hit $C^\mathrm{(sample)}$, a random subset of size $\beta m$, with high probability.

Next, we show $\widetilde{\mSC} \approx_\eps \mSC(G, C)$ in line \ref{line:samplesc} of \cref{algo:dynamicsc2} with high probability, so that the operation $\textsc{SC}()$ is correct.
By \cref{lemma:scapprox}, the underlying Schur complement from shortcutting paths (and using the true resistive lengths) is a $(1+\eps/3)$-approximation to $\mSC(G, C)$. Also, because the resistive lengths were sampled using \cref{lemma:effwalks} in line \ref{line:realwalks} of \DynamicSC.\textsc{Initialize} (\cref{algo:dynamicsc}), the resistive lengths are accurate up to $\exp(\eps/10)$. Now, the dynamic spectral sparsifier in \cref{lemma:dynamicsparsifier} indeed would return a $(1+\eps/3)$-approximation of the underlying approximate Schur complement, so in total it is a $(1+\eps)$-approximation to $\mSC(G, C)$.

\paragraph{Runtime.} We first observe that every vertex in $C \bs C^\mathrm{(congest)}$ is involved in at most $\O(\beta^{-2}\eps^{-2})$ walks built in line \ref{line:realwalks} of \cref{algo:dynamicsc}. This is because the actual number of walks through each vertex differs from the expected value by a factor up to $\O(1)$ with high probability due to a Chernoff bound and \cref{lemma:reducecongestion}.

We go item by item.
\begin{itemize}
    \item $\textsc{Initialize}.$ Sampling the paths and feeding the resistances into the dynamic edge sparsifier as permanent updates requires time
    \[ \O(m \cdot \eps^{-2} \cdot \beta^{-4}\eps^{-2}+m\cdot \eps^{-2}) = \O(m\beta^{-4}\eps^{-4}) \]
    by \cref{lemma:effwalks} and \cref{lemma:dynamicsparsifier}.
    \item $\textsc{PermanentAddTerminals}.$ Each vertex in $\Delta C$ is involved in $\O(\beta^{-2}\eps^{-2})$ paths by the above discussion. Each path is processed as a permanent update in \cref{lemma:dynamicsparsifier}. The total time is $\O(|\Delta C|\beta^{-2}\eps^{-2}).$
    \item $\textsc{TemporaryAddTerminals}.$ Each vertex in $\Delta C$ is involved in $\O(\beta^{-2}\eps^{-2})$ paths. Each path is processed as a temporary update in \cref{lemma:dynamicsparsifier}. The total time for adding $K$ temporary terminals is $\O\left(\left(K\beta^{-2}\eps^{-2}\right)^2\right)$ by \cref{lemma:dynamicsparsifier}.
    \item $\textsc{Update}$. Requires time $\O(|Z|)$ by \cref{lemma:dynamicsparsifier}, as we simply are changing $|Z|$ edge resistances.
    \item $\textsc{SC}.$ Requires time $\O\left(\left(\beta m+\left(K\beta^{-2}\eps^{-2}\right)^2\right)\eps^{-2}\right)$ by \cref{lemma:dynamicsparsifier}, as there are $O(\beta m)$ terminals and $\left(K\beta^{-2}\eps^{-2}\right)^2$ temporary updates to the dynamic sparsifier.
\end{itemize}
\end{proof}

\subsection{Algorithm and Proof for \textsc{Checker}}
\label{subsec:checkersub}
In this section we prove \cref{thm:checker} about the \textsc{Checker}~data structure.
Its pseudocode is in Algorithm~\ref{algo:checker}.
\begin{algorithm}[!ht]
\caption{\Checker: algorithm for approximating flows on edges $e$. \label{algo:checker}}
\SetKwProg{Globals}{global variables}{}{}
\SetKwProg{Proc}{procedure}{}{}
\Globals{}{
    $\beta, \eps$: parameters for size of terminal set, and approximate quality. \\
    $\rr^\init, \rr$: initial resistances and current resistances. \\
    $C^\safe, C$: safe terminals and current terminals. \\
    $D^\sc$: instance of \DynamicSC~from \cref{thm:dynamicsc}. \\
}
\Proc{$\textsc{Initialize}(G, \rr^\init \in \R^{E(G)}_{>0}, \eps, \beta).$}{
    $C^\safe \assign D^\sc.\textsc{Initialize}(G, \rr^\init \in \R^{E(G)}_{>0}, \{s, t\}, \eps/10, \beta)$.
}
\Proc{$\textsc{PermanentUpdate}(e=uv, \rr^\new_e \in \R_{>0}).$}{
    $\Delta C \assign \{u, v\}.$ \label{line:updateaddterminals1} \tcp{Vertices adjacent to edge $e$}
    $D^\sc.\textsc{PermanentAddTerminals}(\Delta C), C \assign C \cup \Delta C$. \label{line:updateaddterminalsperma} \\
    $D^\sc.\textsc{Update}(e, \rr^\new).$ \\
}
\Proc{$\textsc{TemporaryUpdate}(e=uv, \rr^\new_e \in \R_{>0}).$}{
    $\Delta C \assign \{u, v\}.$ \label{line:updateaddterminals2} \tcp{Vertices adjacent to edge $e$}
    $D^\sc.\textsc{TemporaryAddTerminals}(\Delta C), C \assign C \cup \Delta C$. \label{line:updateaddterminalstemp} \\
    $D^\sc.\textsc{Update}(e, \rr^\new).$ \\
}
\Proc{$\textsc{Rollback}().$}{
    Rollback the last operation by undo the changes made by the operation in reverse order.
}
\Proc{$\textsc{Check}(e).$}{
    $\Delta C \assign \{u, v\}.$ \tcp{Endpoints of $e$.}
    $D^\sc.\textsc{PermanentAddTerminals}(\Delta C)$. \label{line:checkaddterminals} \\
    $\widetilde{\mSC} \assign D^\sc.\textsc{SC}().$ \label{line:samplescchecker} \\
    $\xx \assign \widetilde{\mSC}^\dagger \cchi_{st}.$ \label{line:solvex} \\
    $\gg_e \assign \rr_e^{-1}(\xx_u - \xx_v).$ \\
    \If{$\rr_e\gg_e^2 \ge 3\eps^2\cE/4$ \label{line:acceptcheck}}{
        Accept $e$ and return $\gg_e$.
    }
    $D^\sc.\textsc{Rollback}().$ \tcp{Rollback the last $\textsc{PermanentAddTerminals}$} \label{line:deletecheck}
}
\end{algorithm}
Our algorithm for \Checker~simply maintains a dynamic Schur complement using \cref{thm:dynamicsc}. This works because an approximate Schur complement is sufficient to provide additive approximations for the flow on edges with at least $\eps^2$ fraction of the electric energy.
\begin{lemma}
\label{lemma:approxenergy}
Consider a graph $G$ with resistances $\rr_e$, and vertices $s, t$ such that the $s$-$t$ electric flow has energy $\cE$. For an operator $\widetilde{\mL} \approx_\eps \mL(G)$, we have for any edge $e$ that
\[ \rr_e^{-1/2}\left|\cchi_e^\top\left(\mL(G)^\dagger - \widetilde{\mL}^\dagger\right)\cchi_{st}\right| \le \eps\sqrt{\cE}. \]
\end{lemma}
\begin{proof}
It is clear that
\[ \left\|\mL(G)^{1/2}\left(\mL(G)^\dagger - \widetilde{\mL}^\dagger\right)\mL(G)^{1/2}\right\|_2 = \left\|\mI - \mL(G)^{1/2}\widetilde{\mL}^\dagger\mL(G)^{1/2}\right\|_2 \le \eps. \]
Therefore, the Cauchy-Schwarz inequality gives us
\begin{align*}
\rr_e^{-1/2}\left|\cchi_e^\top\left(\mL(G)^\dagger - \widetilde{\mL}^\dagger\right)\cchi_{st}\right| &\le \rr_e^{-1/2}\eps\left(\cchi_e^\top \mL(G)^\dagger \cchi_e\right)^{1/2}\left(\cchi_{st}^\top \mL(G)^\dagger \cchi_{st}\right)^{1/2} \\
&\le \rr_e^{-1/2}\eps \cdot \rr_e^{1/2} \cdot \sqrt{\cE}\\
&\le \eps\sqrt{\cE}.
\end{align*}
Here, we used that the resistance of edge $e$ is $\rr_e$, and $\cchi_{st}^\top \mL(G)^\dagger \cchi_{st} = \cE.$
\end{proof}

More precisely, our algorithm \Checker~works as follows. We start by initializing an instance $D^\sc$ of \DynamicSC~with
$C^\init = \{s, t\}.$
To implement the two types (permanent and temporary) of \textsc{Update} operations, we add the endpoints of the edges as terminals and then update their resistances by shortening random walks. We use $\textsc{PermanentAddTerminals}$ and $\textsc{TemporaryAddTerminals}$ for $\textsc{PermanentUpdate}$ and $\textsc{TemporaryUpdate}$ respectively.
To \textsc{Rollback} an operation, we simply undo the changes made by the operation in reverse order.

Finally, to implement \textsc{Check}, for the edge $e=(u, v)$, we first add its endpoints with the operation $D^\sc.\textsc{TemporaryAddTerminals}(\{u, v\})$ and sample an approximate Schur complement using $D^\sc.\textsc{SC}().$ We then use this approximate solver to estimate the flow across $e$, and return $e$ if the energy estimate is at least $3\eps^2\cE/4.$ After that we eliminate the effects of this \textsc{Check} by rolling back $D^\sc.\textsc{TemporaryAddTerminals}(\{u, v\})$.
We now formally show that~\cref{algo:checker} meets the requirement of the checker data structure as described in~\cref{thm:checker}.

\begin{proof}[Proof of \cref{thm:checker}]
We first check the correctness, then verify the runtime bounds.
\paragraph{Correctness.}
We first verify the correctness of \textsc{Check}, assuming that the calls to $D^\sc$ in \cref{algo:checker} satisfy the guarantees of \cref{thm:dynamicsc}.

By the guarantees of \cref{thm:dynamicsc} we know that for the terminal set $C$, we have $\widetilde{\mSC} \approx_{\eps/10} \mSC(G, C).$ By \cref{lemma:approxenergy} for $\mL(G) = \mSC(G, C)$ and $\widetilde{\mL} = \widetilde{\mSC}$ we have for an edge $e$ that
\begin{align} \rr_e^{1/2}|\Delta\ff_e - \bg_e| = \rr_e^{-1/2}\left|\cchi_e^\top(\mSC(G, C)^\dagger - \widetilde{\mSC}^\dagger)\cchi_{st}\right| \le \eps\sqrt{\cE}/10. \label{eq:errorbound} \end{align}
By this, if $\rr_e\Delta\ff_e^2 \ge \eps^2\cE$, then we know that
\[ \rr_e^{1/2}|\bg_e| \ge \rr_e^{1/2}|\Delta\ff_e| - \eps\sqrt{\cE}/10 \ge 9\eps\sqrt{\cE}/10. \]
Therefore, $\rr_e\bg_e^2 \ge 3\eps^2\cE/4$, so line \ref{line:acceptcheck} of \cref{algo:checker} triggers and edge $e$ is returned with the proper error bound by \eqref{eq:errorbound}.

On the other hand, if $\rr_e\Delta\ff_e^2 \ge \eps^2\cE/2$ we know that
\[ \rr_e^{1/2}|\bg_e| \le \rr_e^{1/2}|\Delta\ff_e| + \eps\sqrt{\cE}/10 \ge \eps\sqrt{\cE}/\sqrt{2} + \eps\sqrt{\cE}/10 \le 5\eps\sqrt{\cE}/6. \]
Therefore, we know that $\rr_e\bg_e^2 < 3\eps^2\cE/4,$ so line \ref{line:acceptcheck} of \cref{algo:checker} is not triggered and $e$ is not returned.

Since the operation $\textsc{TemporaryAddTerminals}$ is rolled back after checking, $\textsc{Check}$ does not modify any data structure or variable we use. Thus, we have the property that the output of $\textsc{Check}(e)$ is independent of any previous calls to $\textsc{Check}$.

Now we verify that the calls to $D^\sc$ in \cref{algo:checker} satisfy the guarantees of \cref{thm:dynamicsc}. We analyze all necessary guarantees in the calls to \cref{thm:dynamicsc}.
\begin{itemize}
    \item $\textsc{Update}:$ we must ensure that both endpoints of $e$ are terminals. This is true because of line \ref{line:updateaddterminals1}, \ref{line:updateaddterminals2} in \cref{algo:checker}.
    \item The size of the terminal set $C$ is $O(\beta m)$ at all times because of the guarantee in \cref{thm:checker} that the total number of permanent and temporary add terminal operations is at most $\beta m$.
\end{itemize}

\paragraph{Runtime.}
We go item by item.
\begin{itemize}
    \item $\textsc{Initialize}.$ Follows from \cref{thm:dynamicsc}, as this procedures makes a single call to $D^\sc.\textsc{Initialize}$.
    \item $\textsc{PermanentUpdate}.$ We know that $|\Delta C| \le 2$ in line \ref{line:updateaddterminals1} of \cref{algo:checker}.  Therefore, the calls to $D^\sc.\textsc{PermanentAddTerminals}$ require time at most $\O(\beta^{-2}\eps^{-2})$ by \cref{thm:dynamicsc}.
    \item $\textsc{TemporaryUpdate}.$ We know that $|\Delta C| \le 2$ in line \ref{line:updateaddterminals1} of \cref{algo:checker}.  Therefore, the calls to $D^\sc.\textsc{PermanentAddTerminals}$ require time at most $\O\left(\left(K\beta^{-2}\eps^{-2}\right)^2\right)$ by \cref{thm:dynamicsc} where $K$ is the number of \textsc{TemporaryUpdate}s that are not rolled back.
    \item $\textsc{Rollback}.$ It costs the same time as the original operation as we undo the changes. 
    \item $\textsc{Check}.$ We know that $|\Delta C| = 2.$ Hence we add at most $2$ new temporary terminals to $D^\expand$. By the guarantee of the \textsc{SC}() operation of \cref{thm:dynamicsc} and a nearly-linear time Laplacian solver (\cref{thm:lap}), sampling the Schur complement and solving in lines \ref{line:samplescchecker} and \ref{line:solvex} requires time $\O\left(\left(\beta m+\left(K\beta^{-2}\eps^{-2}\right)^2\right)\eps^{-2}\right)$. The cost of $\textsc{TemporaryAddTerminals}$ and $\textsc{Rollback}$ are dominated by the cost of sampling the Schur complement by \cref{thm:dynamicsc}.
\end{itemize}
This completes the runtime analysis and the proof.
\end{proof}
\section{Locator for Candidate Edges Against Oblivious Adversaries}
\label{sec:Locator}

In this section we build the \Locator~(\cref{thm:locator} below) which is a heavy hitter data structure for electric flows.
We first reduce the problem of locating high congestion edges to dotting the projection of heavy-hitter vectors onto a set of terminals and the potential vector of the unit electrical flow.
In \cref{subsec:approxproj}, we show how a set of local random walks approximate the change of this projection vector. Then we introduce the \textsc{Projector} data structure maintaining the projection vectors in \cref{subsec:projector}. The \textsc{Projector} is then combined with \textsc{DynamicSC} (\cref{thm:dynamicsc}) for dynamic Schur complements to give the resulting data structure for locating high congestion edges in \cref{subsec:locatorproof}. We note that because \cref{thm:locator} below has no temporary updates, we only require \cref{thm:dynamicsc} for the $K = 0$ situation.

Throughout, we assume that $\cE = 1$ without loss of generality, as we can simply scale resistances to ensure this. Throughout the remainder of this section, we call the dynamic Schur complement data structure $D^{(\mathrm{sc})}$ instead of \textsc{DynamicSC}, and call the projection data structure $D^{(\mathrm{proj})}$ instead of \textsc{Projector}.

\begin{theorem}
\label{thm:locator}
There is a data structure $\textsc{Locator}$ that
for terminal size factor $\beta$,
partial reinitialize threshold $\delta$,
energy threshold $\epsilon$, satisfying $0 < \delta < \beta < \eps < 1$,
and total energy bound $\cE$,
supports the following operations with inputs independent
of the randomness used, i.e. against oblivious adversaries:
\begin{itemize}
\item $\textsc{Initialize}(G, \rr, \eps, \beta, \delta).$ Initializes the data structure given a graph $G$ with resistances $\rr$. Creates an $\ell_2$ heavy hitter sketch, terminal set $C$ with $|C| = O(\beta m)$, and random walks of length $O(\beta^{-1} \log m)$ from each edge $e$ to the terminal set $C$.
Runtime: $\O(m\beta^{-4}\delta^{-2}\eps^{-2})$.
\item $\textsc{Update}(e, \rr^\new)$. Updates $\rr_e \assign \rr^\new$. Runtime: amortized $\O(\delta m\eps^{-3}+\delta^{-2}\beta^{-6}\eps^{-2})$.
\item $\textsc{BatchUpdate}(S, \rr^\new \in \R_{>0}^S)$. For all $e \in S$ updates $\rr_e \assign \rr^\new_e$. Runtime: $\O(m\eps^{-2} + |S|\beta^{-2}\eps^{-2})$.
\item $\textsc{Locate}()$. Returns a set $S \subseteq E(G)$ of size $|S| \le O(\eps^{-2})$ that contains all edges $e$ with energy at least $\eps^2\cE/10$ in an unit $s$-$t$ electric flow. Runtime: $\O(\beta m \eps^{-2})$.
\end{itemize}
Under the guarantees that the total number of edges updated between $\textsc{Update}, \textsc{BatchUpdate}$ is at most $\beta m$, and that the energy of a unit $s$-$t$ flow is at most $\cE$, the algorithm succeeds with high probability and satisfies the given runtimes.
\end{theorem}

Locating the edges with large energy values dynamically requires us to maintain large entries of the vector $\xx \defeq \mR^{1/2}\ff$, where $\mR$ is a diagonal matrix on $\mathbb{R}^{m\times m}_{\ge0}$ such that $\mR_{e, e}=\rr_e$ and $\ff$ is the unit $s$-$t$ electric flow. This can be rewritten
\[
\xx = \mR^{-1/2} \mB \pphi = \mR^{-1/2} \mB (\mB^\top \mR^{-1}\mB)^\dagger \cchi_{st},
\] where $\mB$ is the edge-vertex incidence matrix and $\pphi$ is the potential vector corresponding to voltages.
The following lemma allows us to recover the large entries of $\xx$ by a low-dimensional projection of it, and we prove it in \cref{proofs:heavyhitter}.
\begin{lemma}[$\ell_2$ heavy hitter with errors]
\label{lem:heavyhitter}
There is an algorithm $\textsc{Build}$ that for any
error parameter $0 < \eps < 1 / \log{n}$ and integer $m$,
$\textsc{Build}(\eps, m)$
returns in time $\O(m\ndd)$ a random matrix $\mQ \in \{-1, 0, 1\}^{\ndd \times m}$
with $\ndd = O(\eps^{-2}\log^{3}{m})$ such that every column of
$\mQ$ has $O(\log^{3}{m})$ nonzero entries.

Additionally, there is an algorithm $\textsc{Recover}$ such that for any vector
$\xx \in \R^m$ with $\|\xx\|_2 \le 1$
and $\vv \in \R^\ndd$ satisfying
\[
\norm{\vv - \mQ \xx}_\infty
\le
\eps/10,
\]
$\textsc{Recover}(\vv)$ returns in time $O(\eps^{-2}\log^3 {m} )$
a set $S \subseteq [m]$ with size at most $O(\eps^{-2})$
that with high probability contains all indices $i$ with $|x_i| \geq \eps$.
\end{lemma}

We will apply this heavy hitter lemma to detect large flow edges. To ensure extra stability, we show a simple lemma to argue that edges with large energies must have resistances close to $1$, assuming that $\cE = 1.$
\begin{lemma}
\label{lemma:badres}
Let $G$ be a graph with resistances $\rr \in \R_{>0}^E$
and vertices $s$ and $t$.
Let $\ff$ be the $\cchi_{st}$ electric flow,
and say it has energy $\cE$.
Then the energy of any edge $e$ in this flow $\ff$ w.r.t. $\rr$ satisfies
\[
\rr_e \ff_e^2
\le
\min\left(\rr_e,
\rr_e^{-1}\cE^2\right).
\]
\end{lemma}
\begin{proof}
We show the two bounds separately.
Because $\ff$ routes a total of $1$ unit from $s$ to $t$
and has no cycles, we have $|\ff_e| \leq 1$ on every edge.
Multiplying both sides by $\rr_e$ gives
$\rr_e\ff_e^2 \le \rr_e$, which is the first bound.
For the second bound, let
\[
\pphi = \mL\left(\rr\right)^\dagger\cchi_{st}
\]
be the induced potentials of the electric flow $\ff$.
By Ohm's Law for the edge $e = (u, v)$ we have
\begin{align*}
    \rr_e\ff_e^2 &= \rr_e^{-1}(\pphi_u - \pphi_v)^2 \\
    &\le \rr_e^{-1}(\pphi_s - \pphi_t)^2 = \rr_e^{-1}(\cchi_{st}^\top \mL(\ff)^\dagger \cchi_{st})^2 \\
    &= \rr_e^{-1}\cE^2.
\end{align*}
Here, $|\pphi_s - \pphi_t| \geq |\pphi_u - \pphi_v|$ follows from
the fact that potentials $\pphi$ increase from $t$ to $s$,
and the last equality follows from the fact that
$\cchi_{st}^\top \mL(\ff)^\dagger \cchi_{st}$
is the electric energy.
\end{proof}
Clearly, \cref{lemma:badres} and the assumption of $\cE=1$ allow us to restrict our attention to the set of edges with resistances $\rr_e \in [\eps^2/20, 20\eps^{-2}]$.
We denote this subset of edges as $S$ throughout. 
So the problem becomes, up to a factor of $\O(\epsilon^{-2})$, maintaining
\[
\mQ\xx=\mQ\mI_S\mR^{-1/2} \mB \pphi
\]
for some $\mQ\in \{-1,0,1\}^{\ndd\times m}$ where $\ndd=\O(\eps^{-2})$, and $\mI_S$ is the identity matrix restricted to $S$.

Note that because $\mB \oone = 0$,
\[
\mQ^{\tomato} \mI_S \mR^{-1/2} \mB \oone = \zzero.
\]
So we will treat the rows of $\mQ\mI_S\mR^{-1/2} \mB$ as a demand vectors, and use $\dd[i]$ to denote the $i$-th row.
\begin{definition}
\label{def:settings_of_sec_approx_proj}
Let $\mQ$ be a matrix in $\{-1, 0, 1\}^{\ndd\times m}$ for some $\ndd=\O(\eps^{-2})$ and each column of $\mQ$ contains at most $\O(1)$ nonzero entries.

Also, given a vector $\rr$ of resistances and an accuracy parameter $\eps$, let $S \subseteq E$ be the subset of edges $e$ with $\rr_e \in [\eps^2/20, 20\eps^{-2}].$
Finally, we let $\dd[i]^\top$ be the rows of $\mQ\mI_S\mR^{-1/2}\mB$, so
\[ \mQ\mI_S\mR^{-1/2}\mB = \begin{bmatrix} \dd[1]^\top \\ \vdots \\ \dd[N]^\top \end{bmatrix}. \]
\end{definition}
For simplicity, we start by considering a single such demand vector $\dd$. In doing so, we will drop the indexing $[i]$ for the time being. This way, our goal is to maintain $\dd^\top \pphi$ where $\dd$ is a demand vector given by $\dd^\top = \qq^\top \mI_S \mR^{-1/2}\mB$ for some vector $\qq^\top$ which is a row of the sketch matrix $\mQ$.

We will transfer the inner product $\dd^\top\pphi$ to a smaller terminal set $C$ by using the Cholesky factorization.

\begin{lemma}[Cholesky factorization]
\label{lemma:cholesky}
Given a matrix $\mL \in \R^{n \times n}$ and subset $C \subseteq [n]$, let $F = [n] \bs C.$ We have that
\[ \mL^\dagger = \begin{bmatrix} \mI & -\mL_{FF}^{-1} \mL_{FC} \\ 0 & \mI \end{bmatrix}
\begin{bmatrix} \mL_{FF}^{-1} & 0 \\ 0 & \mSC(\mL, C)^\dagger \end{bmatrix}
\begin{bmatrix} \mI & 0 \\ -\mL_{CF}\mL_{FF}^{-1} & \mI \end{bmatrix}. \]
\end{lemma}
Assuming that $s, t$ are in $C$ we get
\begin{align*}
\dd^\top \pphi &= \dd^\top \mL^\dagger \cchi_{st} \\ &= \dd^\top \begin{bmatrix} \mI & -\mL_{FF}^{-1} \mL_{FC} \\ 0 & \mI \end{bmatrix}
\begin{bmatrix} \mL_{FF}^{-1} & 0 \\ 0 & \mSC(\mL, C)^\dagger \end{bmatrix}
\begin{bmatrix} \mI & 0 \\ -\mL_{CF}\mL_{FF}^{-1} & \mI \end{bmatrix} \cchi_{st} \\ &= \left(\begin{bmatrix} \mI & -\mL_{CF}\mL_{FF}^{-1} \end{bmatrix}\dd\right)^\top \mSC(G, C)^\dagger\cchi_{st}.
\end{align*}
This product has two main parts. One part is $\mSC(G, C)^\dagger\cchi_{st} = \pphi_C,$ i.e. the restriction of the electric potentials to the Schur complement $C$. This can be maintained using dynamic Schur complements (\cref{thm:dynamicsc}).
We call the other part $\begin{bmatrix} \mI & -\mL_{CF}\mL_{FF}^{-1} \end{bmatrix}\dd = \dd_C - \mL_{CF}\mL_{FF}^{-1}\dd_F$ the \emph{projection} via random walks of $\dd$ onto $C$.
\begin{definition}
\label{def:proj}
In a graph $G$ with Laplacian $\mL$,
the projection of a vector on all the vertices, $\dd \in \R^{V}$,
onto terminal set $C$ is
\[
\ppi^{C}\left(\dd\right)
\defeq
\left[
\begin{array}{cc}
-\mL_{FF}^{-1} \mL_{FC}\\
\mI
\end{array}
\right]^{\tomato}
\dd
=
\dd_{C} - \mL_{CF} \mL_{FF}^{-1} \dd_{F}
\]
\end{definition}

Intuitively, $\ppi^{C}(\dd)$ moves the demands $\dd$ to vertices in $C$
by random walks on $G$.
It is in fact a formal vector generalization of the hitting probabilities
defined above in Definition~\ref{def:hit}, in that we have
$p^{C}_{v}(u) = \ppi^{C}(\cchi_{u})_v$ where $\cchi_u$ is the indicator
vector on $u$.
This relation also holds more generally due to the linearity of the
operator: we prove the following folklore result in
\cref{proofs:projrandwalk}.

\begin{lemma}
\label{lemma:projrandwalk}
In any weighted undirected graph $G = (V, E)$, vector $d \in \R^{V(G)}$,
and subset of vertices $C$ and any $v\in C$,
we have that $\ppi^{C}(\dd)_v$
is equal to the hitting probability mass of $\dd$ onto $v$,
with $C$ as the sink set:
\[
\ppi^{C}\left( \dd \right)_{v}
=
\sum_{u\in V} \dd_u \cdot p^{C}_{v} \left( u \right).
\]
\end{lemma}

We now discuss how to leverage this combinatorial interpretation to maintain $\ppi^C(\dd)$. Because $\ppi^C(\dd)$ is defined as moving the demands $\dd$ onto $C$, a natural approach is to estimate the quantity using a set of random walks from each vertex or edge. However, the vector $\dd$ has all entries of size at least $\pm 1$ on average, and $\pphi_C$ can potentially have all entries of size $\pm 1$ also. Hence the variance of estimating $\ppi^C(\dd)^\top \pphi_C$ by sampling random walks from each vertex can potentially be $\Omega(m)$, which is much too large. Instead, we initialize by exactly computing $\ppi^C(\dd)$ by solving a Laplacian system in $\mL_{FF}^{-1}$ so that the initial error is $0$, and estimate the \emph{change} in this quantity when a terminal is added to $C$. This is given by the following formula.
\begin{fact}
\label{fact:projchange}
In any graph $G$, for any subset of vertices $C$
and any vertex $v \notin C$, we have that
\[
\ppi^{C\cup \{v\}}\left(\dd\right)
=
\ppi^{C}\left(\dd\right)
+
\left[\ppi^{C \cup \{v\}}\left(\dd\right)\right]_v
  \cdot \left(\cchi_v - \ppi^{C}\left(\cchi_{v} \right)\right).
\]
\end{fact}
\begin{proof}
It suffices to prove the result for $\dd = \cchi_u$ for some vertex $u \in V(G)$ by linearity. To show this, consider the following random process. Run a random walk starting from vertex $u$ until it hits the set $C \cup \{v\}.$ Then, take any mass on $v$ and random walk that until it hits $C$. Clearly, this is a valid sample for $\ppi^C(\cchi_u)$ by \cref{lemma:projrandwalk}. This implies the claim.
\end{proof}
An alternate algebraic proof of this can be derived using the fact that elimination of rows on matrices is independent of the ordering and instead only depends on the subset. In this view, we first eliminate the demand $\dd$ onto $V\bs(C \cup \{v\})$ to get $\pi^{C \cup\{v\}}(\dd).$ Then we eliminate vertex $v$ to get $\pi^C(\dd)$ and the change is given by the elimination of a single vertex which can be checked to be \begin{align}\left[\ppi^{C \cup \{v\}}\left(\dd\right)\right]_v
  \cdot \left(\cchi_v - \ppi^{C}\left(\cchi_{v} \right)\right) \label{eq:changequantity}
\end{align}
Now, if our estimates of the quantity in \eqref{eq:changequantity} have error $\delta = o(\eps)$ per change, then we can afford to exactly reinitialize every $\delta^{-1}\eps$ terminal insertions while guaranteeing that the total error is at most $\eps$ always.

\subsection{Approximating Projections Using Random Walks}
\label{subsec:approxproj}
The focus of this section is maintaining approximations of $\ppi^C(\dd)$ for demand vectors $\dd$ arising from \cref{def:settings_of_sec_approx_proj} under terminal insertions to $C$. The change in $\ppi^C(\dd)$ from~\cref{fact:projchange}
can be divided into two parts:
\begin{enumerate}
\item
The projection value on $v$, $[\ppi^{C \cup \{v\}}(\dd)]_v$,
which we do in Section~\ref{subsubsec:projvert} and
\item
the ``local'' flow leaving $v$,
$\cchi_v - \ppi^{C}(\cchi_{v})$,
which we do in Section~\ref{subsubsec:localflow}.
\end{enumerate}
We show how to compute approximations for both these quantities using random walks. We care about the order of vertices visited by a random walk only and not the length, and this simplifies the sampling method.
\begin{definition}
\label{def:rwslocator}
Let $G=(V, E)$ be a connected undirected graph with resistances $\rr$. Let $C$ be any set of $O(\beta m)$ vertices. For any vertex $v$ we define a \emph{random walk from $v$ to $C$}, denoted as $\rw$, as the following process: sample a random walk from $v$ where edges $e$ are used with probability proportional to $\rr_e^{-1}$. We go until it hits $C$ or visits a set of distinct vertices whose degrees sum up to at least $O(\beta^{-1} \log m)$. For this walk, we only store the first time it hits vertices.
\end{definition}
The error of our sampling processes will be analyzed
using the following concentration inequality. The proof can be found in \cref{subsec:bernstein}.
\begin{lemma}[Corollary of Bernstein's inequality]
\label{lem:hoeffding}
Let $S=X_1+\dots+X_n$ be the sum of $n$ independent random variables. The range of $X_i$ is $\{0, a_i\}$ for $a_i\in [-M, M]$. Let $t, E$ be positive numbers such that $t\le E$ and $\sum_{i=1}^n\abs{\expec{}{X_i}}\le E$. Then
\[
\pr{}{\abs{S-\E[S]}>t}\le 2\exp\left(-\frac{t^2}{6EM}\right).
\]
\end{lemma}

\subsubsection{Estimating Single Entries of the Projection}
\label{subsubsec:projvert}

We start with the $[\ppi^{C \cup \{v\}}(\dd)]_v$ term. At a high level, we approximate this value by sampling a global set of random walks from each edge up front.
A useful fact for the error analysis is that by augmenting $C$ with the congestion reducing
subset from Lemma~\ref{lemma:reducecongestion},
we can assume that all entries of the projection are small.
Below, the condition $|\dd_u| \le O(\eps^{-1}\deg_u)$ follows from the construction of $\dd = \dd[i]$ in \cref{def:settings_of_sec_approx_proj}, specifically that only edges with $\rr_e \in [\eps^2/20, 20\eps^{-2}]$ are considered.

\begin{lemma}
\label{lem:reducedemand}
Let $G=(V, E)$ be an undirected graph.
The set
\[
C =
\textsc{CongestionReductionSubset}(G, \beta)
\cup
\left\{ v \mid \deg_v \geq \beta^{-1} \right\}
\]
has size at most $O(\beta m)$,
and for any vector $\dd \in \R^{V}$
with $|\dd_u| \le O(\eps^{-1}\deg_u)$ for all vertices $u\in V$, we have that the projection of $\dd$ onto $C$ has small mass on any vertex $v$ not in $C$. Precisely, $|\ppi^{C \cup \{v\}}(\dd)_v| \le \O(\beta^{-2}\eps^{-1})$ for all $v \in V \bs C$.
Note that by Lemma~\ref{lemma:reducecongestion}, this
set can also be computed $\O(m\beta^{-2})$ time.
\end{lemma}
\begin{proof}
We first show that $|C| = O(\beta m).$ Since the total degree is $2m$, the number of vertices with degree at least $\beta^{-1}$ that are added is at most
$2m/\beta^{-1} \le O\left(\beta m\right).$
Combining this with the $O(\beta m)$ size guarantee of $\textsc{CongestionReductionSubset}$ from 
Lemma~\ref{lemma:reducecongestion} gives that $|C| \le O( \beta m)$.

We now turn our attention to bounding $|\ppi^{C \cup \{v\}}(\dd)_v|.$
\cref{lemma:projrandwalk} allows us to express the projection
on $v$ as the original demand vector times hitting probabilities.
We have
\[
\left|\ppi^{C\cup\left\{v\right\}}\left(\dd\right)_v-\dd_v \right|
=
\left|\sum_{u \in V \setminus C \setminus \left\{v \right\}}
\dd_u \cdot \prhit{u}{v}{C\cup\left\{v\right\}} \right|
\le
\eps^{-1}\sum_{u\in V \setminus C \setminus \left\{v \right\}}
\deg_u
\cdot
\prhit{u}{v}{C\cup\left\{v\right\}}.
\]
where the last step uses the triangle inequality and ${\dd_u} \le O(\eps^{-1}\deg_u)$.
Now, the guarantee of Lemma~\ref{lemma:reducecongestion} gives us
\[ \eps^{-1}\sum_{u\in V \setminus C \setminus \left\{v \right\}}
\deg_u
\cdot
\prhit{u}{v}{C\cup\left\{v\right\}} \le \O(\beta^{-2}\eps^{-1}). \]
Also, because all vertices with degrees more than $\beta^{-1}$
were already added to $C$, we have $\deg_v \leq \beta^{-1}$ for $v \notin C.$ Together with the assumption ${\dd_v} \leq \deg_v \eps^{-1}$ we get
\[
{\dd_v} \le \eps^{-1}\beta^{-1} \le \eps^{-1}\beta^{-2}.
\]
Hence
\[ |\ppi^{C \cup \{v\}}(\dd)_v| \le \O(\beta^{-2}\eps^{-1}) + |\dd_v| \le \O(\beta^{-2}\eps^{-1}). \]
\end{proof}

As these projection values have magnitude at most
$\poly(\beta^{-1}, \eps^{-1})$ by this lemma,
we can approximate them by to additive $\delta$ accuracy by
sampling $\poly(\beta^{-1}, \eps^{-1}, \delta^{-1})$
random walks per edge, and treating the samples as an approximate
projection of $\dd$ onto $C$. We now formalize this sampling process by defining
the approximate projection vector, $\ppitil^{C}(\dd)$. This is done in the for loop starting at line \ref{line:cost1} of our implementation in the algorithm in \cref{algo:add_terminal}.
\begin{definition}
\label{def:rws}
For a graph $G = (V, E)$ with resistances $\mR$, along with a vector on edges $\qq \in \{-1, 0, 1\}^E$, define $\dd = \mB^{\tomato} \mR^{-1/2} \mI_S \qq$. For a fixed sampling overhead
$h
\leftarrow
\Omega\left(\delta^{-2}\beta^{-2}\eps^{-2}\polylog{m}\right),
$
and vertex subset $C$,
for each edge $e \in S$ with endpoint $v$,
and each $1 \leq j \leq h$,
let $\rw(v, e, j)$ be a random walk from $v$ to $C$ as in \cref{def:rwslocator},
with associated weight $h^{-1}\qq_e\rr_e^{-1/2}$.
Then for $u \in C$ define $\left[\ppitil^{C}(\dd)\right]_u$ as the sum of weights of all random walks $\rw(v, e, j)$ that hit $C$ for the first time at $u$.
\end{definition}

\begin{lemma}
\label{lemma:approxprojguarantee}
For sampling overhead $h = \Omega\left(\delta^{-2}\beta^{-2}\eps^{-2}\polylog{m}\right)$, graph $G$, subset $C$,
edge vector $\qq$, and demand $\dd$ as in Definition~\ref{def:rws},
Let $\ppitil^C(\dd)$ be the estimated
projections as defined in Definition~\ref{def:rws}.

Let $\ddhat = |\mB|^\top \mR^{-1/2} \mI_S|\qq|$, where $|\mB|,|\qq|$ are defined as taking the entrywise absolute values of the matrix, vector respectively. For vertices $v$ where
$\ppi^{C \cup \{v\}}(\ddhat)_v=\O(\beta^{-2}\eps^{-1})$,
the estimate w.h.p satisfies
\[
\abs{\ppitil^{C\cup\left\{v\right\}}\left(\dd\right)_v
-
\ppi^{C\cup\left\{v\right\}} \left(\dd\right)_v}
\leq
\delta.
\]
\end{lemma}

\begin{proof}
Let $F=V\setminus C$,
and let $\rw(v, e, j)$ be the weighted random walks as defined
in Definition~\ref{def:rws}.

We first split the the exact projection value on $v$ into contributions
per vertex of $F$ using Lemma~\ref{lemma:projrandwalk},
while incorporating the relation between $\dd$ and $\qq$:
\[
\left[\ppi^{C \cup \{v\}}\left(\dd\right)\right]_v - \dd_v
=
\sum_{u\in F} \dd_u \cdot p^{C}_{v}\left( u \right)
=
\sum_{u\in F}
\left(\sum_{e\ni u} \qq_e \rr_e^{-1/2} \right) \cdot \prhit{u}{v}{C} = \sum_{u\in F} \sum_{e\ni u}
\sum_{j=1}^{h} \frac{\qq_e \rr_e^{-1/2} }{h} \prhit{u}{v}{C}.
\]
In the last step we have split it into $h$ copies to compare it to the random sampling process.
Now each of these terms is exactly the expectation of the weights
of the walks $\rw(u, e, j)$ being sent to $v$.
Formally, we define random variables
\[
\mathsf{RWHIT}\left(u, e, j, v\right)
=
\begin{cases}
\frac{\qq_e \rr_e^{-1/2} }{h}
& \text{if $\rw(u, e, j)$ reaches $C$ first at $v$},\\
0 & \text{otherwise}.
\end{cases}
u \in e, j \in \left[ h \right]
\]
which by the above expression yields:
\[
\ppi^{C\cup\left\{v\right\}}\left(\dd\right)_v - \dd_v
=
\expec{\rw\left(u, e, j\right)}
{\sum_{u\in F} \sum_{e\ni u} \sum_{j=1}^{ h} \mathsf{RWHIT}(u, e, j, v)}.
\]

By the assumption of $\ppi^{C \cup \{v\}}(\ddhat)_v=\O(\beta^{-2}\eps^{-1})$, the expected sum of the absolute values of $\mathsf{RWHIT}(u, e, j, v)$s is $\O(\beta^{-2}\eps^{-1})$.
Each term on the other hand has magnitude at most
\[
\frac{\qq_e \rr_e^{-1/2} }{h}
\le
O\left(\frac{\eps^{-1}}{ h}\right)
=
O\left(\delta^{2}\beta^{2}\eps/\polylog m\right),
\]
as $\rr_e \ge \eps^2/20$ for $e \in S.$
As the walks are sampled independently \cref{lem:hoeffding} gives
\begin{align*}
&\pr{\rw\left(u, e, j\right)}
{\abs{\ppitil^{C\cup\left\{v\right\}}\left(\dd\right)_v
-
\ppi^{C\cup\left\{v\right\}}\left(\dd\right)_v}
>\delta
} \\
&\le
2\exp\left(-\frac{\delta^2}
{6\O\left(\beta^{-2}\eps^{-1}\right)\delta^{2}\beta^{2}\eps/\polylog m}\right)
=n^{-10},
\end{align*}
for a suitable choice of $\polylog m$ in $h$.
\end{proof}

\subsubsection{Locally Sampling Flows From \texorpdfstring{$v$}{vertexv}}
\label{subsubsec:localflow}

We now turn to approximating the $\cchi_v - \ppi^C(\cchi_v)$ term from Fact~\ref{fact:projchange}. By \cref{lemma:projrandwalk}, we know that $\ppi^C(\cchi_v)$ is the distribution over $C$ of the first vertex that a random walk starting at $v$ hits $C$ at. To estimate this, we sample \emph{fresh} random walks from $v$ onto $C$, as the sampling overhead of the random walks in e.g. \cref{lemma:approxprojguarantee} is not sufficient.

Ultimately, we take the inner product of $\cchi_v - \ppi^C(\cchi_v)$ against $\mSC(G, \Chat)^\dagger \cchi_{st}$, the potentials on some superset $\Chat \supseteq C$. The variance each random walk contributes to this inner product is $O(1)$ as
$\|\mSC(G, \Chat)^\dagger \cchi_{st}\|_\infty \le O(1)$,
so generating $\O(\delta^{-2}\beta^{-4}\eps^{-2})$ random walks from $v$ is sufficient to dampen out the coefficient of $\ppi^{C \cup\{v\}}(\dd)_v$ in \cref{fact:projchange} which we know is bounded by $\O(\beta^{-2}\eps^{-1})$ for $v \notin C$ by \cref{lem:reducedemand}.
\begin{lemma}
\label{lem:approx_proj_from_one_vertex}
Let $G=(V, E)$ be an undirected graph. Let $C$ be a set of vertices. Let $\pphi$ be a vector in $[-2, 2]^V$. Let $v$ be a vertex in $V\setminus C$. We sample $ h_2=\Omega(\delta^{-2}\beta^{-4}\eps^{-2}\polylog n)$ random walks from $v$ to $C$. For any vertex $u\in C$, let $\left[\ppitil^C(\cchi_v)\right]_u$ be the number of random walks that hit $u$ before other vertices in $C$ divided by $ h_2$. We have, w.h.p.,
\[
\abs{
\left\langle \cchi_v - \ppi^C(\cchi_{v}), \pphi\right\rangle
-
\left\langle \cchi_v - \ppitil^C(\cchi_v), \pphi\right\rangle
}
\le \delta\beta^2\eps
\]
\end{lemma}
\begin{proof}
We first consider the contribution of $\ppi^C(\cchi_{v})$ to each entry of the dot product separately. $\left[\ppitil^C(\cchi_v)\right]_u\cdot \pphi_u$ is the average of $ h_2$ i.i.d. binary random variables with expectation $\left(\ppi^C(\cchi_v)\right)_u\pphi_u$. Summing the contribution up over every entry, we know that $\left\langle \ppitil^C(\cchi_v), \pphi\right\rangle$ is the sum of $ h_2$ independent random variables and the expected sum of them is $\left\langle \ppi^C(\cchi_{v}), \pphi\right\rangle$. Every random variable has range $[0,2/ h_2]$, as
\[ \left\| \left\langle \ppi^C(\cchi_v), \pphi \right\rangle \right\| \le \left\| \ppi^C(\cchi_v) \right\|_1 \left\| \pphi \right\|_\infty \le 2. \]
By Lemma \ref{lem:hoeffding}, the probability over the $h_2$ random walks that our estimate has additive error at most $\delta\beta^2\eps$ satisfies
\[
\pr{}{
\abs{
\left\langle \ppitil^C(\cchi_v), \pphi\right\rangle 
- \left\langle \ppi^C(\cchi_{v}), \pphi\right\rangle
}
\ge \delta\beta^2\eps
}
\le 
2\exp\left(\frac{-\left(\delta\beta^2\eps\right)^2}{ 6\cdot 2\cdot 2/ h_2}\right)
=
n^{-10}
\] by picking a large enough $\polylog n$ term in $h_2$.
\end{proof}

\subsection{Maintaining Approximate Projections}
\label{subsec:projector}

In this subsection, we maintain the estimates of the projections $\ppi^C(\dd[i])$ of the
$\ndd = \O(\eps^{-2})$ vectors $\dd[1]$, $\dd[2]$, $\ldots$, $\dd[\ndd]$ as required by the heavy-hitter data structure described in Lemma~\ref{lem:heavyhitter}
using the random walks structures studied above.
Here, the $\dd[i]$ are constructed as in \cref{def:settings_of_sec_approx_proj}. The maintained $\ppitil^C(\ddi)$ vectors which approximate $\ppi^C(\ddi)$ will be dotted against a potential vector $\pphi$, which is maintained separately using a dynamic Schur complement via \cref{thm:dynamicsc}.

\begin{lemma}
\label{lem:projector}
There is a data structure \textsc{Projector} that supports the following operations. 
\begin{itemize}
\item \textsc{InitRandomWalks($G, \rr, \epsilon, \beta, \delta, \dd[1], \ldots, \dd[\ndd]$)}:
Set the parameters $G, \rr, \epsilon, \beta, \delta$ and $\dd[i]$ for $i \in [N]$ and initialize the data structure (including a initial terminal set $C\subseteq V$). Takes $\O(m\delta^{-2}\beta^{-4}\eps^{-2})$ time. 
\item \textsc{InitProjections}: Solve each $\ppitil^C(\dd[i])$ exactly. Takes $\O(m\eps^{-2})$ time.
\item \textsc{AddTerminal($v$)}: Add $v$ to $C$ in
$\O(\delta^{-2}\beta^{-6}\eps^{-2})$ time.
The algorithm supports at most $\beta m$ \textsc{AddTerminal} operations.


\item \textsc{Query($i\in [\ndd]$)}: Output a vector $\ppitil^C(\dd[i])$. Let $\pphi\in [-2,2]^V$ be a vector chosen by an adversary which is oblivious after given the initial terminal set $C$. The output satisfies
\[
\abs{
\left\langle\ppi^C
\left(\dd\left[i\right]\right)
,
\pphi \right\rangle
-
\left\langle\ppitil^C\left(\dd\left[i\right]\right)
,
\pphi \right\rangle}
=
\O\left(count \cdot \delta\right)
\]
with high probability where $count$ is the number of \textsc{AddTerminal}s
performed after the last \textsc{InitProjections}. The runtime is $\O(\beta m).$
\end{itemize}
The algorithm creates an initial terminal set $C$ of $O(\beta m)$ vertices by its own decision during initialization in \textsc{InitRandomWalks}. Each operation succeeds with high probability against an adversary which is oblivious after given this original set of terminals $C$ built during initialization.
\end{lemma}

As mentioned, our algorithms require sampling random walks.
We show formally that random walks can be sampled efficiently. This runtime is more efficient than \cref{lemma:effwalks} because we do not require the sum of resistances on the walk.

\begin{lemma}
\label{lem:sample_exit_of_rw}
Let $G=(V, E)$ be a connected undirected weighted graph. We choose $
\beta m$ random edges uniformly from $E$.
Let $C$ denote the set of endpoints of the chosen edges. There is an algorithm that
 samples a random walk from $v$ to $C$ as in \cref{def:rwslocator} for any vertex $v\in V \bs C$. The algorithm runs in $\O(\beta^{-2})$ time and succeeds with high probability.
\end{lemma}

\begin{proof}
Fix the vertex $v$ and let the random walk from $v$ be $W$. By the construction of $C$, once the sum of degrees of the vertices visited by $W$ exceeds some number $c=O(\beta^{-1} \log m)$, with high probability, $W$ has visited at least one vertex in $C$. By this property, our algorithm continues sampling a new vertex that $W$ has never visited until 
\begin{enumerate}
\item $W$ hits some vertex in $C$, or 
\item the sum of degrees of the vertices visited by $W$ exceeds $c$.
\end{enumerate} 
In the second case (which happens with $n^{-10}$ probability), the algorithm fails.

Now we discuss how to sample the next new vertex $W$ will visit. We denote this vertex by $u$ (which is a random variable). Suppose $U$ is the set of vertices that $W$ has visited and $W$ is on some vertex $x\in U$ currently. Let $N(U)$ be the set of vertices that are adjacent to $U$ and not contained by $U$. We construct an induced subgraph $H=G[N(U)\cup U]$. It suffices to sample $u$ on $H$ because $u$ must be adjacent to $U$.

By \cref{lemma:projrandwalk}, for any vertex $z\in N(U)$, $\pr{}{u=z}=\left[\ppi^{N(U)}(\cchi_x)\right]_z$. In other words, the mass function for $u$ is $\ppi^{N(U)}(\cchi_x)$. Thus, we may use an exact Laplacian solver to compute 
\[
\ppi^{N(U)}(\cchi_x)=-\mL_{N(U)U}\mL_{UU}^{-1}\cchi_{x}
\] 
and sample according to it. Since the number of edges in $H$ is bounded by $\sum_{v\in U}\deg_v=\O(\beta^{-1})$, the solving process takes $\O(|H|)=\O(\beta^{-1})$ time.

As each new vertex we sample has degree at least $1$, the process must end before $c=O(\beta^{-1}\log n)$ distinct vertices are visited. Thus, the total running time is $\O(c\cdot \beta^{-1})=\O(\beta^{-2})$.
\end{proof}

We start with the initialization algorithms \textsc{InitRandomWalks} and \textsc{InitProjections}, given in \cref{algo:adjust}. We will use the method in \cref{lem:sample_exit_of_rw} for sampling random walks. 

\begin{algorithm}[ht]
\caption{Pseudocode for initialize the estimated $\ppi$ values and the random walks\label{algo:adjust}.}
\SetKwProg{Globals}{global variables}{}{}
\SetKwProg{Proc}{procedure}{}{}
\Globals{}{
    $\dd[1],\ldots, \dd[\ndd]$: the demand vectors defined in \cref{def:settings_of_sec_approx_proj},     $\ndd = \O(\eps^{-2})$ is the total count.\\
	$\rw(v, e, j)$ $v\in V, e\ni v, 1\le j\le  h$: the sampled random walks\\ 
	$C$: the terminal set\\
	$\ppitil^C(\dd[i])$ $(1\le i\le \ndd)$: the estimated $\ppi$ values\\
    $\eps, \beta, \delta$: parameters controlling accuracy, size of the terminal set $C$ and frequency to repair $\ppitil^C(\dd[i])$\\
}
\Proc{$\textsc{Projector.InitProjections}()$}{
  \For{$1 \le i \le \ndd$}{
    Solve $\mL_{FF}\xxtil[i]=\dd[i]_F$.\\
    $\ppitil^C(\dd[i])
      \assign \dd[i]_C - \mL_{CF}\widetilde{\xx}[i]$.\\
  }
}
\Proc{$\textsc{Projector.InitRandomWalks}(G, \rr, \epsilon, \beta, \delta, \dd[1], \ldots, \dd[\ndd])$}{
  Set the parameters $G, \rr, \epsilon, \beta, \delta, \dd[1], \ldots, \dd[\ndd]$ according to the input.\\
  $C \assign \textsc{CongestionReductionSubset}(G, \beta) \cup \left\{ v : \deg_v \geq \beta^{-1} \right\}$. \label{line:reducecongestion}
  \tcp{\cref{lem:reducedemand}}
  $ h\assign \O(\beta^{-2}\delta^{-2}\eps^{-2})$, the sampling overhead of \cref{lemma:approxprojguarantee}.\\
  \For{$v\in V$, $e\ni v$, $1\le j\le  h$} {
    Sample $\rw(v, e, j)$ (a random walk from $v$ to $C$ of length $O(\beta^{-1} \log m)$) by the method in \cref{lem:sample_exit_of_rw}.\\
  }
  \textsc{Projector.InitProjections}().\\
}
\end{algorithm}

Each \textsc{InitProjections} operation calculates
$\ppitil^C(\dd[i])$ for each $1\le i\le \ndd$ with exact Laplacian solves, and stores the $\ppitil^C(\dd[i])$ values in a $\ndd$ by $|C|$ array. This runs in time $\O(\ndd m) = \O(m\eps^{-2})$.

To initialize the random walks in \textsc{InitRandomWalks}, we first initialize the terminals $C$ by \cref{lem:reducedemand}. Note that we have
\[
\left|\ppi^{C \cup \{v\}}\left(\dd[i]\right)_v\right| \le \O(\beta^{-2}\eps^{-1})
\] for any $i \in [N]$ and $v \notin C$ with high probability by \cref{lem:reducedemand}. This will be used to bound the error from adjusting the projections when adding terminals.
Then we sample $\O(\beta^{-2}\delta^{-2}\eps^{-2})$ random walks per edge according to \cref{lem:sample_exit_of_rw}.
Then \textsc{InitProjections} is called to initialize the array for $\ppitil^C(\dd[i])$.

Now, we give the algorithm for $\textsc{AddTerminal}(v)$ in \cref{algo:add_terminal}.

\begin{algorithm}[!ht]
\caption{Pseudocode for adding a vertex as terminal \label{algo:add_terminal}.}
\SetKwProg{Globals}{global variables}{}{}
\SetKwProg{Proc}{procedure}{}{}
\Globals{}{
    $\dd[1],\ldots, \dd[\ndd]$: the demand vectors defined in \cref{def:settings_of_sec_approx_proj},     $\ndd = \O(\eps^{-2})$ is the total count.\\
	$\rw(v, e, j)$ $v\in V, e\ni v, 1\le j\le  h$: the sampled random walks\\ 
	$C$: the terminal set\\
	$\ppitil^C(\dd[i])$ $(1\le i\le \ndd)$: the estimated $\ppi$ values\\
    $\eps, \beta, \delta$: parameters controlling accuracy, size of the terminal set $C$ and frequency to repair $\ppitil^C(\dd[i])$\\
}
\Proc{$\textsc{Projector.AddTerminal}(v)$}{
  \For{$1\le i\le \ndd$}{
      $\left[\ppitil^{C\cup \{v\}}(\dd[i])\right]_v\assign 0$.\\
  }
  \For{$1 \le i \le \ndd$\label{line:cost1}}{
      \For{each $(u, e, j)$ such that $u\in V, e\ni u, 1\le j\le  h$ and $\rw(u, e, j)$ visits $v$}{
        Shorten $\rw(u, e, j)$ to the first time it hits $v$.\\
        $\left[\ppitil^{C\cup \{v\}}(\dd[i])\right]_v \assign \left[\ppitil^{C\cup \{v\}}(\dd[i])\right]_v + \frac{\qq_e\rr_e^{-1/2}}{ h}$.
      }
  }
    \tcp{$\left[\ppitil^{C\cup \{v\}}(\dd[i])\right]_v$ is an estimate of $\left[\ppi^{C\cup \{v\}}(\dd[i])\right]_v$ by \cref{lemma:approxprojguarantee}.}
  $\ppitil^C(\cchi_v)\assign \zzero$.\\
  $ h_2 \assign \O(\delta^{-2}\beta^{-4}\eps^{-2})$, the sampling overhead in \cref{lem:approx_proj_from_one_vertex}. \\
  \For{$1\le i\le  h_2$ \label{line:cost2}} {
    Sample a random walk from $v$ to $C$ by the method in \cref{lemma:effwalks}.\\
    Let $u$ be the vertex at which the sampled random walk above hits $C$. \\
    $\left[\ppitil^C(\cchi_v)\right]_u \assign \left[\ppitil^C(\cchi_v)\right]_u + \frac{1}{ h_2}$.\\
  }
  \tcp{$\ppitil^C(\cchi_v)$ is an estimate of $\ppi^C(\cchi_v)$ by \cref{lem:approx_proj_from_one_vertex}.}
  \For{$1\le i\le \ndd$ \label{line:cost3}} {
    \For{Each $u$ such that $u\in C$ and  $\left[\ppitil^C(\cchi_v)\right]_u\neq 0$ \label{line:nonzerou}}{
      $\left[\ppitil^{C\cup \{v\}}\left(\dd[i]\right)\right]_u\assign \left[\ppitil^C\left(\dd[i]\right)\right]_u + \left[\ppitil^{C \cup \{v\}}\left(\dd[i]\right)\right]_v \cdot \left[\left(\cchi_v - \ppitil^C\left(\cchi_{v} \right)\right)\right]_u$.\\
    }
    $\left[\ppitil^{C\cup \{v\}}\left(\dd[i]\right)\right]_v\assign \left[\ppitil^C\left(\dd[i]\right)\right]_v + \left[\ppitil^{C \cup \{v\}}\left(\dd[i]\right)\right]_v \cdot \left[\left(\cchi_v - \ppitil^C\left(\cchi_{v} \right)\right)\right]_v$.\\
  }
  \tcp{Approximate update corresponding to \cref{fact:projchange}.}
}
\end{algorithm}

The algorithm works as follows. 
\begin{enumerate}
\item By Lemma \ref{fact:projchange}, we update
$\ppitil^C(\dd[i])$ to $\ppitil^{C\cup \{v\}}(\dd[i])$ by 
\[
\ppitil^{C\cup \{v\}}\left(\dd\left[i\right]\right)
\assign
\ppitil^C\left(\dd\left[i\right]\right)
+
\left[\ppitil^{C \cup \{v\}}\left(\dd\left[i\right]\right)\right]_v
\cdot
\left(\cchi_v - \ppitil^C\left(\cchi_{v} \right)\right)
\]
for each $i\in [\ndd]$, for some estimate $\left[\ppitil^{C \cup \{v\}}\left(\dd\left[i\right]\right)\right]_v$ of $\left[\ppi^{C \cup \{v\}}\left(\dd\left[i\right]\right)\right]_v$, and estimate $\ppitil^C\left(\cchi_{v} \right)$ of $\ppi^C\left(\cchi_{v} \right)$.
\item We estimate $\left[\ppitil^{C \cup \{v\}}\left(\dd\left[i\right]\right)\right]_v$ using the initial random walks sampled in \textsc{InitRandomWalks}. Cut the initial random walks by where they hit $v$.
 By Lemma \ref{lemma:approxprojguarantee}, for each $i$, we calculate our estimate $[\ppitil^{C\cup \{v\}}(\dd[i])]_v$ according to the random walks that are cut. Note that the condition $\ppi^{C \cup \{v\}}(\ddhat)_v=\O(\beta^{-2}\eps^{-1})$ of \cref{lemma:approxprojguarantee} is guaranteed by \cref{lem:reducedemand}.
\item We estimate the vector $\ppitil^C\left(\cchi_{v} \right)$ by sampling $\O(\delta^{-2}\beta^{-4}\eps^{-2})$ fresh random walks from $v$ and tracking where they hit $C$. Applying \cref{lem:approx_proj_from_one_vertex} allows us to bound the error.
\item Once the updates above are done, we add $v$ to $C$, $C\leftarrow C\cup\{v\}$.
\end{enumerate}

Finally, for each $\textsc{Query}(i)$, we simply output the maintained vector $\ppitil^C(\dd[i])$ for the current terminal set $C$. This costs $O(|C|) = O(\beta m)$ time.

We can use \cref{algo:adjust,algo:add_terminal} to prove \cref{lem:projector}.

\begin{proof}[Proof of \cref{lem:projector}]
We first show the correctness of the \textsc{Query} operation, then bound the runtimes of each operation.
\paragraph{Correctness.}
It suffices to show the correctness of $\textsc{Query}(i)$, i.e. prove that
\[ \left| \left\langle \ppi^C(\dd[i]) - \ppitil^C(\dd[i]), \pphi \right\rangle \right| \le \O(count \cdot \delta) \] with high probability, where $count$ is the number of operations since the last \textsc{InitProjections} operation. To bound this error we sum the errors caused by each \textsc{AddTerminal($v$)} up.
Let $C^{(0)}$ denote the set $C$ after the last \textsc{InitProjections} operation. Suppose that after the last \textsc{InitProjections} operation, we perform $count$ 
\textsc{AddTerminal} operations.
The $j$-th one of them is \textsc{AddTerminal($v_j$)} ($1\le j\le count$).
Let $C^{(j)}=C^{(j-1)}\cup \{v_j\}$ be the set $C$ after $j$ \textsc{AddTerminal}s.

For simplicity of notation, we overload $\ppi^C$ vector to refer to
the length $n$ vector on $V$ formed by padding $0$ entries for vertices
not in $C$.
Suppose after the $count$ \textsc{AddTerminal} operations, one \textsc{Query($i, \pphi=\pphi_{C^{(count)}}$)} is performed.
We can decouple the error across the steps
and across the decomposition given by \cref{fact:projchange} into:
\begin{align*}
&\abs{\left\langle \ppi^{C^{(count)}}(\dd[i]), \pphi\right\rangle - \left\langle \ppitil^{C^{(count)}}(\dd[i]), \pphi\right\rangle}\\
=&\left|\left\langle \ppi^{C^{(0)}}(\dd[i])+\sum_{j=1}^{count}
\left[\ppi^{C^{(j)}}\left(\dd[i]\right)\right]_{v_j}
  \cdot \left(\cchi_v - \ppi^{C^{(j-1)}}\left(\cchi_{v_j} \right)\right), \pphi\right\rangle\right.\\
      &\qquad -\left.\left\langle \ppitil^{C^{(0)}}(\dd[i])+\sum_{j=1}^{count}
    \left[\ppitil^{C^{(j)}}\left(\dd[i]\right)\right]_{v_j}
      \cdot \left(\cchi_v - \ppitil^{C^{(j-1)}}\left(\cchi_{v_j} \right)\right), \pphi\right\rangle\right|\\
    \le & \abs{\left\langle \ppi^{C^{(0)}}(\dd[i]),\pphi\right\rangle-\left\langle \ppitil^{C^{(0)}}(\dd[i]),\pphi\right\rangle}\\
+ &\sum_{j=1}^{count}
\abs{\left[\ppi^{C^{\left(j\right)}}\left(\dd\left[i\right]\right)\right]_{v_j}\left\langle
\left(\cchi_v - \ppi^{C^{(j-1)}}\left(\cchi_{v_j} \right)\right), \pphi\right\rangle
- \left[\ppitil^{C^{\left(j\right)}}\left(\dd\left[i\right]\right)\right]_{v_j} \left\langle
\left(\cchi_v - \ppitil^{C^{(j-1)}}\left(\cchi_{v_j} \right)\right), \pphi\right\rangle}
\end{align*}

By Lemma \ref{lemma:approxprojguarantee} and Lemma \ref{lem:approx_proj_from_one_vertex}, for each fixed $j\in [count]$,
we have $\delta$-additive approximations to $[\ppi^{C^{(j)}}(\dd[i])]_{v_j}$ and
$\delta\beta^2\eps$-additive approximations to
$\langle \cchi_{v_j} - \ppi^{C^{(j-1)}}(v_j), \pphi \rangle$. 
Therefore for each term $j$, we can further apply triangle inequality
to bound the error by:
\begin{align*}
& \abs{\left[\ppi^{C^{\left(j\right)}}\left(\dd\left[i\right]\right)\right]_{v_j}
\cdot
\left\langle
\cchi_v - \ppi^{C^{(j-1)}}\left(\cchi_{v_j} \right)
,
\pphi
\right\rangle
-
\left[\ppitil^{C^{\left(j\right)}}\left(\dd\left[i\right]\right)\right]_{v_j}
\cdot
\left\langle
\cchi_v - \ppitil^{C^{(j-1)}}\left(\cchi_{v_j} \right)
,
\pphi
\right\rangle}\\
\le
& \abs{\left[\ppi^{C^{(j)}}\left(\dd[i]\right)\right]_{v_j}}
\cdot
\abs{
\left\langle
\cchi_v - \ppi^{C^{(j-1)}}\left(\cchi_{v_j}\right), \pphi\right\rangle
- \left\langle
  \cchi_v - \ppitil^{C^{(j-1)}}\left(\cchi_{v_j} \right), \pphi\right\rangle
}\\
&
\qquad +
\abs{\left[\ppi^{C^{\left(j\right)}}
  \left(\dd\left[i\right]\right)\right]_{v_j}
-
\left[\ppitil^{C^{\left(j\right)}}\left(\dd\left[i\right]\right)\right]_{v_j}
}
\cdot
\abs{\left\langle
  \cchi_v - \ppitil^{C^{(j-1)}}\left(\cchi_{v_j} \right), \pphi\right\rangle}
\end{align*}
The two factors of the first term are bounded by $\O(\beta^{-2}\eps^{-1})$ (\cref{lem:reducedemand}) and $\delta\beta^2\eps$ (\cref{lem:approx_proj_from_one_vertex}).
The first factor of the second term is bounded by $\delta$ (\cref{lemma:approxprojguarantee}).
Putting these together, we get that the error for step $j$ is no more than
\[
\O\left(\beta^{-2}\eps^{-1}\right)\delta\beta^2\eps
+ \delta\left\langle
      \left(\cchi_v - \ppitil^C\left(\cchi_{v_j} \right)\right), \pphi\right\rangle
\le \O\left(\delta\right)
+ \delta \norm{\left(\cchi_v - \ppitil^C\left(\cchi_{v_j} \right)\right)}_1\norm{\pphi}_{\infty}\\
=
\O(\delta).
\]
The exact solver can guarantee that 
\[
\left\langle \ppi^{C^{(0)}}(\dd[i]),\pphi\right\rangle-\left\langle \ppitil^{C^{(0)}}(\dd[i]),\pphi\right\rangle = 0.
\]
Summing up the errors for each step and the initial error, we have 
\[
\abs{
\left\langle \ppi^{C^{\left(count\right)}}\left(\dd[i]\right),
\pphi\right\rangle
-
\left\langle \ppitil^{C^{\left(count\right)}}\left(\dd[i]\right),
\pphi\right\rangle}
=
\O\left(count \cdot \delta\right).
\]
This concludes the correctness proof.

\paragraph{Runtime.} We go operation by operation.
\begin{itemize}
    \item \textsc{InitProjections}. Requires solving $\ndd$ SDD systems in $\mL_{FF}^{-1}.$ Each can be solved in $\O(m)$ time by \cref{thm:lap}, for a total of $\O(\ndd m) = \O(m\eps^{-2})$ time.
    \item \textsc{InitRandomWalks}. In addition to the $\O(\ndd m)$ cost from \textsc{InitProjections}, the algorithm samples $O(h) = O(\beta^{-2}\delta^{-2}\eps^{-2})$ random walks of length $O(\beta^{-1} \log m)$ for each edge incident to each vertex. 
    Now, by \cref{lem:sample_exit_of_rw}, this uses $\O(m \cdot \beta^{-2}\delta^{-2}\eps^{-2} \cdot \beta^{-2}) = \O(m\beta^{-4}\delta^{-2}\eps^{-2})$ time. Also, the cost of calling \textsc{CongestionReductionSubset}$(G, \beta)$ is $\O(m\beta^{-2})$ by \cref{lemma:reducecongestion}, which is dominated by the previous runtime.
    \item \textsc{AddTerminal}. The runtime cost for computing $\left[\ppitil^{C \cup \{v\}}(\dd[i])\right]_v$ for all $i \in [N]$ starting in the for loop in line \ref{line:cost1} of \cref{algo:adjust}, can be bounded by the total number of times the random walks $\rw(u, e, j)$ pass through $v$, multiplied by $N$. As above, $h = \O(\beta^{-2}\delta^{-2}\eps^{-2})$ random walks are sampled per edge. Because $v \notin C$ for $C$ constructed in line \ref{line:reducecongestion} of \cref{algo:adjust}, \cref{lemma:reducecongestion} tells us that $v$ is in $\O(\beta^{-2}\delta^{-2}\eps^{-2} \cdot \beta^{-2}) = \O(\beta^{-4}\delta^{-2}\eps^{-2})$ random walks with high probability. Thus the cost of this loop is $\O(\beta^{-4}\delta^{-2}\eps^{-2}) \cdot N = \O(\beta^{-4}\delta^{-2}\eps^{-4}).$
    
    The runtime cost for computing $\ppitil^C(\cchi_v)$, which starts in the for loop of line \ref{line:cost2} of \cref{algo:adjust}, is given by the cost of sampling $ h_2 = \O(\delta^{-2}\beta^{-4}\eps^{-2})$ random walks from $v$ of length $O(\beta^{-1} \log m)$ until they hit $C$. By \cref{lem:sample_exit_of_rw}, this costs $\O(\beta^{-2})$ per walk, for a total time of $\O(\delta^{-2}\beta^{-4}\eps^{-2} \cdot \beta^{-2}) = \O(\delta^{-2}\beta^{-6}\eps^{-2}).$
    
    Finally, we bound the cost of aggregating these estimates in the for loop starting in line \ref{line:cost3} of \cref{algo:adjust}. At most $ h_2 = \O(\delta^{-2}\beta^{-4}\eps^{-2})$ values of $u$ are accessed in line \ref{line:nonzerou}, as $h_2$ random walks are sampled from $v$. Therefore, the runtime is bounded by
    $\O(\delta^{-2}\beta^{-4}\eps^{-2} \cdot \ndd) = \O(\delta^{-2}\beta^{-4}\eps^{-4}).$
    
    For $\beta < \eps,$ the $\O(\delta^{-2}\beta^{-6}\eps^{-2})$ term dominates the runtime of all three parts.
    \item \textsc{Query}. As $\ppitil^C(\dd[i])$ is supported on $C$, this costs $O(|C|) = O(\beta m)$ time.
\end{itemize}
\end{proof}

\subsection{Locator Using Maintained Projections}
\label{subsec:locatorproof}

Here, we give the necessary algorithms to prove \cref{thm:locator}, where we without loss of generality assume $\cE = 1.$ At a high level our approach maintains both an dynamic Schur complement data structure \textsc{DynamicSC} using \cref{thm:dynamicsc}, as well as the demand projector \textsc{Projector} in \cref{lem:projector}.

We give the algorithm that implements the initialization operation \Locator.\textsc{Initialize} in \cref{algo:initlocator}. To initialize \Locator, i.e. \Locator.\textsc{Initialize}, we first sample a random heavy hitter matrix $\mQ$ and compute vectors $\dd[i]$ as in \cref{def:settings_of_sec_approx_proj}.
To emphasize, we note that we only consider edges with $\rr_e \in [\eps^2/20, 20\eps^2].$

\begin{algorithm}[!ht]
\caption{Pseudocode for \textsc{Locator} initialization. \label{algo:initlocator}}
\SetKwProg{Globals}{global variables}{}{}
\SetKwProg{Proc}{procedure}{}{}
\Globals{}{
    $D^{(\mathrm{sc})}$: Instance of \textsc{DynamicSC} data structure of \cref{thm:dynamicsc}. \\
    $D^{(\mathrm{proj})}$: Instance of \textsc{Projector} data structure of \cref{lem:projector}. \\
    $\eps, \beta, \delta$: pamaters controlling accuracy, size of the terminal sets $C, T$ and frequency to call $D^{(\mathrm{proj})}.\textsc{InitProjections}()$.\\
    $\rr$: the resistance vector.\\
	$C$, $T$: the current terminal sets of $D^{(\mathrm{proj})}$ and $D^{(\mathrm{sc})}$.\\
	$\mQ$: heavy-hitter matrix produced by \cref{lem:heavyhitter}.\\
	$count$, $\overline{count}$: number of terminals added since last initialization, and threshold.
}
  \Proc{\textsc{Locator.Initialize}($G, \rr, \eps, \beta, \delta$)}{
  Set the parameters $G, \rr, \epsilon, \beta, \delta$ according to the input.\\
  Initialize $\mQ\in \mathbb{R}^{\ndd\times m}$ by \cref{lem:heavyhitter}, and $\dd[i] = 0 \forall 1 \leq i \leq \ndd$.\\
  \For{each $e\in E$ and $i\in [\ndd]$ with $\mQ_{i, e}\neq 0$ \label{line:multiplybyq}} {
    \If{$ \eps^{2}/20 \leq \rr_e \leq 20\eps^{-2}$}{
   $\dd[i] \assign \dd[i] + \mQ_{i, e} \rr_e^{-1/2} \cchi_e.$\\
    }
  }
  $D^{(\mathrm{proj})}$.\textsc{InitRandomWalks($G, \rr, \eps, \beta, \delta$)}.\\
  $D^{(\mathrm{sc})}$.\textsc{Initialize}($G, \rr, \{s, t\}, \eps/100$, $\beta$).\\
  \For{each vertex $v \in C \setminus T$ \label{line:terminalsame1}} {
    $D^{(\mathrm{sc})}$.\textsc{PermanentAddTerminals}$(\{v\})$.\\  
  }
  \For{each vertex $v$ in $T\setminus C$ \label{line:terminalsame2}} {
    $D^{(\mathrm{proj})}$.\textsc{AddTerminal}$(v)$.\\  
  }
  $D^{(\mathrm{proj})}$.\textsc{InitProjections}(), $count\assign 0$. \label{line:initprojectionslocator} \\
  $\overline{count}\assign \delta^{-1}\eps/\polylog n$ for a large enough $\polylog n$ factor such that the error of $D^{(\mathrm{proj})}$.\textsc{Query} as in \cref{lem:projector} satisfies
  $\O(\overline{count} \cdot \delta) \leq \eps/100$.\\
    }
\end{algorithm}

We then initialize the dynamic Schur complement and projection data structures. We ensure using lines \ref{line:terminalsame1} and \ref{line:terminalsame2} that these data structures have the same terminal set. We then compute the vectors $\ppi^C(\dd[i])$ exactly to start. Finally, we set $\overline{count}$ as the number of terminal additions the projector data structure \textsc{Projector} supports before doing an exact recomputation of $\ppi^C(\dd[i])$.

Now, we present the pseudocode for adjusting edge resistances in \Locator~for both single edge updates and batched updates in \cref{algo:locatoraddterminal}. This requires adding terminals to both $D^{(\mathrm{sc})}$ and $D^{(\mathrm{proj})}$.

\begin{algorithm}[!ht]
\caption{Pseudocode for updating $\rr_e$ in \textsc{Locator},
either single, or in batches. \label{algo:locatoraddterminal}}
\SetKwProg{Globals}{global variables}{}{}
\SetKwProg{Proc}{procedure}{}{}
\Globals{}{
    $D^{(\mathrm{sc})}$: Instance of \textsc{DynamicSC} data structure of \cref{thm:dynamicsc}. \\
    $D^{(\mathrm{proj})}$: Instance of \textsc{Projector} data structure of \cref{lem:projector}. \\
    $\eps, \beta, \delta$: pamaters controlling accuracy, size of the terminal sets $C, T$ and frequency to call $D^{(\mathrm{proj})}.\textsc{InitProjections}()$.\\
    $\rr$: the resistance vector.\\
	$C$, $T$: the current terminal sets of $D^{(\mathrm{proj})}$ and $D^{(\mathrm{sc})}$.\\
	$\mQ$: heavy-hitter matrix produced by \cref{lem:heavyhitter}. \\
	$count$, $\overline{count}$: number of terminals added since last initialization, and threshold.\\
	$\Delta\dd[i]$ for $i \in [N]$: change to $\dd[i]$ from resistance changes.
}
\Proc{\textsc{Locator.UpdateD}($e, \rr^\new$) \label{line:updated}} {
  \For{each $i\in [\ndd]$ such that $\mQ_{i, e}\neq 0$} {
        \If{$\rr_e \in [\eps^2/20, 20\eps^{-2}]$}{
            $\Delta\dd[i] \assign \Delta\dd[i] - \mQ_{i,e} \rr_e^{-1/2} \cchi_e.$ \label{line:updatedelta1}
        }
        $\rr_e \assign \rr^\new.$ \\
        \If{$\rr_e \in [\eps^2/20, 20\eps^{-2}]$}{
            $\Delta\dd[i] \assign \Delta\dd[i] + \mQ_{i,e} \rr_e^{-1/2} \cchi_e.$ \label{line:updatedelta2}
        }
    }
}
\Proc{\textsc{Locator.Update}($e, \rr^\new$)}{
  \For{each $v\in e$} {
    \If{$v$ is not in $C$} {
      $D^{(\mathrm{proj})}$.\textsc{AddTerminal}($v$) \label{line:projaddterminalslocate}.\\
      $count\assign count + 1$.\\
      \If{$count = \overline{count}$ \label{line:resetcompletely}} {
        $D^{(\mathrm{proj})}$.\textsc{InitProjections}(). \label{line:remakeproj} \\
        $count\assign 0$.\\
      }
      $D^{(\mathrm{sc})}$.\textsc{PermanentAddTermimals($\{v\}$)}
    }
  }
  $D^{(\mathrm{sc})}$.\textsc{Update}$(e, \rr^\new)$.\\
  \textsc{Locator.UpdateD}($e, \rr^\new)$.\\
}
\Proc{$\textsc{Locator.BatchUpdate}(S, \rr^\new\in \mathbb{R}^S_{>0})$}{
    \For{each $(u, v) = e\in S$} {
        $D^{(\mathrm{proj})}$.$C$ = $D^{(\mathrm{proj})}$.$C \cup \{u,v\}$.\\
        $D^{(\mathrm{sc})}$.\textsc{PermanentAddTerminals}$(\{u,v\})$. \\
        $D^{(\mathrm{sc})}$.\textsc{Update}$(e, \rr^\new_e)$.\\
        \textsc{Locator.UpdateD}($e, \rr^\new_e)$.\\
    }
    $D^{(\mathrm{proj})}$.\textsc{InitProjections}().\\
}
\end{algorithm}
To implement single edge updates in \Locator.\textsc{Update}, we first add both endpoints of the edge $e$ as terminals in both $D^{(\mathrm{sc})}$ and $D^{(\mathrm{proj})}$, and then update the resistance of edge $e$ in $D^{(\mathrm{sc})}$. In the batched case \Locator.\textsc{BatchUpdate} or when $count$ has increased to value $\overline{count},$ i.e. the errors of $\ppitil^C(\dd[i])$ have accumulated in $D^{(\mathrm{proj})}$, we do a full recomputation of $\ppi^C(\dd[i])$ for the current terminals using exact Laplacian solves.

There is a subtlety in the algorithm, which is that the resistances $\mR$ change throughout the algorithm, whereas we fixed the demand vectors $\dd[i] = \oone_i^\top\mQ\mI_S\mR^{-1/2}\mB$ at the beginning. We handle this by explicitly tracking the changes to $\dd[i]$ via $\Delta\dd[i]$ in \Locator.\textsc{UpdateD}. This interacts nicely with maintaining $\ppi^C(\dd[i])$ and accumulates no errors because $\Delta\dd[i]$ is supported only on terminals $C$ by definition, so
\[ \ppi^C(\dd[i] + \Delta\dd[i]) = \ppi^C(\dd[i]) + \Delta\dd[i]. \]

Finally, we give the pseudocode for \Locator.\textsc{Locate}.

\begin{algorithm}[!ht]
\caption{Pseudocode \Locator.\textsc{Locate}, detection of large edges. \label{algo:locatorlocate}}
\SetKwProg{Globals}{global variables}{}{}
\SetKwProg{Proc}{procedure}{}{}
\Globals{}{
    $D^{(\mathrm{sc})}$: Instance of \textsc{DynamicSC} data structure of \cref{thm:dynamicsc}. \\
    $D^{(\mathrm{proj})}$: Instance of \textsc{Projector} data structure of \cref{lem:projector}. \\
    $\eps$: accuracy.\\
    $\rr$: the resistance vector.\\
	$C$, $T$: the current terminal sets of $D^{(\mathrm{proj})}$ and $D^{(\mathrm{sc})}$.\\
	$\mQ$: heavy-hitter matrix produced by \cref{lem:heavyhitter}. \\
	$count$, $\overline{count}$: number of terminals added since last initialization, and threshold.
	$\Delta\dd[i]$ for $i \in [N]$: change to $\dd[i]$ from resistance changes.
}
\Proc{$\textsc{Locator.Locate}()$}{
  Let $\vv$ be an array of length $\ndd$.\\
  $\widetilde{\mSC} \assign D^{(\mathrm{sc})}.\textsc{SC}()$, with error $\eps/100$ \label{line:samplesclocate} \\
  \For{each $i\in [\ndd]$} {
    $\vv_i\assign \left\langle D^{(\mathrm{proj})}.\textsc{Query}(i) + \Delta\dd[i], \widetilde{\mSC}^\dagger \cchi_{st}\right\rangle$. \label{line:definev} \\
  }
  Call \textsc{Recover}($\vv$) in \cref{lem:heavyhitter} and return the edges returned by \cref{lem:heavyhitter}.
}
\end{algorithm}

To implement \Locator.\textsc{Locate}, we first solve $\cchi_{st}$ on an approximate Schur complement by using the $D^{(\mathrm{sc})}$ data structure to compute $\pphi = \widetilde{\mSC}^\dagger\cchi_{st}$. We can use this to compute the vector $\vv$ which we plug into \textsc{Recover} of \cref{lem:heavyhitter}. Specifically, our $\vv$ satisfies
\begin{align*} \vv_i &= \ppitil^C(\dd[i])^\top \pphi = \ppitil^C(\dd[i])^\top \widetilde{\mSC}^\dagger \cchi_{st} \approx \ppi^C(\dd[i])^\top\mSC(G, C)^\dagger\cchi_{st} = \dd[i]^\top \mL(G)^\dagger\cchi_{st} \\ &= \oone_i^\top \mQ\mR^{-1/2}\mB\mL(G)^\dagger\cchi_{st}.
\end{align*}
We would like to point out a subtlety: the \textsc{Projector} algorithm does not know the vector $\pphi$ in advance, and additionally $\pphi$ must be generated by an oblivious adversary given the initial terminal set (see last paragraph of \cref{lem:projector}). However, this is guaranteed because the randomness used for $D^{(\mathrm{sc})}$ is separated from the randomness for $D^{(\mathrm{proj})}$, so there are no actual issues.

\begin{proof}[Proof for \cref{thm:locator}]
We show that \cref{algo:initlocator,algo:locatoraddterminal,algo:locatorlocate} together prove \cref{thm:locator}. We first show the correctness of \Locator.\textsc{Locate}, then analyze the runtimes. Throughout the proofs, we will assume oblivious adversarial inputs to all of \cref{lem:heavyhitter,lem:projector,thm:dynamicsc} as \cref{thm:locator} (\Locator) assumes an oblivious adversary.

\paragraph{Correctness.} Let $\rr$ be the original resistances, and let $\rrhat$ be the resistances right before \Locator.\textsc{Locate} is called. Let $S, \Shat$ respectively be the set of edges $e$ with $\rr_e \in [\eps^2/20, 20\eps^{-2}]$, respectively $\rrhat_e.$ Let $G, \Ghat$ be the graphs with resistances $\rr, \rrhat$ respectively. Let $\Chat$ be the final terminal set. Let $\dd[i], \ddhat[i]$ be the columns of
\[ \mQ\mI_S\mR^{-1/2}\mB, \enspace \text{ respectively } \enspace \mQ\mI_{\Shat} \mRhat^{-1/2}\mB\]
as in \cref{def:settings_of_sec_approx_proj}. Note that by construction in \Locator.\textsc{UpdateD} starting in line \ref{line:updated} and the fact that $\Delta\dd[i]$ is supported on the terminals $\Chat$, we know that \begin{align} \ppi^{\Chat}(\ddhat[i]) = \ppi^{\Chat}(\dd[i]) + \Delta\dd[i]. \label{eq:supportupdate} \end{align}

Let $\widetilde{\mSC} \approx_{\eps/100} \mSC(\Ghat, \Chat)$ as sampled in line \ref{line:samplesclocate} of \cref{algo:locatorlocate}, where the approximation follows by the guarantees of \cref{thm:dynamicsc}.

We first show that our algorithm successfully detects edges with large flow, assuming that the potentials are computing using $\widetilde{\mSC}$ as a solver, instead of the true Laplacian. Then, we transfer this to the true flows using \cref{lemma:approxenergy}.
Let $\pphi = \widetilde{\mSC}^\dagger\cchi_{st}$, and let $\vv \in \R^n$ be the vector as defined in line \ref{line:definev} of \cref{algo:locatorlocate}, i.e. 
\[ \vv_i = \left\langle \ppitil^{\Chat}(\dd[i]) + \Delta\dd[i], \pphi\right\rangle. \] Note that except for knowing $C$ of $D^{(\mathrm{proj})}$, the randomness of $D^{(\mathrm{sc})}$ is independent of that of $D^{(\mathrm{proj})}$. Thus, the vector $\pphi$ can be treated as an input to $D^{(\mathrm{proj})}$ only adapting for $C$. Additionally, note that $\pphi \in [-2, 2]^C$ as $\widetilde{\mSC}$ is a spectral sparsifier of $\mSC(\Ghat, \Chat),$ which induces a flow of energy at most $2$ by assumption.

For simplicity, define $\mL \defeq \mL(\Ghat, \Chat)$, and define the operator $\widetilde{\mL}$ as
\[ \widetilde{\mL}^\dagger = \begin{bmatrix} \mI & -\mL_{FF}^{-1} \mL_{FC} \\ 0 & \mI \end{bmatrix}
\begin{bmatrix} \mL_{FF}^{-1} & 0 \\ 0 & \widetilde{\mSC} \end{bmatrix}
\begin{bmatrix} \mI & 0 \\ -\mL_{CF}\mL_{FF}^{-1} & \mI \end{bmatrix}, \] so that $\widetilde{\mL} \approx_{\eps/100} \mL$.
Therefore, \cref{lem:projector} gives us that
\begin{align*}
&\left|\vv_i - \oone_i^\top\mQ\left(\mI_{\Shat}\mRhat^{-1/2}\mB\widetilde{\mL}^\dagger\cchi_{st}\right) \right| = \left|\left\langle \ppitil^{\Chat}(\dd[i]) + \Delta\dd[i], \pphi\right\rangle - \left\langle \ppi^{\Chat}(\ddhat[i]), \widetilde{\mSC}^\dagger\cchi_{st} \right\rangle \right| \\
&= \left|\left\langle \ppitil^{\Chat}(\dd[i]) + \Delta\dd[i], \pphi\right\rangle - \left\langle \ppi^{\Chat}(\dd[i]) + \Delta\dd[i], \pphi\right\rangle \right| \\
&= \left|\left\langle \ppitil^{\Chat}(\dd[i]) - \ppi^{\Chat}(\dd[i]), \pphi\right\rangle \right| \le \O(\overline{count} \cdot \delta) \le \eps/100.
\end{align*}
The first equality follows from the definition of $\vv_i,$ and expanding out the definition of the operator $\widetilde{\mL}^\dagger.$ The second equality comes from \eqref{eq:supportupdate}, and the final inequalities follow from the guarantees of $D^{(\mathrm{proj})}$.\textsc{Query} on the vector $\pphi \in [-2, 2]^V$, and that line \ref{line:resetcompletely} of \cref{algo:locatoraddterminal} ensures that $count \le \overline{count}.$ Additionally, we because $\widetilde{\mL} \approx_{\eps/100} \mL$ we know that
\[ \left\|\mI_{\Shat}\mRhat^{-1/2}\mB\widetilde{\mL}^\dagger\cchi_{st} \right\|_2 \le 1.1\left\|\mRhat^{-1/2}\mB\widetilde{\mL}^\dagger\cchi_{st} \right\|_2 = 1.1\cE^{1/2} = 1.1. \] Therefore, $\textsc{Recover}(\vv)$ w.h.p. returns a set of size $O(\eps^{-2})$ containing all edges $e$ with
\[ \left(\mI_{\Shat}\mRhat^{-1/2}\mB\widetilde{\mL}^\dagger\cchi_{st}\right)_e \ge \eps/10. \]
It remains to convert this guarantee to the true electric flow $\mRhat^{-1/2}\mB\mL^\dagger\cchi_{st}.$ For edges $e \in \Shat,$ \cref{lemma:badres} for $\cE = 1$ gives us that the energy is at most \[ \min(\rrhat_e, \rrhat_e^{-1}) \le \eps^2/20, \] so we do not need to return $e$.

For edges $e \notin \Shat,$ we have
\begin{align*}
&\left|\left(\mRhat^{-1/2}\mB\widetilde{\mL}^\dagger\cchi_{st}\right)_e - \left(\mRhat^{-1/2}\mB\mL^\dagger\cchi_{st}\right)_e \right| = \rrhat_e^{-1/2}\left|\cchi_e^\top\left(\widetilde{\mL}^\dagger - \mL^\dagger \right)\cchi_{st}\right| \le \eps/100,
\end{align*}
where the inequality follows from \cref{lemma:approxenergy}. Therefore, all edges of energy at most
\[ \left(\frac{\eps}{\sqrt{20}} + \frac{\eps}{100} \right)^2 < \eps^2/10 \] are returned with high probability, as desired.

\paragraph{Runtime.} We go item by item, including \Locator.\textsc{UpdateD} which we will require to bound the runtimes of \Locator.\textsc{Update} and \Locator.\textsc{BatchUpdate}.

\begin{itemize}
    \item \textsc{UpdateD}. This requires $\O(1)$ time. This is because, by the construction of $\mQ$ in \cref{lem:heavyhitter}, each edge $e$ only has $\O(1)$ nonzeros in its corresponding column. Each update in lines \ref{line:updatedelta1} and \ref{line:updatedelta2} of \cref{algo:locatoraddterminal} affects at most two coordinates of $\Delta\dd[i]$, so the total time is $\O(1).$
    \item \textsc{Initialize}. The cost of initializing $D^{(\mathrm{proj})}$ and $D^{(\mathrm{sc})}$ is $\O(m\beta^{-4}\delta^{-2}\eps^{-2})$ and $\O(m\beta^{-4}\eps^{-4})$ respectively by \cref{lem:projector,thm:dynamicsc}. As $\delta < \eps,$ the first term dominates. In addition to this, line \ref{line:multiplybyq} of \cref{algo:initlocator} requires $\O(m)$ time by the same analysis as for \textsc{UpdateD}.
    The cost of adding terminals in \ref{line:terminalsame1} and \ref{line:terminalsame2} of \cref{algo:initlocator} are dominated by the initialization time. Initializing the projections ($D^{(\mathrm{proj})}$.\textsc{InitProjections}()) requires $\O(Nm) = \O(m\eps^{-2})$ time by \cref{lem:projector}. The $\O(m\beta^{-4}\delta^{-2}\eps^{-2})$ dominates.
    \item \textsc{Update}. The first cost is calling $D^{(\mathrm{proj})}$.\textsc{AddTerminal} on at most two vertices in line \ref{line:projaddterminalslocate} of \cref{algo:locatoraddterminal}. By \cref{lem:projector} this costs $\O(\delta^{-2}\beta^{-6}\eps^{-2})$ time. Also, $D^{(\mathrm{proj})}$.\textsc{InitProjections} is called in line \ref{line:remakeproj} of \cref{algo:locatoraddterminal} when $count$ is a multiple of $\overline{count}.$ As a result, the amortized complexity of this step is $\O(\overline{count}^{-1} m\eps^{-2}) = \O(\delta m\eps^{-3})$ by the choice of $\overline{count}.$
    Finally, the calls to $D^{(\mathrm{sc})}$.\textsc{PermanentAddTerminals} and $D^{(\mathrm{sc})}$.\textsc{Update} cost $\O(\beta^{-2}\eps^{-2})$ by \cref{thm:dynamicsc}, as $c = 0$ because there are no temporary updates. Finally, $\Locator.\textsc{UpdateD}$ costs $\O(1)$ time by the above.
    The total amortized time is thus $\O(\delta m\eps^{-3} + \delta^{-2}\beta^{-6}\eps^{-2}).$
    \item \textsc{BatchUpdate}. The costs to $D^{(\mathrm{sc})}$.\textsc{PermanentAddTerminals}, $D^{(\mathrm{sc})}$.\textsc{Update}, and \Locator.\textsc{UpdateD} can be bounded as the above, and takes $O(|S|\beta^{-2}\eps^{-2})$ time. The last line $D^{(\mathrm{proj})}$.\textsc{InitProjections} takes time $\O(m\eps^{-2})$ by \cref{lem:projector}. The total time is thus $\O(m\eps^{-2} + |S|\beta^{-2}\eps^{-2}).$
    \item \textsc{Locate}. Sampling $\widetilde{\mSC}$ takes time $\O(\beta m\eps^{-2})$ by \cref{thm:dynamicsc}. Computing $\widetilde{\mSC}^\dagger\cchi_{st}$ also needs time $\O(\beta m\eps^{-2})$ time by \cref{thm:lap}, and we can store is as a length $O(\beta m)$ vector. Each vector $D^{(\mathrm{proj})}.\textsc{Query}(i)$ is length $O(\beta m)$ and takes time $O(\beta m)$ to compute by \cref{lem:projector}. Because we have precomputed $\pphi = \widetilde{\mSC}^\dagger\cchi_{st}$, computing each \[ \vv_i = \left\langle \ppitil^{\Chat}(\dd[i]) + \Delta\dd[i], \widetilde{\mSC}^\dagger\cchi_{st} \right\rangle \] requires time $O(\beta m),$ so the total time is $O(N\beta m) = \O(\beta m\eps^{-2}).$ Now, running $\textsc{Recover}(\vv)$ requires time $\O(N) = \O(\eps^{-2}).$ So the total time is $\O(\beta m\eps^{-2})$ as desired.
\end{itemize}
\end{proof}
\section{Correctness of Recentering Batch}
\label{sec:pathcorrect}

In this section we explain a single step of our interior point method that uses the data structures in \cref{sec:checker,sec:Locator} to argue that we can make $k/\sqrt{m}$ flow progress (as opposed to $1/\sqrt{m}$) in amortized $\O(m)$ time. Before introducing our algorithm and analysis we set up the method loop and show several stability properties that are essential to our analysis.

\subsection{IPM Setup and Stability}

For a (slightly modified) undirected capacitated graph $G$ with capacities $\uu \in \R^E$ we design algorithms that maintain
a parameter $\mu$ with $0 \leq \mu \leq F^{*}$
and a flow $\ff(\mu) \in \R^E$ that approximately minimizes the logarithmic barrier function $V$ as given in \eqref{eq:logbarrier}.
Here the parameter $\mu$ indicates progress along the central path.
It starts at $\mu = F^{*}$, which corresponds to the empty flow.
This flow is trivially central due to the symmetry of
capacities.
We will gradually multiplicatively decrease $\mu$ until it is
less than $1$,
at which point we can round to an integral flow that
routes a strictly larger amount~\cite{KP15:arxiv},
namely $F^{*}$.

Our algorithm uses the most standard central path. On the other hand, several other works \cite{M13,LS19,M16,LS20_STOC,KLS20, BLSS20,BLLSSSW20:arxiv,BLNPSSSW20} work with \emph{weighted barriers} or \emph{robust} central paths to achieve their runtimes. Weighted barriers refer to not weighting each edge equally in the logarithmic barrier, and in robust central path analysis the centrality error is calculated in the $\ell_\infty$ norm instead of the standard $\ell_2$.
While we also don't need to maintain the exact flow,
our main centrality error is still calculated in the $\ell_2$ norm.
One benefit of this is that we can guarantee that we can efficiently recenter back to the true central path flow any time,
which limits the amount of adaptivity that our data structures need to handle.

Our algorithms will repeatedly update a flow $\ff$ by adding an electric flow, which we will
denote using $\Delta \ff$ (or $\Delta \cdot $ more generally).
The resulting `new' vector will be denoted using $\ffhat$.
$\fftil$ will denote flows that are approximately close
to what the exact algorithm would maintain.

\subsubsection{Central Path}

The central path is defined as the sequence of flows
(for different parameters of $\mu$)
that satisfy the optimality condition in \eqref{eq:logbarrier}.
By the KKT conditions, $\ff(\mu)$ is optimal if and only if there is
a vector $\pphi(\mu) \in \R^V$ such that
\begin{align}
\pphi\left( \mu \right) _u - \pphi\left( \mu \right)_v
=
\left(\mB\pphi\left( \mu \right) \right)_e
=
\frac{1}{\uu_e -\ff\left(\mu\right)_e}
-
\frac{1}{\uu_e + \ff\left(\mu\right)_e}
\text{ for all }
e = \left(u, v \right) \in E. \label{eq:central}
\end{align}
For notational shorthand, we will use $\uu^{+}(\ff)_e = \uu_e - \ff_e$
and $\uu^{-}(\ff)_e = \ff_e + \uu_e$ to denote the upper and lower
remaining capacities of $\ff$ on edge $e$, and
$\uu(\ff)_e = \min\{\uu^{+}(\ff)_e, \uu^{-}(\ff)_e\}$
to denote the minimum residual capacity in either direction.

Given a flow $\ff$ we define the induced \emph{resistances} as
\begin{align*}
\rr\left(\ff\right)_e
=
\frac{1}{\uu^{+}\left( \ff \right)_e^2}
+
\frac{1}{\uu^{-}\left( \ff \right)_e^2}
\text{ for all } e \in E
\end{align*}
and $\mR(\ff) = \diag(\rr(\ff))$.
$\mR$ is the Hessian of $V$ as $\ff$, i.e. $\mR = \nabla^2 V(\ff)$.
Note that $\uu(\ff)_e^{-2} \le \rr(\ff)_e \le 2\uu(\ff)_e^{-2}$.
Our progress steps using electrical flows, which we described above
in Section~\ref{sec:Flows}, will be computed using approximations to
these resistances.
We note that if flows $\ff$ and $\fftil$ have similar residual capacities on an edge $e$, then their resistances on that edge are also similar. In particular, if $\rr(\ff)_e^{1/2}|\ff_e - \fftil_e| \le \eps$ for some $\eps < 1/2$, then $\rr(\ff)_e \approx_{O(\eps)} \rr(\fftil)e.$

An important piece of the analysis is tracking how close a flow $\ff$ is to satisfying \eqref{eq:central}. We define the \emph{centrality} of a flow $\ff$ and (implicit) dual variable $\pphi$ as
\begin{align*}
\norm{
\mB\pphi - \left(
\frac{\oone}{\uu^{+}\left( \ff \right)}
-
\frac{\oone}{\uu^{-} \left( \ff \right)}\right)
}
_{\mR\left(\ff\right)^{-1}},
\end{align*}
i.e. the error of \eqref{eq:central} measure in the $\mR(\ff)^{-1}$ norm.

Finally, the flows $\ff$ we maintain do not necessarily perfectly satisfy the demands $(F^*-\mu)\cchi_{st}.$
To capture this, define the Laplacian
$\mL(\ff) \defeq \mB^\top \mR(\ff)^{-1}\mB$ and define the
\emph{demand error} for a centrality parameter $\mu$ as
\begin{align*}
\norm{\mB^\top \ff - \left(F^{*} - \mu \right) \cchi_{st}}_{\mL\left(\ff\right)^\dagger},
\end{align*}
i.e. the error of $\ff$ off satisfying the demands
$(F^{*} - \mu)\cchi_{st}$ measured in the norm of inverse of the Laplacian. Putting these conditions together allows us to define a \emph{centered point}.
\begin{definition}[Centered point]
\label{def:centeredpoint}
We say that a flow $\ff$ is
\emph{$\gamma$-centered} for a parameter $\mu$ if there is a dual variable $\pphi$ such that
\[
\norm{
\mB\pphi - \left(\frac{\oone}{\uu^{+}\left( \ff\right)}
- \frac{\oone}{\uu^{-}\left(\ff\right)}\right)
}_{\mR\left(\ff\right)^{-1}}
\le
\gamma
\]
and
\[
\norm{\mB^\top \ff - \left(F^{*} - \mu \right) \cchi_{st}}
_{\mL\left(\ff\right)^\dagger} \le \gamma.
\]
\end{definition}

It is a standard fact that if a point is $\gamma$-centered then we can compute the true central path flow $\ff(\mu)$ using $\O(1)$ recentering steps. More generally, it is known that given a flow $\fftil$ whose residual capacities are within an $\ell_\infty$ ball of the residual capacities of the true central path flow $\ff(\mu),$ we can compute $\ff(\mu)$ with $\O(1)$ Laplacian system solves. We show the following in \cref{sec:proofs}.
\begin{lemma}[$\ell_\infty$ approximation implies recentering]
\label{lemma:recenter}
There is an algorithm $\Recenter(\fftil, \mu)$ which given a path parameter $\mu$ and flow $\fftil$ satisfying
\[
\norm{
\mR\left(\fftil\right)^{1/2}
\left(\fftil - \ff\left(\mu\right)\right)
}_\infty \le 1/10
\]
returns $\ff(\mu)$ exactly in $\O(m)$ time.
\end{lemma}
In reality, our algorithm can only maintain all quantities (such as flows and resistances) up to an addition $1/\poly(m)$ approximation, i.e. $\O(1)$ bits.
However, this does not affect the success of the algorithm and is discussed in more detail in previous works \cite{M13,M16,LS20_STOC,KLS20}.
This introduces a subtle issue that an adaptive adversary may be able to learn information from the rounding to additive approximation $1/\poly(m)$,
however this can resolved by picking a random rounding threshold \cite{S05}. Thus we assume throughout the paper that all Laplacian system solves and parameters maintained are exact.

We also verify in \cref{sec:proofs} that the standard notion of
approximate centrality satisfies the input condition of this recentering.
\begin{lemma}[Recentering approximate demands and centrality]
\label{lemma:recentererror}
If flow $\ff$ with path parameter $\mu$ is $1/1000$-centered as in \cref{def:centeredpoint},
then we have that
\[
\norm{
\mR\left(\ff\right)^{1/2}\left(\ff - \ff\left(\mu\right)\right)
}_\infty \le 1/10.
\]
\end{lemma}

\subsubsection{Preconditioning Arcs}
\label{subsection:precon}

As a technical point, we must work with $s$-$t$ electrical flows in graphs with preconditioning
edges added, in the same manner as~\cite{M16,LS20_STOC,KLS20}.
We say an undirected graph is \emph{preconditioned} if at least half of its edges are between $s$-$t$ with capacity $2U$ (which we refer to as
\emph{preconditioning edges}), and the remaining edges are capacity at most $U$.

To precondition an undirected graph with $m$ edges, we simply add $m$ more edges
of capacity $2U$ between $s$ and $t$, which 
increases the amount of $s$-$t$ flow by exactly $2mU$.

Preconditioning ensures that for any centered $\ff$,
the preconditioning edges route a significant fraction
of the remaining residual flow.
The following lemma was shown in \cite{LS20_STOC} Section B.5.
\begin{lemma}
\label{lemma:precon}
\label{lemma:elecenergy}
For any preconditioned graph $G$ and parameter $\mu$,
the central flow $\ff(\mu)$ satisfies
\[
\uu\left(\ff\left( \mu \right) \right)_e
\ge
\frac{\mu}{14m}
\]
for all preconditioning edges $e$. In particular, for any preconditioning edge we have
$\rr(\ff)_e \le \frac{400m^2}{\mu^2}$,
and in the graph with resistances $\rr(\ff)$,
the energy of a unit $s$-$t$ electric flow
is at most $\frac{400m}{\mu^{2}}.$
\end{lemma}

In one step of a standard IPM, $\Theta(\mu/\sqrt{m})$ units of flow are routed to decrease the residual flow by a $(1-\frac{1}{\sqrt{m}})$ factor,
and \cref{lemma:precon} ensures that the energy of this flow is $O(1).$

\subsubsection{Central Path Stability}
\label{sec:stability}

Our data structures based approach is built upon an observation that
dates back to the original interior point method by Karmarkar~\cite{K84},
and used in all subsequent faster runtime bounds~\cite{V90,LS15,CLS19,B20b,BLNPSSSW20,BLLSSSW20:arxiv}:
that the residual capacities of the central path, or the resistances,
are slowly changing. We show the necessary stability bounds that we need in this section. Proofs are deferred to \cref{sec:proofs}.

Out first bound is a general bound bounding the total change in flow versus residual capacity over $k/\sqrt{m}$ progress over the central path.
\begin{lemma}
\label{lemma:l2lemma}
Consider a preconditioned graph $G$ and flow values $t$ and $\that$ such that
\[
\muhat
=
\mu - \frac{k \mu}{\sqrt{m}}
\]
for some $1 \le k \le \sqrt{m}/10.$
Then for $\Delta\ff = \ff(\muhat) - \ff(\mu)$ we have
\[ \left[ \left( \Delta \ff \right)^{2} \right]^{\tomato}
\left[ 
\frac{\oone}{\uu^{+} \left( \ff\left( \mu \right) \right) \circ
\uu^{+} \left( \ff\left( \muhat \right) \right)}
+
\frac{\oone}{
\uu^{-} \left( \ff\left( \mu \right) \right) \circ
\uu^{-} \left( \ff\left( \muhat \right) \right)
}\right] \le 500k^2. \]
\end{lemma}

This bound in turn gives, via Markov's inequality,
an upper bound on the the number of edges
whose resistance change by more than a certain threshold
over several steps of IPM. Here, we count an edge multiple times if its residual capacity swings up and down multiple times over the course of the central path.

\begin{lemma}[$\ell_2$ change in residual capacities]
\label{lemma:l2change}
Let $G$ be a preconditioned graph, let $\mu$ be a path parameter, and define $\muhat = \left(1-\frac{T}{\sqrt{m}}\right)\mu$ for $T \le \sqrt{m}/10.$ For each edge $e$ and $\gamma < 1$, define $\change(e, \gamma)$ to be the largest integer $\tau \ge 0$ such that there are real numbers $\mu \ge \mu^{(1)} > \mu^{(2)} > \dots > \mu^{(\tau+1)} \ge \muhat$ such that for all $i \le \tau,$
\[ \rr(\ff(\mu^{(i)}))_e^{1/2}\left|\ff(\mu^{(i)})_e - \ff(\mu^{(i+1)})_e \right| \ge \gamma. \]
Then $\sum_{e \in E(G)} \change(e, \gamma) \le O(T^2\gamma^{-2}).$
\end{lemma}

A key lemma used later to prove the correctness of our IPM is
that the residual capacities are polynomially stable along the central path -- we get a multiplicative change of $O(k^2)$ over $k/\sqrt{m}$ progress.
\begin{lemma}[Central path stability]
\label{lemma:stable}
Consider a preconditioned graph $G$ and residual amounts
$\mu$ and $\muhat$ such that
\[
\muhat = \mu - \frac{k \mu}{\sqrt{m}}
\] for some $1 \le k \le \sqrt{m}/10.$
Then the resistances of $\ff(\mu)$ and $\ff(\muhat)$ are
approximated entry-wise on each edge as:
\[
\frac{1}{10^6 k^4} \rr\left(\ff\left(\muhat\right)\right)
\le
\rr\left(\ff\left(\mu\right)\right)
\le
10^6 k^4 \cdot \rr\left(\ff\left(\muhat\right)\right).
\]
\end{lemma}

Note that this immediately implies approximations in the Hessian or Laplacian matrices associated with $\mu$ and $\muhat$:
\[
\mL\left(\ff\left(\mu\right)\right)
\approx_{O\left(k^4\right)}
\mL\left(\ff\left(\muhat\right)\right).
\]
We will show this by proving that the residual capacities
on both sides are within a factor of $O(k^2)$.
\begin{proof}[Proof of \cref{lemma:stable}]
Let $\Delta \ff = \ff(\muhat) - \ff(\mu)$.
We will bound the relative change in residual capacity on an edge by
\[
\frac{\left( \Delta \ff \right)_e^2}
{\uu^{+}\left( \ff\left( \mu \right) \right)_e
\cdot
\uu^{+}\left( \ff\left( \muhat \right) \right)_e
}
=
\frac{\left( \Delta \ff \right)_e^2}
{\left( \uu_e - \ff\left( \mu \right) \right)
\cdot
\left( \uu_e - \ff\left( \muhat \right) \right)
}
.
\]
To see this, WLOG by symmetry, assume $\ff(\mu)_e > 0$, and $\ff(\muhat)_e > \ff(\mu)_e$, as the other case is similar.
This implies $\uu_e - \ff(\muhat)_e < \uu_e - \ff(\mu)_e$
Then for some $\theta > 1$, to get
\[
\frac{\uu_e - \ff\left( \mu \right)_e}{\uu_e - \ff\left(\muhat\right)_e}
>
\theta
\]
is equivalent to 
\[
\uu_e - \ff\left(\muhat\right)_e
<
\frac{1}{\theta} \left( \uu_e - \ff\left( \mu \right)_e \right)
\]
or
\[
\left( \uu_e - \ff\left(\mu\right)_e\right) - \Delta \ff_e
<
\frac{1}{\theta} \left( \uu_e - \ff\left(\mu\right)_e \right)
\]
which upon rearranging gives (assuming $\theta \ge 500$)
\[
\Delta \ff_e
\geq
\left( 1 - \theta^{-1}\right) \left( \uu_e - \ff\left( \mu \right)_e \right)
\geq
0.9 \left( \uu_e - \ff\left( \mu \right)_e \right)
\geq
0.9 \theta \left( \uu_e - \ff\left(\muhat\right)_e \right).
\]
where the last condition follows from plugging the assumption of the ratio
between the two residues back in again.

Since the ratios are all positive, the previous claim gives us the bound
\begin{multline*}
\frac{\uu_e - \ff\left(\mu\right)_e}{\uu_e - \ff\left(\muhat\right)_e}
\le 2\frac{\left( \Delta \ff \right)_e^2}
{\left( \uu_e - \ff\left( \mu \right)_e \right)
\left( \uu_e - \ff\left( \muhat \right)_e \right)}
+
2\frac{\left( \Delta \ff \right)_e^2}
{\left( \uu_e + \ff\left( \mu \right)_e \right)
\left( \uu_e + \ff\left( \muhat \right)_e \right)}\\
\le 2\left[ \left( \Delta \ff \right)^{2} \right]^{\tomato}
\left[ 
\frac{\oone}{\left(\uu-\ff\left( \mu \right) \right) \circ \left(\uu-\ff\left( \muhat \right) \right)}
+
\frac{\oone}{\left(\uu+\ff\left( \mu \right) \right) \circ \left(\uu+\ff\left( \muhat \right) \right)}
\right].
\end{multline*}
The conclusion follows from \cref{lemma:l2lemma}.
\end{proof}

\subsection{Algorithm and Main Theorem Statement}
We present our algorithm \textsc{RecenteringBatch} in \cref{algo:batchsteps} for making $k/\sqrt{m}$ progress along the central path using data structures and prove the main theorem that we need.

\begin{algorithm}[!ht]
\caption{Pseudocode for Taking a Batch of Steps Using Data Structures\label{algo:batchsteps}.}
\SetKwProg{Globals}{global variables}{}{}
\SetKwProg{Proc}{procedure}{}{}
\Globals{}{
	$k = m^{1/328}, \eps_\step = ck^{-3}, \eps_\solve = ck^{-3}, \eps = ck^{-3}\eps_\solve$ for sufficiently small constant $c$. \\
	$\ffbar^\init$: initial central path flow.\\
	$\fftil^\init$: approximate flow corresponding to resistances inside data structures. \\
	$D^\chk_i$ for $i \in [k\eps_\step^{-1}]$: $\Theta(k^4)$ distinct \Checker~data structures (Theorem~\ref{thm:checker}), one per small step. \\
	$D^\loc$: \Locator~data structure (Theorem~\ref{thm:locator}).\\
}
\Proc{$\textsc{RecenteringBatch}(\mu, \ffbar^\init, \fftil^\init, k)$}{
  $\ffhat^{(0)} \assign \ffbar^\init$, $\Ehat \leftarrow \emptyset$.\\
  $\theta \assign \frac{\eps_\step \mu}{\sqrt{m}}.$ \\
  \For{$1 \le i \le \tau \defeq k\eps_\step^{-1} = \Theta(k^4)$ \label{line:isteps}}{
    $\mu \assign \mu - \theta.$ \label{line:ffbar} \\
    $S^{(i)} \assign D^\loc.\textsc{Locate}(G)).$ \label{line:locates}\\
    \For{$e \in S^{(i)}$ \label{line:recenterz}}{
      $\gg^{(i)}_e \assign D^\chk.\textsc{Check}(e).$ \label{line:checkz} \tcp{$\gg^{(i)}_e = 0$ if $e$ is not accepted by $D^\chk_i.$}
      \If{$\gg^{(i)}_e \neq 0$}{
          $\ffhat_e^{(i+1)} \assign \ffhat_e^{(i)} + \theta\gg_e^{(i)}$, $\Ehat \leftarrow \Ehat \cup \{e\}$. \label{line:makef1}\\
          \For{$j=i+1,\ldots, \tau$}{
            $D^\chk_j.\textsc{TemporaryUpdate}(e, \rr(\ffhat^{(i+1)})_e)$,  \label{line:batchupdateres2} \\
          }
          $D^\loc.\textsc{Update}(e, \rr(\ffhat^{(i+1)})_e)$. \label{line:batchupdateres}
        }
    }
    Implicitly set $\ffhat_e^{(i+1)} \assign \ffhat_e^{(i)}$
    for all $e \notin S^{(i)}$\label{line:makef2}.
  }
  $\ffbar^\fin \assign \Recenter(\ffhat^{(\tau)}, \mu).$ \label{line:recenter3} \tcp{\cref{lemma:recenter}}
  \For{$j=1,\ldots, \tau$}{
    Repeat $D^\chk_j.\textsc{Rollback}()$ to undo all changes to $D^\chk_j$ in \textsc{RecenteringBatch}. \label{line:checkerrollback}\\
  }
  $D^\loc.\textsc{Update}(\Ehat, \rr(\fftil^\init)_{\Ehat}).$ \label{line:rollback} \\ 
  \Return $\ffbar^\fin.$
}
\end{algorithm}

The algorithm works as follows.
We first split the larger step of size $k/\sqrt{m}$ into $k\eps_\step^{-1}$ smaller steps. We use the data structures $\Locator$ which we call $D^\loc$ and a separate \Checker~for each of the $k\eps_\step^{-1}$ smaller steps, which we call $D^\chk_i$. For the $i$-th smaller step (for $i \le k\eps_\step^{-1}$) we first call $D^\loc$ (\Locator, \cref{thm:locator}) to return a set $S^{(i)}$ of at most $O(\eps^{-2})$ edges that contains all edges where we would possibly want to update the underlying flow and resistance. Now for each edge $e \in S^{(i)}$ we call the $i$-th \Checker~data structure $D^\chk_i$ to estimate the flow on $e$. Depending on what is returned we make flow updates and pass resistance updates to $D^\loc$ and the later \Checker s $D^\chk_j$ for $j > i$. We want to note that for the sake of randomness issues we discuss later in \cref{sec:proofmain} the updates to $D^\chk_i$ during this phase are \emph{temporary} and we roll them back at the end. Because we show later that \Locator~only needs to work against oblivious adversaries, the updates to $D^\loc$ can be assumed to be essentially permanent.

\begin{theorem}
\label{thm:batch}
Algorithm \textsc{RecenteringBatch} (\cref{algo:batchsteps}) takes as input preconditioned graph $G$,
desired progress parameter $k$,
upper bound on edges updated $C k^{16}$, 
along with access to
\begin{enumerate}
    \item $k\eps_\step^{-1} = O(k^4)$ distinct instantiations of $\textsc{Checker}$ as in Theorem~\ref{thm:checker} with
    threshold $\epsilon = ck^{-6}$, which we call $D^\chk_i$ for $i \in [k\eps_\step^{-1}]$, and
    \item $\textsc{Locator}$ as in Theorem~\ref{thm:locator} with $\epsilon = ck^{-6}$, which we call $D^\loc$.
\end{enumerate}
These \Checker s and \Locator~can take any choice of $\beta$ and $\delta$ as input. Also, the algorithm has access to a central path point $\ffbar^\init \defeq \ff(\mu)$
and flow $\fftil^\init$ satisfying for solve accuracy $\eps_\solve = ck^{-3}$
\[
\norm{
\mR\left(\fftil^\init\right)^{1/2}
\left(\fftil^\init-\ff\left(\mu\right)\right)
}_\infty
\le \eps_\solve/2
\]
such that the current resistances for all the $D^\chk_i$ for $i \in [k\eps_\step^{-1}]$ and $D^\loc$~are $\rr\left(\fftil^\init\right)$.
With high probability, it outputs the central path flow
\[
\ffbar^\fin \defeq \ff\left( \mu - \frac{k}{\sqrt{m}} \mu \right),
\]
changes all of $D^\chk_i$~back to their state at the start of the call to \textsc{RecenteringBatch}, and changes the resistances of $D^\loc$~back to $\rr(\fftil^\init)$. The total cost is
\begin{enumerate}
\item Calling $D^\loc.\textsc{Locate}$ $O(k\eps_\step^{-1}) = O(k^4)$ times, \label{batch:item1}
\item For all $i \in [k\eps_\step^{-1}]$, calling each $D^\chk_i.\textsc{Check}$ on $O(\eps^{-2}) = O(k^{12})$ edges. \label{batch:item2}
\item \label{batch:item3} Calling $D^\chk_i.\textsc{TemporaryUpdate}$ for all $i \in [k\eps_\step^{-1}]$ on sets of edges of total size \\ $O(k\eps_\step^{-1}\eps^{-2}) = O(k^{16})$.
\item Calling $D^\loc.\textsc{Update}$ on sets of edges of total size $O(k\eps_\step^{-1}\eps^{-2}) = O(k^{16})$, \label{batch:item4}
\item Calling $D^\chk_i.\textsc{Rollback}$ for all $i \in [k\eps_\step^{-1}]$ on a total of $O(k\eps_\step^{-1}\eps^{-2}) = O(k^{16})$ temporary edge updates, \label{batch:item5}
\item an overhead of $\O(m)$ from calling \Recenter~as in \cref{lemma:recenter}. \label{batch:item6}
\end{enumerate}
\end{theorem}

To prove this, we show a general claim that we can take several crude $s$-$t$ flow steps,
and still exactly recenter in $\O(m)$ time. This strongly uses the lemma that within $k$ central path steps,
the residual capacities change by at most a multiplicative $O(k^2)$
by \cref{lemma:stable}. We then accumulate the costs of using the data structures. We defer discussion of randomness and adaptivity within these algorithms until \cref{sec:proofmain}.
\begin{lemma}[Inductive bound on centrality and demand errors]
\label{lemma:cdmaintain}
Let
\[
\ff^{\left( 0 \right)}
\defeq
\ff\left( \mu \right)
\]
be a flow on the central path.
For $\tau \defeq k\eps_\step^{-1}$
and each $0 \leq i < \tau$, define $\ff^{(i+1)}$ to be the flow obtained from $\ff^{(i)}$
by taking a step given by $\eps_\solve$-approximate resistances and solver, that is
\[
\ff^{\left( i+1 \right)}
=
\ff^{\left( i \right)}
+
\frac{\eps_\step \mu}{\sqrt{m}}
\left(
\left( \widetilde{\mR}^{\left( i \right)} \right)^{-1}
\mB
\widetilde{\mL}^{\left( i \right) \dag}
\cchi_{st}
\right)
\]
where
\[
\widetilde{\rr}^{\left( i \right)}
\approx_{\eps_\solve}
\rr\left( \ff^{\left( i \right)} \right)
\enspace \text{ and } \enspace
\widetilde{\mL}^{\left( i \right)}
\approx_{\eps_\solve}
\mL\left( \ff^{\left( i \right)} \right).
\]
Then for all $0 \le i \le \tau$ we have
\begin{enumerate}
\item $\ff^{(i)}$ is $1/1000$-centered for path parameter
$\mu^{(i)} \defeq \mu - \frac{i\eps_\step}{\sqrt{m}}\mu$, i.e.
\begin{equation}
\norm{\mB^{\tomato} \ff^{\left( i \right)}
- \left( F^{*} - \mu^{\left(i\right)} \right) \cchi_{st}}
_{\mL \left( \ff^{\left( i \right)} \right)^{\dag}}
\leq
0.001 \label{eq:demandstuff}
\end{equation}
and there exists some $\pphi^{(i)}$ such that
\begin{equation}
\norm{\mB \pphi^{\left( i \right)}
- \left(\frac{\oone}{\uu - \ff^{\left( i\right)}}
    - \frac{\oone}{\uu + \ff^{\left( i\right)}}\right)}_{\left(\mR^{\left( i\right)} \right)^{-1}} 
\leq
0.001. \label{eq:phistuff}
\end{equation}
\item We have that
\[
\norm{\mR(\ff^{\left(i\right)})^{1/2}
\left(\ff^{\left(i\right)}
  -
 \ff\left(\mu^{\left(i\right)} \right) \right)
 }_\infty \le 0.1 \]
and
\[
\mR(\ff^{\left(i\right)})
\approx_{0.2}
\mR\left(\ff\left(\mu^{\left(i\right)}\right)\right). \]
\end{enumerate}
\end{lemma}
\begin{proof}
We proceed by induction. Note that the second item follows from the first and \cref{lemma:recentererror}, so we focus on proving the first item. Start by defining the demand error due to solver and resistance approximations in step $j$:
\[
\err^{\left(j \right)}
=
\frac{\eps_\step \mu}{\sqrt{m}}
\left(\mI
-
\mB^\tomato
\left( \widetilde{\mR}^{\left( j\right)} \right)^{-1}
\mB
\widetilde{\mL}^{\left( j \right)\dagger}
\right)
\cchi_{st}.
\]
Note that the error in demand is the sum of these per vector errors:
\[
\mB^\top \ff^{\left( i \right)}
-
\mu^{\left(i\right)} \cchi_{st}
=
\sum_{0 \leq j < i} \err^{\left(j \right)}.
\]
By the error guarantees that $\widetilde{\mL}^{(j)} \approx_{\eps_\solve} \mL(\ff^{(j)})$ and $\widetilde{\mR}^{(j)} \approx_{\eps_\solve} \mR(\ff^{(j)})$, we have
\begin{align*}
&\left\|\mI - \widetilde{\mL}^{(j) \dagger/2}\mB^\top \left( \widetilde{\mR}^{(j)} \right)^{-1} \mB\widetilde{\mL}^{(j) \dagger/2}\right\|_2
\\ &\le \left\|\mI - \widetilde{\mL}^{(j) \dagger/2}\mL(\ff^{(j)})\widetilde{\mL}^{(j) \dagger/2}\right\|_2 + \left\| \mI - \widetilde{\mL}^{\left( j \right) \dagger/2}\mB^\top\left( \left( \widetilde{\mR}^{\left( j \right)} \right)^{-1} - \left( \mR(\ff^{(j)}) \right)^{-1}\right) \mB\widetilde{\mL}^{\left( j \right) \dagger/2} \right\|_2 \\ 
&\le O(\eps_\solve).
\end{align*}

Therefore, we have that
\[
\norm{\err^{\left(j \right)}}_{\mL^{\left( j \right)^\dagger}}
\le
O\left(\frac{\eps_\solve\eps_\step \mu}{\sqrt{m}}\right)
\norm{\cchi_{st}}_{\mL^{\left( j \right)^\dagger}}
\le O(\eps_\solve\eps_\step)
\] by the energy bound of \cref{lemma:elecenergy}.
Using induction on item 2 and \cref{lemma:stable} we have that
\begin{align*}
\norm{
\sum_{0 \le j < i}
\err^{\left( j \right)}}_{\mL^{\left( i \right)^\dagger}}
&\le 
\sum_{0 \le j < i}
\norm{
\err^{\left( j \right)}}_{\mL^{\left( i \right)^\dagger}}
\le
O(k^2) \cdot 
\sum_{0 \le j < i}
\norm{
\err^{\left( j \right)}}_{\mL^{\left( j \right)^\dagger}}\\
&\le
O(k^2
\cdot
k\eps_\step^{-1}\eps_\solve\eps_\step)
\le O(k^3\eps_\solve)
\le 1/1000.
\end{align*}
for sufficiently small constant $c$ in the definition $\eps_\solve = ck^{-3}.$ This proves \eqref{eq:demandstuff}.

To show \eqref{eq:phistuff} we must bound two errors -- the first order error from resistance and solver approximations, and the second order centrality error from the linear approximation due to using electric flows which is standard in IPMs.
We denote the change with $\Delta \ff^{(j)} = \ff^{(j + 1)} - \ff^{(j)}$, and set $\pphi^{(j+1)} = \pphi^{(j)} + \Delta\pphi^{(j)}$ for
\[
\Delta\pphi^{\left( j \right)}
=
\frac{\eps_\step \mu}{\sqrt{m}}
\mL\left( \ff^{\left( j \right)} \right)^{\dagger} \cchi_{st}.
\]
Defined this way, the change in centrality error can be bounded by
\[
\cerr^{\left( j \right)}
\defeq
\mB\Delta\pphi^{(j)}
-
\left(
\frac{\oone}{\uu - \ff^{\left( j \right)} - \Delta\ff^{\left( j \right)}}
-
\frac{\oone}{\uu + \ff^{\left( j \right)} + \Delta\ff^{\left( j \right)}}
\right)
+
\left(\frac{\oone}{\uu - \ff^{\left( j \right)} }
-
\frac{\oone}{\uu + \ff^{\left( j \right)} }\right).
\]
Taylor expansion shows that the $\mR(\ff^{(j)})^{-1}$-norm of this can be bounded by
\begin{align*}
& \norm{\cerr^{\left(j \right)}}_{\mR\left(\ff^{\left( j \right)}\right)^{-1}}\\
&\le
\norm{\mB\Delta\pphi^{\left( j\right)}
-
\mR\left(\ff^{\left( j\right)}\right) \Delta \ff^{\left(j\right)}}
_{\mR\left( \ff^{\left(j\right)}\right)^{-1}}
+
5 \norm{
\left(\frac{\left( \Delta \ff^{\left( j \right) } \right)^{2}}{\left(\uu - \ff^{\left(j\right)}\right)^3}
+ \frac{\left( \Delta \ff^{\left( j \right) } \right)^{2}}{\left(\uu + \ff^{\left(j\right)}\right)^3}\right)
}_{\mR\left(\ff^{\left(j \right)}\right)^{-1}} \\
&\le
\norm{\mB\Delta\pphi^{\left( j\right)}
-
\mR\left(\ff^{\left( j\right)}\right) \Delta \ff^{\left(j\right)}}
_{\mR\left( \ff^{\left(j\right)}\right)^{-1}}
+ 10 \norm{\mR(\ff^{(j)})^{1/2}\Delta\ff^{(j)}}_4^2\\
& \le
\norm{\mB\Delta\pphi^{\left( j\right)}
-
\mR\left(\ff^{\left( j\right)}\right) \Delta \ff^{\left(j\right)}}
_{\mR\left( \ff^{\left(j\right)}\right)^{-1}}
+ 10 \norm{\mR(\ff^{(j)})^{1/2}\Delta\ff^{(j)}}_2^2\\
&\le
\norm{\mB\Delta\pphi^{\left( j\right)}
-
\mR\left(\ff^{\left( j\right)}\right) \Delta \ff^{\left(j\right)}}
_{\mR\left( \ff^{\left(j\right)}\right)^{-1}}
+
O\left(\eps_\step^2\right),
\end{align*}
where the last inequality follows by the energy bound from \cref{lemma:elecenergy}.
For the first term, we bound
\begin{align*}
&
\norm{\mB\Delta\pphi^{\left( j\right)}
-
\mR\left(\ff^{\left( j\right)}\right) \Delta \ff^{\left(j\right)}}
_{\mR\left( \ff^{\left(j\right)}\right)^{-1}}\\
& =
\frac{\eps_\step \mu}{\sqrt{m}}
\norm{
\mB\mL\left( \ff^{\left( j \right)} \right)^{\dagger} \cchi_{st}
-
\mR\left(\ff^{\left(j\right)}\right)
\left(\widetilde{\mR}^{\left( j \right)} \right)^{-1}
\mB \widetilde{\mL}^{\left( j \right) \dagger} \cchi_{st}
}_{\mR\left(\ff^{\left( j \right)}\right)^{-1}} \\
&\le
\frac{\eps_\step \mu}{\sqrt{m}}
\norm{
\mB\left(
\mL\left(\ff^{\left( j \right)}\right)^{\dagger}
-
\widetilde{\mL}^{\left( j \right)^\dagger} \right) \cchi_{st}
}_{\mR\left(\ff^{\left( j \right)}\right)^{-1}}
\\ &+ \frac{\eps_\step \mu}{\sqrt{m}}
\norm{
\left(\mI-\mR\left(\ff^{\left( j \right)}\right)
\left(\widetilde{\mR}^{\left( j \right)}\right)^{-1}
\right)
\mB\widetilde{\mL}^{\left( j \right)\dagger}
\cchi_{st}
}_{\mR\left(\ff^{\left( j \right)}\right)^{-1}} 
\\
&\le
10\frac{\eps_\step\eps_\solve \mu}{\sqrt{m}}
\norm{\cchi_{st}}_{\mL\left(\ff^{\left( j \right)}\right)^\dagger}+10\frac{\eps_\step\eps \mu}{\sqrt{m}}\le O(\eps_\step\eps_\solve).
\end{align*}
where the first term is by solver error, second term is by resistance approximation error, and last term is by the energy bound in \cref{lemma:elecenergy} again. Therefore by induction on item 2 and \cref{lemma:stable} the total centrality error can also be bounded by
\[
\norm{\sum_{0 \le j < i} \cerr^{\left( j \right)} }_{\mR\left( \ff^{\left( i \right)} \right)^{-1}}
\le O(k^2) \cdot \sum_{0 \le j < i} \norm{\cerr^{\left( j \right)} }_{\mR\left( \ff^{\left( j \right)} \right)^{-1}}
\le O(k^2 \cdot k\eps_\step^{-1} \cdot \left(\eps_\step\eps_\solve+\eps_\step^2\right))
\le
\frac{1}{1000}
\]
for sufficiently small constant $c$ due to the settings of
$\eps_\step = \eps_\solve = ck^{-3}$.
This finishes the proof of item 1, which completes the induction.
\end{proof}

Next, we combine the guarantees of \Checker~and \Locator~into a single statement to argue that the algorithm accurately estimates the flow on edges.
\begin{lemma}[Returning high energy edge estimates]
\label{lemma:chklocate}
Consider a graph $G$ with resistances $\rr$ such that the unit $a$-$b$ electric flow has energy $\cE$, and data structures $D^\chk$ (\Checker) and $D^\loc$ (\Locator) initialized with resistances $\rr$. Define $S \defeq D^\loc.\textsc{Locate}(G)$.
Then with high probability, if $\Delta\ff$ is the $s$-$t$ electric flow routing $\theta$ units, then we have for all $e \in S$ and $\gg_e \defeq D^\chk.\textsc{Check}(e)$ that
\[ \rr_e^{1/2}|\theta\gg_e - \Delta\ff_e| \le \eps \theta\sqrt{\cE}, \]
and for all $e \notin S$ we have $\rr_e^{1/2}|\Delta\ff_e| \le \eps \theta\sqrt{\cE}.$
\end{lemma}
\begin{proof}
By scaling, we assume $\theta = 1.$ Therefore, for edges $e \in S$, the conclusion follows from \cref{eq:chkguarantee} of the guarantees of \Checker~in \cref{thm:checker}.

For $e \notin S$, we need to bound their original electrical energy by $\eps\sqrt{\cE}$. Indeed, by \cref{thm:locator} we know that $D^\loc$ will return any edge that has energy at least $\eps^2\cE/10$ with high probability so $\rr_e^{1/2}|\Delta\ff_e| \le \eps \theta\sqrt{\cE}$ if $e \notin S$, as desired.
\end{proof}

We now prove that $\ffhat^{(i)}$ defined in Algorithm \textsc{RecenteringBatch} (\cref{algo:batchsteps}) lines \ref{line:makef1}, \ref{line:makef2} stays in an $\ell_\infty$ ball of a central path flow at all times.
\begin{lemma}
\label{lemma:pathstability}
Consider flows $\fftil^\init, \ffbar^\init$ satisfying the conditions of \cref{thm:batch}, and define step count $\tau = k\eps_\step^{-1}$. Let $\ffhat^{(i)}$ for $0 \le i \le \tau$ be the sequence of flows in \textsc{RecenteringBatch} (\cref{algo:batchsteps}) with inputs $\fftil^\init, \ffbar^\init$. For steps $0 \le i < \tau$,
let $\rr^{(i)}$ denote the resistances stored in $D^\chk_j$ for $j \ge i+1$ and $D^\loc$ at the start of step $i+1$ in line \ref{line:isteps} of a call to \textsc{RecenteringBatch} (\cref{algo:batchsteps}), and $\mR^{(i)}$ as the corresponding diagonal matrix to $\rr^{(i)}$.
Define the sequence of flow $\ff^{(0)} \defeq \ffbar^\init = \ff(\mu),$ and for $0 \le i < \tau$
\[
\ff^{(i+1)} = \ff^{(i)} + \frac{\eps_\step \mu}{\sqrt{m}}(\mR^{(i)})^{-1}\mB(\mB^\top(\mR^{(i)})^{-1}\mB)^\dagger \cchi_{st}.
\]
Then for all $0 \le i \le \tau$ we have that:
\begin{enumerate}
    \item the tracked flow values are close to the true implicitly
    updated flow values:
    \begin{align}
    \rr\left(\ff^{(i)}\right)^{1/2}_e
    \abs{\ff^{\left(i\right)}_e - \ffhat^{\left(i\right)}_e}
    \le
    \frac{\eps_\solve}{2},
    \qquad
    \forall e
    \label{eq:ffbarclose}
    \end{align}
    \label{item:InductionFlowDistance}
    \item The resistances underlying $D^\chk_j$ and $D^\loc$ at each step are close to the resistances
    induced by $\ff^{(i)}$:
    \[
    \rr^{\left(i\right)}
    \approx_{\eps_{\solve}}
    \rr\left(\ff^{\left(i\right)}\right)
    \]
    \label{item:InductionResistanceDistance}
    \item For $\mu^{(i)} \defeq \mu - \frac{i\eps_\step}{\sqrt{m}}\mu$ we have that $\ff^{(i)}$ is $1/1000$-centered for path parameter $\mu^{(i)}$.
    \label{item:InductionCentrality}
    \item We have $\mR(\ff^{(i)}) \approx_{0.2} \mR(\ff(\mu^{(i)})).$
    \label{item:InductionCentralResistance}
\end{enumerate}
\end{lemma}

\begin{proof}
We prove the conclusion by strong induction on $i$.

The base case follows from $\ff^{(0)} = \ffbar^\init$.
For the inductive case, assume the above hypotheses hold for all
earlier iterations, $j < i$.

\paragraph{
Item~\ref{item:InductionFlowDistance}: flow values are close.}

Define the ``true'' electric flow at step $i$ as
\[
\Delta\ff^{\left(i\right)}
\defeq
\frac{\eps_\step \mu}{\sqrt{m}}\left(\mR^{\left(i\right)}\right)^{-1}
\mB\left(\mB^\top\left(\mR^{\left(i\right)}\right)^{-1}\mB\right)^\dagger
\cchi_{st}.
\]

For the purpose of analysis, extend the vector $\gg^{(j)}$ defined in line \ref{line:checkz} of algorithm \textsc{RecenteringBatch} (\cref{algo:batchsteps}) to all coordinates.
That is, we create $\gg^{(j)} \in \R^E$ with $\gg^{(j)}_{S^{(j)}}$
the same as returned by the checker, and
$\gg^{(j)} = 0$ if $e \notin S^{(j)}$ at step $j$.
Then the error breaks down into
\begin{align*}
\rr\left(\ff^{\left(i\right)}\right)^{1/2}
\abs{\ff^{\left(i\right)}_e - \ffhat^{\left(i\right)}_e}
&\le
\sum_{j < i}
\rr(\ff^{\left(i\right)})^{1/2}
\abs{\Delta\ff^{\left(j\right)}_e - \gg^{\left(j\right)}_e
}\\
&\le \left(\max_{j < i} \sqrt{\frac{\rr(\ff^{\left(i\right)})}
{\rr(\ff^{\left(j\right)})}}\right)
\sum_{j < i} \rr\left(\ff^{\left(j\right)}\right)^{1/2}
\abs{\Delta\ff^{\left(j\right)}_e - \gg^{\left(j\right)}_e},
\end{align*}
where the last inequality is via direct manipulation.

Incorporating the ratio to true central path resistances given by
Item~\ref{item:InductionCentralResistance} of the inductive hypothesis
then gives the previous expression is
\[
\le O\left(\max_{j < i} \sqrt{\frac{\rr(\ff(\mu^{\left(i\right)}))}{\rr(\ff(\mu^{\left(j\right)}))}}\right) \sum_{j < i} \rr(\ff^{\left(j\right)})^{1/2} \abs{\Delta\ff^{\left(j\right)}_e - \gg^{\left(j\right)}_e}.
\]
Applying the stability of central path resistances
from \cref{lemma:stable} shows the previous expression is
\[
\le O(k^2) \cdot \sum_{j < i} \rr(\ff^{(j)})^{1/2} \abs{\Delta\ff^{(j)}_e - \gg^{(j)}_e}.
\]
By \cref{lemma:precon},
we can set $\theta = \frac{\eps_\step \mu}{\sqrt{m}}$
and $\cE = \frac{400m}{\mu^2}$ as the parameters of \cref{lemma:chklocate},
which gives an error of $\epsilon \theta \sqrt{\cE}$ per step,
for a total of
\[
\le O\left(k^2\right) \cdot k\eps_\step^{-1} \cdot
O\left( \eps \theta\sqrt{\cE}\right)
\le O\left(k^2 \cdot k\eps_\step^{-1} \cdot \eps \theta\sqrt{\cE}\right)
= O\left(k^3\eps\right)
\le \eps_\solve/2
\]
for small enough constant $c$ in the definition $\eps = ck^{-3}\eps_\solve$ of \textsc{RecenteringBatch} (\cref{algo:batchsteps}).

\paragraph{
Item~\ref{item:InductionResistanceDistance}: resistances are close.}

We show this for each edge $e$.
Because the initial flows $\ffbar$ and $\fftil$ may also differ,
there are two cases to consider: whether $e$ has already been updated
in some previous iteration, or whether it has carried over from the
very start.
We consider these cases separately:
\begin{itemize}
\item If $e \in S^{(j)}$ for some $j < i$,
then in steps after $j$, we have by line \ref{line:batchupdateres2} and line \ref{line:batchupdateres} of \textsc{RecenteringBatch} (\cref{algo:batchsteps}) that
\[
\rr^{\left(i\right)}_e
=
\rr\left(\ffhat^{\left(i\right)}_e\right)
\]
 so the claim follows the error between $\ffhat$ and $\ff$
 given by Item~\ref{item:InductionFlowDistance}.
\item 
The other case is that edge $e$ is never involved in a resistance
update in Line~\ref{line:recenterz}, i.e.
\[
\rr^{\left(i\right)}_e
= \rr\left(\fftil^\init\right)_e.
\]
In this case, the relative change on $e$
at each step $j < i$ is at most
\[
4 \cdot \rr\left(\ff^{\left(j\right)}\right)_e^{1/2}
\abs{\Delta\ff^{\left(j\right)}_e}
\leq
8 \left(\rr^{\left(j\right)}\right)^{1/2}_e
\abs{\Delta\ff^{\left(j\right)}_e}
\]
where the last inequality follows from Item~\ref{item:InductionResistanceDistance}
of the inductive hypothesis.
Applying~\cref{lemma:chklocate} as above allows us to bound this by
\[
\leq
O\left( \eps \theta\sqrt{\cE} \right)
= O(\eps \cdot \eps_\step).
\]
Therefore, the total multiplicative resistance change over all steps is bounded by
\[
O\left(\eps \cdot \eps_\step \cdot k\eps_\step^{-1}\right)
=
O\left(k\eps\right)
\le \frac{\eps_\solve}{3}
\]
for small enough constant $c$ in the definition $\eps = ck^{-3}\eps_\solve$.
As we initially had
\[
\rr\left(\fftil^\init\right)^{1/2}_e
\abs{\fftil^\init_e - \ffhat^{\left(0\right)}}
=
\rr\left(\fftil^\init\right)^{1/2}_e
\abs{\fftil^\init_e - \ff\left(\mu \right)}
\le \eps_\solve/2,
\]
the total multiplicative approximation between resistances
is $\eps_\solve$, as desired.
\end{itemize}

It remains to show that these imply the last two inductive hypothesis
for $i$:
\begin{enumerate}
    \item The inductive hypothesis gives that the resistances in steps
    $1 \ldots i$ satisfy the requirements of \cref{lemma:cdmaintain},
    so by it, we get Item~\ref{item:InductionCentrality}:
    that $\ff^{(i)}$ is $1 / 1000$-centered.
    \item \cref{lemma:recentererror} then gives that
    this centrality error in turn implies all resistances are
    close to true ones, giving Item~\ref{item:InductionCentralResistance}.
\end{enumerate} 
Thus the inductive hypothesis holds for step $i$ as well.

\end{proof}

We can combine these pieces and analyze the runtime costs to show \cref{thm:batch}.
\begin{proof}[Proof of \cref{thm:batch}]
To show that $\ffbar^\fin = \ff(\mu - k\mu / \sqrt{m})$,
it suffices to combine \cref{lemma:pathstability} \eqref{eq:ffbarclose} and \cref{lemma:recenter}, as $\ffbar^\fin$ is computed by recentering $\ffhat^{(\tau)}$ in line \ref{line:recenter3}.
Additionally, the state of the $D^\chk_i$~is rolled back to the original state in \ref{line:checkerrollback} and resistances of $D^\loc$~are updated to $\rr(\fftil^\init)$ in line \ref{line:rollback} of Algorithm \textsc{RecenteringBatch} (\cref{algo:batchsteps}).

To complete the proof of \cref{thm:batch} it suffices to analyze the total costs. We do this by items, as in \cref{thm:batch}.
\begin{enumerate}
    \item $D^\loc.\textsc{Locate}(G)$ is called once in line \ref{line:checkz} per each step of the while loop starting in line \ref{line:isteps}. This is $O(k\eps_\step^{-1}) = O(k^4)$ times.
    \item Each call to $D^\loc.\textsc{Locate}(G)$ returns a set $S^{(i)}$ of size $O(\eps^{-2})$ by the guarantees in \cref{thm:locator}. Hence $D^\chk_i.\textsc{Check}$ is called $O(\eps^{-2})$ edges for each $i \in [k\eps_\step^{-1}].$
    \item As above, the sets $S^{(i)}$ in line \ref{line:locates} is $O(\eps^{-2})$. Each edge contributes one update in lines \ref{line:batchupdateres2} per step for $k\eps_\step^{-1}$ steps, for a total of $O(k\eps_\step^{-1}\eps^{-2}) = O(k^{16})$ calls to $D^\chk_j.\textsc{TemporaryUpdate}$ per $j$.
    \item By the same discussion as the previous item, $D^\loc.\textsc{Update}$ is called in line \ref{line:batchupdateres} on $O(k\eps_\step^{-1}\eps^{-2})$ total edges. Additionally, $D^\loc.\textsc{Update}$ was called again in line \ref{line:rollback} the undo the same $|\Ehat| = O(k\eps_\step^{-1}\eps^{-2})$ edges.
    \item Rollbacks happen to edges that were updated in line \ref{line:batchupdateres2}, which has size $O(k\eps_\step^{-1}\eps^{-2}) = O(k^{16})$.
    \item Line \ref{line:recenter3} uses $\O(m)$ time by \cref{lemma:recenter}.
\end{enumerate}
\end{proof}
\section{Finding the Maxflow with Batched Steps}
\label{sec:proofmain}

In this section, we use recentering batches in \cref{thm:batch} to build a method to solve the maxflow problem. Then we analyze the total costs to prove the algorithm runs in time $\O(m^{\frac32-\frac{1}{328}})$, showing \cref{thm:main}.
The pseudocode, written as the decision version
(via the standard binary search reduction)
of checking if $F$ units of flow can be routed is in Algorithm~\ref{algo:main} which calls the batched steps in algorithm \textsc{RecenteringBatch} in \cref{algo:batchsteps} from \cref{sec:pathcorrect}.

In \cref{algo:main} below we assume that the thresholds in lines \ref{line:isteps3}, \ref{line:isteps1}, \ref{line:isteps2} of \cref{algo:main} satisfy: $c\eps_\solve\sqrt{{\betaloc} m}/k$ is a multiple of $c\eps_\solve\sqrt{\betachk} m/k$, and this is a multiple of $c\eps_\solve\sqrt{\delta^{-1}\eps}/k$. Also assume these are all integers.

\begin{algorithm}[!ht]
\caption{Pseudocode for Augmenting Electrical Flow Using Batched Recentering\label{algo:main}.}
\SetKwProg{Globals}{global variables}{}{}
\SetKwProg{Proc}{procedure}{}{}
\Globals{}{
  $\ffbar \leftarrow 0$, $\mu \leftarrow F$: Exact central path flows and residual flow. \\
  $k \assign m^{1/328}, \eps_\step \assign ck^{-3}, \eps_\solve \assign ck^{-3}, \eps \assign ck^{-3}\eps_\solve = \Theta(k^{-6})$. \\
  $\betaloc \assign k^{-16}, \betachk \assign k^{-28}, \delta \assign k^{-38}.$\\
  $i\assign 0$: step counter. \\
  $D^\chk_j$ for $j \in [k\eps_\step^{-1}]$: $\Theta(k^4)$ distinct \Checker~data structures (Theorem~\ref{thm:checker}), one per small step. Initialized with resistances $\rr(\zzero)$. \\
  $D^\loc$: \Locator~data structure (Theorem~\ref{thm:locator}). Initialized with resistances $\rr(\zzero).$\\
}
\Proc{$\textsc{FindFeasibleFlow}(G, \uu, F)$}{
  \While{$\mu > 1$ \label{line:startmainwhile}}{
    $\ffbar \assign \textsc{RecenteringBatch}(\mu, \ffbar, \fftil, k)$. \label{line:recenterbatch} \\
    $\mu \assign \mu - \frac{k \mu}{\sqrt{m}}.$ \\
    \If{$i$ is a multiple of $c\eps_\solve\sqrt{\betaloc m}/k = \Theta(k^{-12}\sqrt{m})$ \label{line:isteps3}}{
        $\fftil \leftarrow \ffbar.$\\
        $D^\loc.\textsc{Initialize}(G, \rr(\ffbar), \eps, \betaloc, \delta).$ \label{line:initializelocator} \\
        For $j=1, \ldots, k\eps_\step^{-1}$ do $D^\chk_j.\textsc{Initialize}(G, \rr(\ffbar), \eps, \betachk, \delta).$     \label{line:initializechecker} \\
        }
    \ElseIf{$i$ is a multiple of $c\eps_\solve\sqrt{\betachk m}/k = \Theta(k^{-18}\sqrt{m})$ \label{line:isteps1}}{
        $Z \assign \{e : \rr(\fftil)^{1/2}|\fftil_e - \ffbar_e| \ge \eps_\solve/8 \}$. \label{line:slack8} \\
        Assign $\fftil_e \assign \ffbar_e$ for all $e \in Z$. \\
        $D^\loc.\textsc{BatchUpdate}(Z, \rr(\ffbar)_Z).$         \label{line:batchupdatelocator1} \\
        For $j=1, \ldots, k\eps_\step^{-1}$ do $D^\chk_j.\textsc{Initialize}(G, \rr(\ffbar), \eps, \betachk, \delta).$     \label{line:initializechecker2} \\
    } \ElseIf{$i$ is a multiple of $c\eps_\solve\sqrt{\delta^{-1}\eps}/k = \Theta(k^{12})$ \label{line:isteps2}}{
        $Z \assign \{e : \rr(\fftil)^{1/2}|\fftil_e - \ffbar_e| \ge \eps_\solve/4 \}$. \label{line:slack4} \\
        $D^\loc.\textsc{BatchUpdate}(Z, \rr(\ffbar)_Z).$ \label{line:batchupdatelocator2} \\
        \For{$e \in Z$}{
            $\fftil_e \assign \ffbar_e$. \\
            For $j=1, \ldots, k\eps_\step^{-1}$  do $D^\chk_j.\textsc{Update}(e, \rr(\ffbar)_e).$ \label{line:updatechecker} \\
        }
    } \Else{
        $Z \assign \{e : \rr(\fftil)^{1/2}|\fftil_e - \ffbar_e| \ge \eps_\solve/2 \}$. \label{line:slack2}\\
        \For{$e \in Z$}{
            $\fftil_e \assign \ffbar_e.$\\
            $D^\loc.\textsc{Update}(e, \rr(\ffbar)_e)$.\label{line:updatelocator}\\
            For $j=1, \ldots, k\eps_\step^{-1}$ do
                $D^\chk_j.\textsc{Update}(e, \rr(\ffbar)_e).$ \label{line:updatechecker2} \\
        }
    }
    $i\assign i+1.$\\
    
  }
  Round $\ffbar$ to an integral flow and run $O(1)$ rounds of augmenting paths to finish. \label{line:round}
}
\end{algorithm}
For a holistic understanding of the algorithm,
all key parameters in it, and lower level function calls
(to both the recentering batch, and the data structures),
are listed in the appendix (with polylog factors omitted)
in Table~\ref{tab:parameters}.

At a high level \cref{algo:main} uses the \textsc{RecenteringBatch} procedure of \cref{thm:batch} and \cref{algo:batchsteps} a total of $\O(\sqrt{m}/k)$ times to compute the optimal flow. After one of these steps the algorithm decides how to update the internal resistances of the $D^\chk_j$~and $D^\loc$~data structures. To ensure that not too many changes are passed to these data structures, the algorithm essentially only changes edges whose resistances have changed by more than $\Omega(\eps_\solve)$ multiplicatively from the last change, but the precise thresholds for changing the resistances depend on the total of updates the data structures $D^\loc$~and $D^\chk_j$~have received. For example, when $D^\loc$~has received $\betaloc m$ terminal updates we rebuild both $D^\loc$~and $D^\chk_j$~and reset their resistances to the exact resistances induced by the central path flow $\ffbar$. Additionally, there are intermediate thresholds for rebuilding the $D^\chk_j$~data structures only, and for \textsc{BatchUpdate} calls to $D^\loc$~because it has a high amortized runtime.

Our correctness proof will use the stability lemmas along the central path proven in \cref{sec:stability}. Also, we observe that in \textsc{FindFeasibleFlow} (\cref{algo:main}) the flows $\ffbar$ and $\fftil$ are deterministic.
\begin{obs}
\label{obs:deterministic}
At any time of the execution of \textsc{FindFeasibleFlow} (\cref{algo:main}) the flows $\fftil, \ffbar$, and the resistances in $D^\chk_j$ for all $j$ are deterministic.
\end{obs}
\begin{proof}
The flow $\ffbar$ in line \ref{line:recenterbatch} is deterministic because it is on the central path (\cref{thm:batch}) and the value of $\mu$ is deterministic. The updates to $\fftil$ depend only on $\ffbar$ during the time steps, so $\fftil$ is also deterministic. The determinism of resistances in $D^\chk_j$~follows because \cref{thm:batch} tells us that the resistance and terminal changes to $D^\chk_j$~during the batched steps in \textsc{RecenteringBatch} (\cref{algo:batchsteps}) are rolled back via $D^\chk_j$.\textsc{Rollback}.
\end{proof}

Towards proving the correctness of \textsc{FindFeasibleFlow} (\cref{algo:main}) we first claim that the necessary conditions for \Checker~(\cref{thm:checker}) and \Locator~(\cref{thm:locator}) are satisfied.
\begin{lemma}[Conditions of \Checker~and \Locator]
\label{lemma:dscorrect}
Throughout an execution of \textsc{FindFeasibleFlow} in \cref{algo:main} and its calls to Algorithm \textsc{RecenteringBatch} in \cref{algo:batchsteps}, the conditions of every instantiation of \Checker~(\cref{thm:checker}) and \Locator~(\cref{thm:locator}) are satisfied. In particular, 
\begin{enumerate}
    \item Each instantiation of $D^\chk_j$~receives at most $O(\betachk m)$ edge updates.
        \label{item:num_of_updates_checker}
    \item Each instantiation of $D^\loc$~receives at most $O(\betaloc m)$ edge updates.
        \label{item:num_of_updates_locator}
    \item The calls to each instantiation of $D^\chk_j$~are made by an oblivious adversary.
        \label{item:obliviousness_of_adversary_checker}
    \item The calls to each instantiation of $D^\loc$~are made by an oblivious adversary.
        \label{item:obliviousness_of_adversary_locator}
\end{enumerate}
\end{lemma}

We prove the items separately.

\begin{proof}(of Lemma~\ref{lemma:dscorrect}
Item~\ref{item:num_of_updates_locator},
number of edge updates for $D^\loc$)

$D^\loc$~is initialized in line \ref{line:initializelocator} in \textsc{FindFeasibleFlow} (\cref{algo:main}). We focus on the number of updates between two executions of line \ref{line:initializelocator} which reinitialize $D^\loc$.

We first show that there are at most $\betaloc m$ total edge updates. Theorem \ref{thm:batch} tells us that the number updates to $D^\loc$~resulting from calls to Algorithm \textsc{RecenteringBatch} in \cref{algo:batchsteps} is bounded by $O(k^{16})$ times $c\eps_\solve\sqrt{\betaloc m}/k$ calls to \textsc{RecenteringBatch} between two reinitializations of $D^\loc$.
This is at most
\[ O(k^{16}) \cdot c\eps_\solve\sqrt{\betaloc m}/k = \sqrt{m}k^4 \le k^{-16}m = \betaloc m \]
by the choice of $\eps_\solve, \betaloc$, and $k$.

Other than this, the remaining edge updates happen in lines \ref{line:batchupdatelocator1}, \ref{line:batchupdatelocator2}, \ref{line:updatelocator} of \textsc{FindFeasibleFlow} (\cref{algo:main}).
Note that these updates happen only when $\rr(\fftil)^{1/2}|\fftil_e - \ffbar_e|$ is at least $\eps_\solve/8$ for some edge $e$.
Between two reinitializations of $D^\loc$, the path parameter $\mu$ decreases by at most
\[
O\left(\frac{c\eps_\solve\sqrt{\betaloc m}\mu}{\sqrt{m}}\right)
=
O\left( c \eps_\solve \sqrt{\betaloc} \mu\right),
\]
hence \cref{lemma:l2change} for
$T = c\eps_\solve\sqrt{\betaloc m}$ and $\gamma \assign \eps_\solve/8$ shows that at most $O(\betaloc m)$ edges have their resistances changed as desired.

\end{proof}

\begin{proof}(of Lemma~\ref{lemma:dscorrect}
Item~\ref{item:num_of_updates_checker},
number of edge updates for $D^\chk_j$)

We consider $D^\chk_j$ for a fixed $j \in [k\eps_\step^{-1}]$. Again by \cref{thm:batch}, the number of updates to $D^\chk_j$~resulting from calls to Algorithm \textsc{RecenteringBatch} in \cref{algo:batchsteps} is bounded by 
\[
\sqrt{m}k^4=O(\betachk m).
\]

The remaining edge updates happen in line \ref{line:updatechecker}. These updates happen only when $\ffbar^{(init)}_e$ and $\ffbar_e$ differ by at least $\eps_\solve/10$ where $\ffbar^{(init)}$ denotes the value of $\ffbar$ when the last time $D^\chk_j$~reinitializes: Right before a reinitialization of $D^\chk_j$, we ensure that for each edge $e$, $\ffbar^{(init)}_e$ and $\fftil_e$ agree up to $1\pm \eps_\solve / 8$. Note that when we reinitialize $D^\chk_j$~with the updated values, no edge updates in $D^\chk_j$~happens. Then before the next reinitialization, when line \ref{line:updatechecker} is executed, for each edge $e$ in $Z$ we have $\fftil_e$ and $\ffbar_e$ differ by at least $\eps_\solve/4$. Hence $\ffbar_e$ must have changed by say at least $\eps_\solve/10$ from $\ffbar^{(init)}_e$. As reinitialization happens every $c\eps_\solve\sqrt{\betachk m}/k$ steps, by \cref{lemma:l2change} for $T=c\eps_\solve\sqrt{\betachk m}/k$ and $\gamma=\eps_\solve/10$, the total number of edge updates for one instantiation of $D^\chk_j$~ is 
\[
O\left(\left(c\eps_\solve\sqrt{\betachk m}/k\right)^2(\eps_\solve)^{-2}\right)=O(\betachk m).
\]
\end{proof}
Since $D^\chk_j$.\textsc{Check}$()$ rolls back after checking each edge (Algorithm \ref{algo:checker}) to ensure that a call to $D^\chk_j$.\textsc{Check}$()$ does not change the state of $D^\chk_j$ (see \cref{thm:checker}), we have the following.
\begin{claim}
\label{lem:independent_check}
For any set of edges $Z$, whether some edge $e$ is returned by $D^\chk_j$.\textsc{Check}$(Z)$ and the returned flow value $\gg_e$ do not depend on $Z\setminus \{e\}$ (i.e., other edges in $Z$).
\end{claim}
\begin{claim}
\label{lem:check_do_not_modify}
$D^\chk_j$.\textsc{Check}$()$ does not modify $D^\chk_j$.
\end{claim}
\begin{proof}[Proof of Lemma~\ref{lemma:dscorrect},
Item~\ref{item:obliviousness_of_adversary_locator}:
obliviousness of the adversary to $D^\loc$]

We define the adversary of $D^\loc$~independent of the output of $D^\loc$. Let $Q$ be the sequence of queries to $D^\loc$~throughout an execution of Algorithm \ref{algo:main}. Suppose we have an alternate implementation of $D^\loc$~that always return $E$ (the set of all edges) for \textsc{Locate}$()$ and ignores other kinds of queries. We call this implementation \TrivialLocator. Let $\widehat{Q}$ be the sequence of queries to \TrivialLocator~through out an execution of Algorithm \ref{algo:main} with $D^\loc$~replaced by \TrivialLocator. Since the output of \TrivialLocator~is independent of the output of $D^\loc$, $\widehat{Q}$ is independent of the output of $D^\loc$~as well. We will show that when the two executions use the same randomness for $D^\chk_j$, $\widehat{Q}=Q$ with high probability. Thus, $D^\loc$~is ran against an oblivious adversary.

Before the first call of \textsc{Locate}$()$ (in \cref{algo:batchsteps}), the two executions are identical except for $D^\loc$~and \TrivialLocator. In particular, $Q$ is identical to $\widehat{Q}$. We will prove that the $(Z^{(i)}, \gg^{(i)} \in \R^{Z^{(i)}})$ returned by $D^\chk_i.\textsc{Check}(D^\loc.\textsc{Locate}(G))$ on line \ref{line:checkz} of Algorithm \ref{algo:batchsteps} is equal to $(\widehat{Z}^{(i)}, \widehat{\gg}^{(i)})\defeq D^\chk_i.\textsc{Check}(E)$ with high probability. This follows from the guarantees of $D^\loc$.\textsc{Locate}$()$ and $D^\chk_j$.\textsc{Check}$()$.
By \cref{lem:independent_check}, it is enough to prove that $Z^{(i)}=\widehat{Z}^{(i)}$. $\gg^{(i)}=\widehat{\gg}^{(i)}$ follows from $Z^{(i)}=\widehat{Z}^{(i)}$. 

$Z^{(i)}\subseteq \widehat{Z}^{(i)}$: For each edge $e\in Z^{(i)}$, by \cref{lem:independent_check} and $e\in E$, $e$ is in $\widehat{Z}^{(i)}$ as well. 

$\widehat{Z}^{(i)}\subseteq Z^{(i)}$: For each $e\in \widehat{Z}^{(i)}$, by \cref{thm:checker}, the energy of $e$ is at least $\eps^2\cE/2$. By \cref{thm:locator}, $e\in D^\loc.\textsc{Locate}()$ w.h.p. Conditioning on this, by \cref{lem:independent_check}, $e \in Z^{(i)}$. 

Thus, w.h.p., $Z^{(i)}=\widehat{Z}^{(i)}$. When they are equal, we have that the two executions are identical (which implies a prefix of $Q$ is identical to a prefix of $\widehat{Q}$) until the next call of \textsc{Locate}$()$ because of \cref{lem:check_do_not_modify}. By induction, $Q$ is equal to $\widehat{Q}$.

\end{proof}

\begin{proof}(of Lemma~\ref{lemma:dscorrect}
Item~\ref{item:obliviousness_of_adversary_checker},
obliviousness of the adversary to $D^\chk_j$)

Since \textsc{RecenteringBatch} (\cref{algo:batchsteps}) undoes all changes to $D^\chk_j$ before it finishes, we may consider \textsc{RecenteringBatch} and \textsc{FindFeasibleFlow} (\cref{algo:main}) independently. In \textsc{FindFeasibleFlow}, the updates are decided by $\fftil$ and $\ffbar$ which are deterministic (\cref{obs:deterministic}). In \textsc{RecenteringBatch}, each $D^\chk_j$~answers only $1$ query after which it no longer receive any update or query.
\end{proof}

Towards analyzing the runtime of \textsc{FindFeasibleFlow} (\cref{algo:main}) we bound the number of times that $D^\chk_j$~and $D^\loc$~must be reinitialized.
\begin{lemma}
\label{lemma:rebuild}

In a call to Algorithm \textsc{FindFeasibleFlow} in \cref{algo:main}, $D^\chk_j$ is reinitialized $\O(\eps_\solve^{-1}\betachk^{-1/2})$ times for each $j\in [k\eps_\step^{-1}]$. $D^\loc$~is reinitialized $\O(\eps_\solve^{-1}\betaloc^{-1/2})$ times.
\end{lemma}
\begin{proof}
Between two reinitializations of $D^\chk_j$, $\mu$ decreases by \[
1 - \frac{c\eps_\solve\sqrt{\betachk m}}{\sqrt{m}}
=
1 - \Omega\left(\eps_\solve\betachk^{1/2}\right). 
\] Then after $\O(\eps_\solve^{-1}\betachk^{-1/2})$ reinitializations, the amount of residual flow is less than $1$, as we are assuming that $U = \poly(m).$ 

Similarly, between two reinitializations of $D^\loc$, $\mu$ decreases by $1-\Omega\left(\eps_\solve\betaloc^{1/2}\right)$.
Thus, $D^\loc$~is reinitialized $\O(\eps_\solve^{-1}\betaloc^{-1/2})$ times.

\end{proof}
In our $D^\loc$~data structure, we are balancing the cost of batched and single updates. This requires a slightly finer control on the number of updates to $D^\loc$~in \textsc{FindFeasibleFlow} (\cref{algo:main}) and in \textsc{RecenteringBatch}.
\begin{lemma}
\label{lemma:nobatch}
Between two reinitializations of $D^\loc$~on line \ref{line:initializelocator} in \textsc{FindFeasibleFlow} (\cref{algo:main}), $D^\loc.\textsc{Update}$ is called on at most
$O(\sqrt{\betaloc\delta^{-1}\eps m})$ total edges, including the ones called within $\textsc{RecenteringBatch}$ in line \ref{line:recenterbatch}.
\end{lemma}
\begin{proof}
There are two such costs, one within calls to Algorithm \textsc{RecenteringBatch} in line \ref{line:recenterbatch}, and another in line \ref{line:updatelocator}.

By \cref{thm:batch} the cost within calls to Algorithm \textsc{RecenteringBatch} can be bounded by
\[ O\left(k^{16} \cdot \frac{\eps_\solve\sqrt{\betaloc m}}{k} \right) = O(k^4\sqrt{m}) \le O(k^8\sqrt{m}) = \sqrt{\betaloc\delta^{-1}\eps m}.\]

To bound the cost of line \ref{line:updatelocator} in \textsc{FindFeasibleFlow} (\cref{algo:main}), note that in line \ref{line:slack8} and line \ref{line:slack4} $\fftil_e$ and $\ffbar_e$ have residual capacities agreeing up to $1 \pm \eps_\solve/4$ for every edge $e$. However, in line \ref{line:updatelocator} of \textsc{FindFeasibleFlow} (\cref{algo:main}) we update $\fftil_e$ if they differ by $\eps_\solve/2.$ Hence an update happens only if the residual capacity of $\ffbar_e$ changed by say at least $\eps_\solve/6$ within two consecutive times when $i$ was a multiple of $c\eps_\solve\sqrt{\delta^{-1}\eps}/k$. As one of line \ref{line:isteps2} and line \ref{line:isteps1} of \textsc{FindFeasibleFlow} (\cref{algo:main}) occurs every $\eps_\solve\sqrt{\delta^{-1}\eps}/k$ steps, by \cref{lemma:l2change} for $T = \eps_\solve\sqrt{\delta^{-1}\eps}$ and $\gamma = \eps_\solve/6$ the total number of changes is at most
\[ O\left(\frac{\eps_\solve\sqrt{\betaloc m}}{\eps_\solve\sqrt{\delta^{-1}\eps}} \cdot \eps_\solve^{-2}(\eps_\solve\sqrt{\delta^{-1}\eps})^2\right) \le O(\sqrt{\betaloc\delta^{-1}\eps m}). \]
\end{proof}

Combining these and carefully accumulating runtime costs gives a proof of \cref{thm:main}.
\begin{proof}[Proof of \cref{thm:main}]

We first show correctness.
\cref{lemma:dscorrect} gives
that the conditions necessary for \Checker~and \Locator~are always satisfied.
Inductively applying \cref{thm:batch} also gives that the
flow $\ffbar$ in line \ref{line:recenterbatch} is on the central path
in all steps with high probability.

It remains to bound the runtime. We observe that \textsc{RecenteringBatch} (\cref{algo:batchsteps}) is called $\O(\sqrt{m}/k)$ times since we initialize $\mu$ as $F=\poly(m)$. We first bound the total runtime of operations of $D^\chk_j$. By \cref{thm:checker} the runtimes of operations of all the $D^\chk_j$ combined~are:
\begin{itemize}
    \item $\textsc{Initialize}(G, \rr^\init \in \R^{E(G)}_{>0}, \eps, \betachk).$ Initializes graph $G$. Each \textsc{Initialize} runs in 
    \[\O\left(m\betachk^{-4}\eps^{-4}\right)=\O\left(m\left(k^{-28}\right)^{-4}\left(k^{-6}\right)^{-4}\right)=\O(mk^{136}).\]
    By \cref{lemma:rebuild}, the checkers are reinitialized \[\O\left(\eps_\solve^{-1}\betachk^{-1/2}\right)=\O\left(\left(k^{-3}\right)^{-1}\left(k^{-28}\right)^{-1/2}\right)=\O\left(k^{17}\right)\]
    times. As there are $k\eps_\step^{-1}=O\left(k/(k^{-3})\right)=O(k^4)$ distinct \Checker~data structures $D^\chk_j$, the total time cost of the reinitializations is 
    \[\O\left(mk^{136}\cdot k^{17}\cdot k^4\right)=\O\left(mk^{157}\right).\]
    
    \item $\textsc{PermanentUpdate}(e, \rr^\new \in \R^Z_{>0}).$ Each \textsc{PermanentUpdate} runs  in amortized 
    \[\O\left(\betachk^{-2}\eps^{-2}\right)=\O\left(\left(k^{-28}\right)^{-2}\left(k^{-6}\right)^{-2}\right)=\O\left(k^{68}\right)\]
    time. Between two reinitializations of a $D^\chk_j$, it processes at most $\betachk m=mk^{-28}$ \textsc{PermanentUpdate}s in 
    \[\O\left(mk^{-28}\cdot k^{68}\right)=\O\left(mk^{40}\right)\]
    time. Because there are $k^4$ checkers and 
    \[\O\left(\eps_\solve^{-1}\betachk^{-1/2}\right)=\O\left(\left(k^{-3}\right)^{-1}\left(k^{-28}\right)^{-1/2}\right)=\O\left(k^{17}\right)\] 
    reinitializations for each of them by \cref{lemma:rebuild}, the total time of all \textsc{PermanentUpdate}s is 
    \[\O\left(mk^{40}\cdot k^{4} \cdot k^{17}\right)=\O\left(mk^{61}\right).\]
    
    \item $\textsc{TemporaryUpdate}(e, \rr^\new_e >0).$ Updates $\rr_e \assign \rr^\new_e$. Runs in worst case \[\O\left(\left(K\betachk^{-2}\eps^{-2}\right)^2\right)=\O\left(\left(K\left(k^{-28}\right)^{-2}\left(k^{-6}\right)^{-2}\right)^2\right)=\O\left(K^2k^{136}\right)\]
    time for $K$ \textsc{TemporaryUpdate}s. \textsc{TemporaryUpdate} appears only in \textsc{RecenteringBatch} (\cref{algo:batchsteps}) and are rolled back before \textsc{RecenteringBatch} ends. By \cref{thm:batch}, there are $K = O(k^{16})$ \textsc{TemporaryUpdate}s for each of the $O(k^4)$ distinct $D^\chk_j$. Since \cref{algo:batchsteps} is called $\sqrt{m}/k$ times, the total time cost of all \textsc{TemporaryUpdate}s is 
    \[\O\left(\sqrt{m}/k\cdot k^4\cdot \left(k^{16}\right)^2k^{136}\right)=\O\left(\sqrt{m}k^{171}\right)\].
    
    \item $\textsc{Rollback}().$ Rollback the last \textsc{TemporaryUpdate} if exists. We charge the cost to the original operation as it costs the same time.
    
    \item $\textsc{Check}(e).$ Each \textsc{Check} costs 
    \begin{align*}
    &\O\left(\left(\betachk m+\left(K\betachk^{-2}\eps^{-2}\right)^2\right)\eps^{-2}\right)
    \\ =~&\O\left(\left(k^{-28}m+\left(K\left(k^{-28}\right)^{-2}\left(k^{-6}\right)^{-2}\right)^2\right)k^{12}\right)\\=~&\O\left(mk^{-16}+K^2k^{148}\right)
    \end{align*}
    time where $K$ is the number of $\textsc{TemporaryUpdate}$s that are not rolled back. For \textbf{all} $D^\chk_j$ s, \textsc{Check} is called a total of 
    \[k\eps_\step^{-1}\cdot O(\eps^{-2})=O\left(k\left(k^{-3}\right)^{-1}\left(k^{-6}\right)^{-2}\right)=O(k^{16})
    \] times in \cref{algo:batchsteps} by \cref{thm:batch}. $K$ is bounded by $O(k^{16})$ for each $D^\chk_j$~by \cref{thm:batch}. \cref{algo:batchsteps} is called $\sqrt{m}/k$ times. The total runtime of all \textsc{Check}s is 
    \[
    \O\left(\sqrt{m}/k\cdot k^{16} \cdot (mk^{-16}+\left(k^{16}\right)^2k^{148})\right)=\O\left(m^{3/2}k^{-1}+\sqrt{m}k^{195}\right).
    \]
\end{itemize}

We next bound the runtimes of operations of $D^\loc$. They are:
\begin{itemize}
    \item $\textsc{Initialize}(G, \rr, \eps, \betaloc, \delta).$ Initializes the data structure in
    \begin{align*}
    &\O\left(m\betaloc^{-4}\delta^{-2}\eps^{-2}+m\betaloc^{-4}\eps^{-4}\right) \\ =~ &\O\left(m\left(k^{-16}\right)^{-4}\left(k^{-38}\right)^{-2}\left(k^{-6}\right)^{-2}+m\left(k^{-16}\right)^{-4}\left(k^{-6}\right)^{-4}\right)\\
    =~&\O\left(mk^{152}\right).
    \end{align*} By \cref{lemma:rebuild}, $D^\loc$~is reinitilized \begin{equation} \O\left(\eps_\solve^{-1}\betaloc^{-1/2}\right)=\O\left(\left(k^{-3}\right)^{-1}\left(k^{-16}\right)^{-1/2}\right)=\O(k^{11}) \label{eq:locinit} \end{equation}
    times. The total time cost is 
    \[\O(mk^{152}\cdot k^{11})=\O(mk^{163}).\]
    
    \item $\textsc{Update}(e, \rr^\new)$. Runs in amortized 
    \begin{align*}&\O\left(\delta m\eps^{-3}+\betaloc^{-6}\delta^{-2}\eps^{-2}\right) \\ =~&\O\left(k^{-38} m\left(k^{-6}\right)^{-3}+\left(k^{-16}\right)^{-6}\left(k^{-38}\right)^{-2}\left(k^{-6}\right)^{-2}\right) \\=~&\O\left(mk^{-20}+k^{184}\right)
    \end{align*} time. By \cref{lemma:nobatch} the number of calls to $D^\loc.\textsc{Update}$ between two reinitializations is at most 
    \[\O\left(\sqrt{\betaloc\delta^{-1}\eps m}\right)=\O\left(\sqrt{k^{-16}\left(k^{-38}\right)^{-1}k^{-6} m}\right)=\O\left(\sqrt{m}k^8\right)
    \]
    while there are $\O(k^{11})$ reinitializations (see \eqref{eq:locinit}). Thus, the total time cost of all \textsc{Update} is \[\O\left(k^{11}\cdot \sqrt{m}k^{8}\cdot\left(mk^{-20}+k^{184}\right)\right) =\O\left(m^{3/2}k^{-1}+\sqrt{m}k^{203}\right).\]
    
    \item $\textsc{BatchUpdate}(S, \rr^\new \in \R_{>0}^S)$. Runs in 
    \[\O\left(m\eps^{-2} + |S|\betaloc^{-2}\eps^{-2}\right)=\O\left(m\left(k^{-6}\right)^{-2} + |S|\left(k^{-16}\right)^{-2}\left(k^{-6}\right)^{-2}\right)=\O\left(mk^{12}+|S|k^{44}\right)
    \]
    time. This occurs in line \ref{line:batchupdatelocator1} and \ref{line:batchupdatelocator2} of \textsc{FindFeasibleFlow} (\cref{algo:main}). We bound the number of times these two lines are called. Between two reinitializations, of which there are $\O(k^{11})$, line \ref{line:batchupdatelocator1} and \ref{line:batchupdatelocator2} are called
    \[ \O\left(\frac{\eps_\solve\sqrt{\betaloc m}/k}{\eps_\solve\sqrt{\delta^{-1}\eps}/k}\right) = \O\left(\frac{k^{-3}\sqrt{k^{-16} m}/k}{k^{-3}\sqrt{\left(k^{-38}\right)^{-1}k^{-6}}/k}\right) =\sqrt{m}k^{-24} \] times, as exactly one of them is called only if $i$ is a multiple of $c\eps_\solve\sqrt{\delta^{-1}\eps}/k$.
    Also, a total of $\betaloc m=mk^{-16}$ edges is updated due to \cref{lemma:dscorrect} between two reinitializations. There are $O(k^{11})$ reinitializations (see \eqref{eq:locinit}).
    Thus, the total cost is bounded by
    \begin{align*}
    &\O\left(k^{11}\cdot \left(\sqrt{m}k^{-24}\cdot mk^{12} + mk^{-16}k^{44}\right)\right)\\
    =&\O\left(m^{3/2}k^{-1}+mk^{39}\right).
    \end{align*}
    
    \item $\textsc{Locate}()$. Runs in 
    \[\O\left(\betaloc m \eps^{-2}\right)=\O\left(\left(k^{-16}\right) m \left(k^{-6}\right)^{-2}\right)=\O\left(mk^{-4}\right)\] 
    time. In each call to \textsc{RecenteringBatch}, $\textsc{Locate}$ is called 
    \[O\left(k\eps_{\solve}^{-1}\right)=O\left(k\left(k^{-3}\right)^{-1}\right)=O\left(k^4\right)\] times by \cref{thm:batch}. As there are $\sqrt{m}/k$ calls to \textsc{RecenteringBatch}, $\textsc{Locate}$ is called \[\O\left(k^4\cdot \sqrt{m}/k\right)=\O\left(\sqrt{m}k^3\right)\] times for total time \[\O\left(\sqrt{m}k^3\cdot mk^{-4}\right)=\O\left(m^{3/2}k^{-1}\right).\]
\end{itemize}

Accumulating all the costs above gives an overall bound of
\[
\O\left(mk^{163} + m^{3/2}/k + \sqrt{m}k^{203}\right).
\]
While there are a few basic arithmetic and set operations in \textsc{FindFeasibleFlow} as well, the time cost of the operations to $D^\chk_j$~and $D^\loc$~dominates the final runtime. For the choice $k = m^{1/328},$ the final runtime is $\O(m^{\frac32-\frac{1}{328}})$ as desired.
\end{proof}

\section*{Acknowledgments}
Yang P. Liu was supported by the Department of Defense (DoD) through
the National Defense Science and Engineering Graduate Fellowship (NDSEG) Program.
Richard Peng is supported by the National Science Foundation (NSF)
under Grant No. 1846218.

We thank Jan van den Brand, Arun Jambulapati, Yin Tat Lee, and Aaron Sidford for helpful discussions and pointing out typos in an earlier version of this manuscript. We especially thank Aaron Sidford for discussions during which an error in the handling of adaptivity and randomness in the original version of this manuscript was pointed out.

{\small
\bibliographystyle{alpha}
\bibliography{refs.bib}}

\newcommand{\etalchar}[1]{$^{#1}$}
\begin{thebibliography}{KNPW11}

\bibitem[Ach03]{Ach03}
Dimitris Achlioptas.
\newblock Database-friendly random projections: Johnson-lindenstrauss with
  binary coins.
\newblock {\em Journal of Computer and System Sciences}, 66(4):671--687, 2003.

\bibitem[ADK{\etalchar{+}}16]{ADKKP16}
Ittai Abraham, David Durfee, Ioannis Koutis, Sebastian Krinninger, and Richard
  Peng.
\newblock On fully dynamic graph sparsifiers.
\newblock In {\em 57th {IEEE} Annual Symposium on Foundations of Computer
  Science, {FOCS} 2016, Hyatt Regency, New Brunswick, New Jersey, USA, October
  9-11, 2016}, pages 335--344. {IEEE} Computer Society, 2016.
\newblock Available at~\url{https://arxiv.org/abs/1604.02094}.

\bibitem[AO91]{AO91}
Ravindra~K Ahuja and James~B Orlin.
\newblock Distance-directed augmenting path algorithms for maximum flow and
  parametric maximum flow problems.
\newblock {\em Naval Research Logistics}, 38(3):413--430, 1991.

\bibitem[BBG{\etalchar{+}}20]{BBG20}
Aaron Bernstein, Jan van~den Brand, Maximilian~Probst Gutenberg, Danupon
  Nanongkai, Thatchaphol Saranurak, Aaron Sidford, and He~Sun.
\newblock Fully-dynamic graph sparsifiers against an adaptive adversary.
\newblock {\em CoRR}, abs/2004.08432, 2020.
\newblock Available at~\url{https://arxiv.org/abs/2004.08432}.

\bibitem[BHN16]{BHN16}
Sayan Bhattacharya, Monika Henzinger, and Danupon Nanongkai.
\newblock New deterministic approximation algorithms for fully dynamic
  matching.
\newblock In {\em Proceedings of the 48th Annual {ACM} {SIGACT} Symposium on
  Theory of Computing, {STOC} 2016, Cambridge, MA, USA, June 18-21, 2016},
  pages 398--411. {ACM}, 2016.
\newblock Available at~\url{https://arxiv.org/abs/1604.05765}.

\bibitem[BLL{\etalchar{+}}21]{BLLSSSW20:arxiv}
Jan van~den Brand, Yin~Tat Lee, Yang~P. Liu, Thatchaphol Saranurak, Aaron
  Sidford, Zhao Song, and Di~Wang.
\newblock Minimum cost flows, {MDP}s, and $\ell_1$-regression in nearly linear
  time for dense instances.
\newblock {\em CoRR}, abs/2101.05719, 2021.
\newblock Available at~\url{https://arxiv.org/abs/2101.05719}.

\bibitem[BLN{\etalchar{+}}20]{BLNPSSSW20}
Jan van~den Brand, Yin~Tat Lee, Danupon Nanongkai, Richard Peng, Thatchaphol
  Saranurak, Aaron Sidford, Zhao Song, and Di~Wang.
\newblock Bipartite matching in nearly-linear time on moderately dense graphs.
\newblock In {\em 61st {IEEE} Annual Symposium on Foundations of Computer
  Science, {FOCS} 2020, Durham, NC, USA, November 16-19, 2020}, pages 919--930,
  2020.

\bibitem[BLSS20]{BLSS20}
Jan van~den Brand, Yin~Tat Lee, Aaron Sidford, and Zhao Song.
\newblock Solving tall dense linear programs in nearly linear time.
\newblock In {\em Proccedings of the 52nd Annual {ACM} {SIGACT} Symposium on
  Theory of Computing, {STOC} 2020, Chicago, IL, USA, June 22-26, 2020}, pages
  775--788. {ACM}, 2020.
\newblock Available at~\url{https://arxiv.org/abs/2002.02304}.

\bibitem[Bra20]{B20a}
Jan van~den Brand.
\newblock A deterministic linear program solver in current matrix
  multiplication time.
\newblock In {\em Proceedings of the Thirty-First Annual {ACM-SIAM} Symposium
  on Discrete Algorithms, {SODA} 2020, Salt Lake City, UT, USA, January 5-8,
  2020}, pages 259--278. {SIAM}, 2020.
\newblock Available at~\url{https://arxiv.org/abs/1910.11957}.

\bibitem[BS15]{BS15}
Aaron Bernstein and Cliff Stein.
\newblock Fully dynamic matching in bipartite graphs.
\newblock In {\em Automata, Languages, and Programming - 42nd International
  Colloquium, {ICALP} 2015, Kyoto, Japan, July 6-10, 2015, Proceedings, Part
  {I}}, volume 9134 of {\em Lecture Notes in Computer Science}, pages 167--179.
  Springer, 2015.
\newblock Available at~\url{https://arxiv.org/abs/1506.07076}.

\bibitem[CGH{\etalchar{+}}20]{CGHPS20}
Li~Chen, Gramoz Goranci, Monika Henzinger, Richard Peng, and Thatchaphol
  Saranurak.
\newblock Fast dynamic cuts, distances and effective resistances via vertex
  sparsifiers.
\newblock In {\em 61st {IEEE} Annual Symposium on Foundations of Computer
  Science, {FOCS} 2020, Durham, NC, USA, November 16-19, 2020}, pages
  1135--1146. {IEEE}, 2020.
\newblock Available at~\url{https://arxiv.org/abs/2005.02368}.

\bibitem[CGL{\etalchar{+}}20]{CGLNPS19}
Julia Chuzhoy, Yu~Gao, Jason Li, Danupon Nanongkai, Richard Peng, and
  Thatchaphol Saranurak.
\newblock A deterministic algorithm for balanced cut with applications to
  dynamic connectivity, flows, and beyond.
\newblock In {\em 61st {IEEE} Annual Symposium on Foundations of Computer
  Science, {FOCS} 2020, Durham, NC, USA, November 16-19, 2020}, pages
  1158--1167. {IEEE}, 2020.
\newblock Available at~\url{https://arxiv.org/abs/1910.08025}.

\bibitem[CKK{\etalchar{+}}18]{CKKPPRS18}
Michael~B. Cohen, Jonathan~A. Kelner, Rasmus Kyng, John Peebles, Richard Peng,
  Anup~B. Rao, and Aaron Sidford.
\newblock Solving directed {L}aplacian systems in nearly-linear time through
  sparse {LU} factorizations.
\newblock In {\em 59th {IEEE} Annual Symposium on Foundations of Computer
  Science, {FOCS} 2018, Paris, France, October 7-9, 2018}, pages 898--909.
  {IEEE} Computer Society, 2018.
\newblock Available at~\url{https://arxiv.org/abs/1811.10722}.

\bibitem[CKM{\etalchar{+}}11]{CKMST11}
Paul Christiano, Jonathan~A. Kelner, Aleksander Madry, Daniel~A. Spielman, and
  Shang{-}Hua Teng.
\newblock Electrical flows, {L}aplacian systems, and faster approximation of
  maximum flow in undirected graphs.
\newblock In {\em Proceedings of the 43rd {ACM} Symposium on Theory of
  Computing, {STOC} 2011, San Jose, CA, USA, June 6-8 2011}, pages 273--282.
  {ACM}, 2011.
\newblock Available at \url{https://arxiv.org/abs/1010.2921}.

\bibitem[CKM{\etalchar{+}}14]{CKMPPRX14}
Michael~B. Cohen, Rasmus Kyng, Gary~L. Miller, Jakub~W. Pachocki, Richard Peng,
  Anup~B. Rao, and Shen~Chen Xu.
\newblock Solving {SDD} linear systems in nearly {$m\log^{1/2}n$} time.
\newblock In {\em Proceedings of the 46th Annual {ACM} Symposium on Theory of
  Computing, {STOC} 2014, New York, NY, USA, June 1-3, 2014}, pages 343--352,
  2014.

\bibitem[CLRS09]{CLRS}
Thomas~H. Cormen, Charles~E. Leiserson, Ronald~L. Rivest, and Clifford Stein.
\newblock {\em Introduction to Algorithms, 3rd Edition}.
\newblock {MIT} Press, 2009.

\bibitem[CLS19]{CLS19}
Michael~B. Cohen, Yin~Tat Lee, and Zhao Song.
\newblock Solving linear programs in the current matrix multiplication time.
\newblock In {\em Proceedings of the 51st Annual {ACM} {SIGACT} Symposium on
  Theory of Computing, {STOC} 2019, Phoenix, AZ, USA, June 23-26, 2019}, pages
  938--942. {ACM}, 2019.
\newblock Available at~\url{https://arxiv.org/abs/1810.07896}.

\bibitem[CMSV17]{CMSV17}
Michael~B. Cohen, Aleksander Madry, Piotr Sankowski, and Adrian Vladu.
\newblock Negative-weight shortest paths and unit capacity minimum cost flow in
  $\widetilde{O}(m^{10/7} \log{W})$ time (extended abstract).
\newblock In {\em Proceedings of the Twenty-Eighth Annual {ACM-SIAM} Symposium
  on Discrete Algorithms, {SODA} 2017, Barcelona, Spain, Hotel Porta Fira,
  January 16-19}, pages 752--771. {SIAM}, 2017.
\newblock Available at~\url{https://arxiv.org/abs/1605.01717}.

\bibitem[DFGX18]{DFGX18}
David Durfee, Matthew Fahrbach, Yu~Gao, and Tao Xiao.
\newblock Nearly tight bounds for sandpile transience on the grid.
\newblock In {\em Proceedings of the Twenty-Ninth Annual {ACM-SIAM} Symposium
  on Discrete Algorithms, {SODA} 2018, New Orleans, LA, USA, January 7-10,
  2018}, pages 605--624. {SIAM}, 2018.
\newblock Available at~\url{https://arxiv.org/abs/1704.04830}.

\bibitem[DGGP19]{DGGP19}
David Durfee, Yu~Gao, Gramoz Goranci, and Richard Peng.
\newblock Fully dynamic spectral vertex sparsifiers and applications.
\newblock In {\em Proceedings of the 51st Annual {ACM} {SIGACT} Symposium on
  Theory of Computing, {STOC} 2019, Phoenix, AZ, USA, June 23-26, 2019}, pages
  914--925. {ACM}, 2019.
\newblock Available at~\url{https://arxiv.org/abs/1906.10530}.

\bibitem[Din70]{Dinitz70}
E.A. Dinic.
\newblock Algorithm for solution of a problem of maximum flow in networks with
  power estimation.
\newblock {\em Soviet Mathematics Doklady}, 11:1277--1280, 1970.

\bibitem[DS84]{doyleS84}
Peter~G. Doyle and J.~Laurie Snell.
\newblock {\em Random Walks and Electric Networks}.
\newblock Mathematical Association of America, 1984.
\newblock Available at~\url{https://arxiv.org/abs/math/0001057}.

\bibitem[DS08]{DS08}
Samuel~I. Daitch and Daniel~A. Spielman.
\newblock Faster approximate lossy generalized flow via interior point
  algorithms.
\newblock In {\em Proceedings of the 40th annual {ACM} Symposium on Theory of
  Computing, {STOC} 2008, Victoria, BC, Canada, May 17-20, 2008}, pages
  451--460, New York, NY, USA, 2008. ACM.
\newblock Available at~\url{http://arxiv.org/abs/0803.0988}.

\bibitem[EK72]{EK73}
Jack Edmonds and Richard~M. Karp.
\newblock Theoretical improvements in algorithmic efficiency for network flow
  problems.
\newblock {\em Journal of the ACM}, 19(2):248–264, April 1972.

\bibitem[ET75]{ET75}
Shimon Even and R~Endre Tarjan.
\newblock Network flow and testing graph connectivity.
\newblock {\em SIAM Journal on Computing}, 4(4):507--518, 1975.

\bibitem[FG19]{FG19}
Sebastian Forster and Gramoz Goranci.
\newblock Dynamic low-stretch trees via dynamic low-diameter decompositions.
\newblock In {\em Proceedings of the 51st Annual {ACM} {SIGACT} Symposium on
  Theory of Computing, {STOC} 2019, Phoenix, AZ, USA, June 23-26, 2019}, pages
  377--388. {ACM}, 2019.
\newblock Available at~\url{https://arxiv.org/abs/1804.04928}.

\bibitem[FGL{\etalchar{+}}20]{FGLPSY20:arxiv}
Sebastian Forster, Gramoz Goranci, Yang~P. Liu, Richard Peng, Xiaorui Sun, and
  Mingquan Ye.
\newblock Minor sparsifiers and the distributed {L}aplacian paradigm, 2020.
\newblock Available at~\url{https://arxiv.org/abs/2012.15675}.

\bibitem[FMP{\etalchar{+}}18]{FMPSWX18}
Matthew Fahrbach, Gary~L. Miller, Richard Peng, Saurabh Sawlani, Junxing Wang,
  and Shen~Chen Xu.
\newblock Graph sketching against adaptive adversaries applied to the minimum
  degree algorithm.
\newblock In {\em 59th {IEEE} Annual Symposium on Foundations of Computer
  Science, {FOCS} 2018, Paris, France, October 7-9, 2018}, pages 101--112.
  {IEEE} Computer Society, 2018.
\newblock Available at~\url{https://arxiv.org/abs/1804.04239}.

\bibitem[GG11]{GG11}
Eran Gat and Shafi Goldwasser.
\newblock Probabilistic search algorithms with unique answers and their
  cryptographic applications.
\newblock {\em Electronic Colloquium on Computational Complexity (ECCC)},
  18:136, 2011.
\newblock Available at~\url{https://eccc.weizmann.ac.il/report/2011/136/}.

\bibitem[GGR13]{GGR13}
Oded Goldreich, Shafi Goldwasser, and Dana Ron.
\newblock On the possibilities and limitations of pseudodeterministic
  algorithms.
\newblock In {\em Proceedings of the 4th Conference on Innovations in
  Theoretical Computer Science}, ITCS '13, pages 127--138, New York, NY, USA,
  2013. ACM.

\bibitem[GHP17]{GHP17}
Gramoz Goranci, Monika Henzinger, and Pan Peng.
\newblock The power of vertex sparsifiers in dynamic graph algorithms.
\newblock In {\em 25th Annual European Symposium on Algorithms, {ESA} 2017,
  September 4-6, 2017, Vienna, Austria}, volume~87 of {\em LIPIcs}, pages
  45:1--45:14. Schloss Dagstuhl - Leibniz-Zentrum f{\"{u}}r Informatik, 2017.
\newblock Available at~\url{https://arxiv.org/abs/1712.06473}.

\bibitem[GHP18]{GoranciHP18}
Gramoz Goranci, Monika Henzinger, and Pan Peng.
\newblock Dynamic effective resistances and approximate schur complement on
  separable graphs.
\newblock In {\em 26th Annual European Symposium on Algorithms, {ESA} 2018,
  August 20-22, 2018, Helsinki, Finland}, volume 112 of {\em LIPIcs}, pages
  40:1--40:15. Schloss Dagstuhl - Leibniz-Zentrum f{\"{u}}r Informatik, 2018.
\newblock Available at~\url{https://arxiv.org/abs/1802.09111}.

\bibitem[GN80]{GN80}
Zvi Galil and Amnon Naamad.
\newblock An ${O}({E}{V}\log^{2}{{V}})$ algorithm for the maximal flow problem.
\newblock {\em Journal of Computer and System Sciences}, 21(2):203--217, 1980.

\bibitem[Gor19]{Goranci19:thesis}
Gramoz Goranci.
\newblock Dynamic graph algorithms and graph sparsification: New techniques and
  connections.
\newblock {\em CoRR}, abs/1909.06413, 2019.
\newblock Available at~\url{https://arxiv.org/abs/1909.06413}.

\bibitem[GP13]{GP13}
Manoj Gupta and Richard Peng.
\newblock Fully dynamic $(1 + \epsilon)$-approximate matchings.
\newblock In {\em 54th Annual {IEEE} Symposium on Foundations of Computer
  Science, {FOCS} 2013, Berkeley, CA, {USA}, October 26-29, 2013}, pages
  548--557, 2013.
\newblock Available at~\url{http://arxiv.org/abs/1304.0378}.

\bibitem[GR98]{GR98}
Andrew~V. Goldberg and Satish Rao.
\newblock Beyond the flow decomposition barrier.
\newblock {\em Journal of the ACM}, 45(5):783--797, 1998.
\newblock Announced at FOCS'97.

\bibitem[GT88]{GT88}
Andrew~V. Goldberg and Robert~Endre Tarjan.
\newblock A new approach to the maximum-flow problem.
\newblock {\em J. {ACM}}, 35(4):921--940, 1988.

\bibitem[GT14]{GT14}
Andrew~V Goldberg and Robert~E Tarjan.
\newblock Efficient maximum flow algorithms.
\newblock {\em Communications of the ACM}, 57(8):82--89, 2014.
\newblock Available
  at~\url{https://cacm.acm.org/magazines/2014/8/177011-efficient-maximum-flow-algorithms}.

\bibitem[HK73]{HK73}
John~E Hopcroft and Richard~M Karp.
\newblock A $n^{5/2}$ algorithm for maximum matchings in bipartite graphs.
\newblock {\em SIAM Journal on Computing}, 2(4):225--231, 1973.

\bibitem[HKN18]{HKN18}
Monika Henzinger, Sebastian Krinninger, and Danupon Nanongkai.
\newblock Decremental single-source shortest paths on undirected graphs in
  near-linear total update time.
\newblock {\em Journal of the ACM}, 65(6):36:1--36:40, 2018.
\newblock Available at~\url{https://arxiv.org/abs/1512.08148}.

\bibitem[JKL{\etalchar{+}}20]{JKLPS20}
Haotian Jiang, Tarun Kathuria, Yin~Tat Lee, Swati Padmanabhan, and Zhao Song.
\newblock A faster interior point method for semidefinite programming.
\newblock In {\em 61st {IEEE} Annual Symposium on Foundations of Computer
  Science, {FOCS} 2020, Durham, NC, USA, November 16-19, 2020}, pages 910--918.
  {IEEE}, 2020.
\newblock Available at:~\url{https://arxiv.org/abs/2009.10217}.

\bibitem[JS20]{JS20:arxiv}
Wenyu Jin and Xiaorui Sun.
\newblock Fully dynamic c-edge connectivity in subpolynomial time.
\newblock {\em CoRR}, abs/2004.07650, 2020.
\newblock Available at~\url{https://arxiv.org/abs/2004.07650}.

\bibitem[Kar73]{K73}
Alexander~V Karzanov.
\newblock On finding maximum flows in networks with special structure and some
  applications.
\newblock {\em Matematicheskie Voprosy Upravleniya Proizvodstvom}, 5:81--94,
  1973.

\bibitem[Kar74]{Kar74}
Alexander Karzanov.
\newblock Determining the maximal flow in a network by the method of preflows.
\newblock {\em Doklady Mathematics}, 15:434–437, 02 1974.

\bibitem[Kar84]{K84}
N~Karmarkar.
\newblock A new polynomial-time algorithm for linear programming.
\newblock {\em Combinatorica}, 4(4):373--395, 1984.

\bibitem[KLOS14]{KLOS14}
Jonathan~A. Kelner, Yin~Tat Lee, Lorenzo Orecchia, and Aaron Sidford.
\newblock An almost-linear-time algorithm for approximate max flow in
  undirected graphs, and its multicommodity generalizations.
\newblock In {\em Proceedings of the Twenty-Fifth Annual {ACM-SIAM} Symposium
  on Discrete Algorithms, {SODA} 2014, Portland, OR, USA, January 5-7, 2014},
  pages 217--226, 2014.
\newblock Available at~\url{https://arxiv.org/abs/1304.2338}.

\bibitem[KLP{\etalchar{+}}16]{KLPSS16}
Rasmus Kyng, Yin~Tat Lee, Richard Peng, Sushant Sachdeva, and Daniel~A.
  Spielman.
\newblock Sparsified {C}holesky and multigrid solvers for connection
  {L}aplacians.
\newblock In {\em Proceedings of the 48th Annual {ACM} {SIGACT} Symposium on
  Theory of Computing, {STOC} 2016, Cambridge, MA, USA, June 18-21, 2016},
  pages 842--850, 2016.
\newblock Available at~\url{https://arxiv.org/abs/1512.01892}.

\bibitem[KLS20]{KLS20}
Tarun Kathuria, Yang~P. Liu, and Aaron Sidford.
\newblock Unit capacity maxflow in almost {$O(m^{4/3})$} time.
\newblock In {\em 61st {IEEE} Annual Symposium on Foundations of Computer
  Science, {FOCS} 2020, Durham, NC, USA, November 16-19, 2020}, pages 119--130.
  {IEEE}, 2020.

\bibitem[KMP10]{KMP10}
Ioannis Koutis, Gary~L. Miller, and Richard Peng.
\newblock Approaching optimality for solving {SDD} linear systems.
\newblock In {\em 51th {IEEE} Annual Symposium on Foundations of Computer
  Science, {FOCS} 2010, Las Vegas, NV, {USA}, October 23-26, 2010}, pages
  235--244, 2010.
\newblock Available at~\url{https://arxiv.org/abs/1003.2958}.

\bibitem[KMP11]{KMP11}
Ioannis Koutis, Gary~L. Miller, and Richard Peng.
\newblock A nearly-m log n time solver for {SDD} linear systems.
\newblock In {\em 52nd {IEEE} Annual Symposium on Foundations of Computer
  Science, {FOCS} 2011, Palm Springs, CA, USA, October 22-25, 2011}, pages
  590--598, 2011.
\newblock Available at~\url{https://arxiv.org/abs/1102.4842}.

\bibitem[KNPW11]{KNPW11}
Daniel~M. Kane, Jelani Nelson, Ely Porat, and David~P. Woodruff.
\newblock Fast moment estimation in data streams in optimal space.
\newblock In {\em Proceedings of the 43rd {ACM} Symposium on Theory of
  Computing, {STOC} 2011, San Jose, CA, USA, June 6-8 2011}, pages 745--754.
  {ACM}, 2011.
\newblock Available at~\url{https://arxiv.org/abs/1007.4191}.

\bibitem[KOSA13]{KOSZ13}
Jonathan~A. Kelner, Lorenzo Orecchia, Aaron Sidford, and Zeyuan {Allen Zhu}.
\newblock A simple, combinatorial algorithm for solving {SDD} systems in
  nearly-linear time.
\newblock In {\em Proceedings of the 45th Annual {ACM} Symposium on Theory of
  Computing, {STOC} 2013, Palo Alto, CA, USA, June 1-4, 2013}, pages 911--920,
  2013.
\newblock Available at~\url{https://arxiv.org/abs/1301.6628}.

\bibitem[KP15]{KP15:arxiv}
Donggu Kang and James Payor.
\newblock Flow rounding, 2015.
\newblock Available at~\url{https://arxiv.org/abs/1507.08139}.

\bibitem[KPSW19]{KPSW19}
Rasmus Kyng, Richard Peng, Sushant Sachdeva, and Di~Wang.
\newblock Flows in almost linear time via adaptive preconditioning.
\newblock In {\em Proceedings of the 51st Annual {ACM} {SIGACT} Symposium on
  Theory of Computing, {STOC} 2019, Phoenix, AZ, USA, June 23-26, 2019}, pages
  902--913. {ACM}, 2019.
\newblock Available at~\url{https://arxiv.org/abs/1906.10340}.

\bibitem[KRT94]{KRT94}
V.~King, S.~Rao, and R.~Tarjan.
\newblock A faster deterministic maximum flow algorithm.
\newblock {\em Journal of Algorithms}, 17(3):447--474, 1994.

\bibitem[KS16]{KS16}
Rasmus Kyng and Sushant Sachdeva.
\newblock Approximate gaussian elimination for {L}aplacians - fast, sparse, and
  simple.
\newblock In {\em 57th {IEEE} Annual Symposium on Foundations of Computer
  Science, {FOCS} 2016, Hyatt Regency, New Brunswick, NJ, {USA}, October 9-11,
  2016}, pages 573--582, 2016.
\newblock Available at~\url{https://arxiv.org/abs/1605.02353}.

\bibitem[Lin09]{Lin09}
Henry Lin.
\newblock Reducing directed max flow to undirected max flow.
\newblock {\em Unpublished Manuscript}, 4(2), 2009.

\bibitem[LS]{LSpersonal}
Yin~Tat Lee and Aaron Sidford.
\newblock Personal communication.

\bibitem[LS15]{LS15}
Yin~Tat Lee and Aaron Sidford.
\newblock Efficient inverse maintenance and faster algorithms for linear
  programming.
\newblock In {\em 56th {IEEE} Annual Symposium on Foundations of Computer
  Science, {FOCS} 2015, Berkeley, CA, USA, October 17-20, 2015}, pages
  230--249. {IEEE} Computer Society, 2015.
\newblock Available at~\url{https://arxiv.org/abs/1503.01752}.

\bibitem[LS19]{LS19}
Yin~Tat Lee and Aaron Sidford.
\newblock Solving linear programs with sqrt(rank) linear system solves.
\newblock {\em CoRR}, abs/1910.08033, 2019.

\bibitem[LS20a]{LS20:arxiv}
Yang~P. Liu and Aaron Sidford.
\newblock Faster divergence maximization for faster maximum flow.
\newblock {\em CoRR}, abs/2003.08929, 2020.
\newblock Available at~\url{https://arxiv.org/abs/2003.08929}.

\bibitem[LS20b]{LS20_STOC}
Yang~P. Liu and Aaron Sidford.
\newblock Faster energy maximization for faster maximum flow.
\newblock In {\em Proccedings of the 52nd Annual {ACM} {SIGACT} Symposium on
  Theory of Computing, {STOC} 2020, Chicago, IL, USA, June 22-26, 2020}, pages
  803--814. {ACM}, 2020.
\newblock Available at~\url{https://arxiv.org/abs/1910.14276}.

\bibitem[Mad13]{M13}
Aleksander Madry.
\newblock Navigating central path with electrical flows: From flows to
  matchings, and back.
\newblock In {\em 54th {IEEE} Annual Symposium on Foundations of Computer
  Science, {FOCS} 2013, Berkeley, CA, {USA}, October 26-29, 2013}, pages
  253--262. {IEEE} Computer Society, 2013.
\newblock Available at~\url{https://arxiv.org/abs/1307.2205}.

\bibitem[Mad16]{M16}
Aleksander Madry.
\newblock Computing maximum flow with augmenting electrical flows.
\newblock In {\em 57th {IEEE} Annual Symposium on Foundations of Computer
  Science, {FOCS} 2016, 9-11 October 2016, Hyatt Regency, New Brunswick, New
  Jersey, {USA}}, pages 593--602. {IEEE} Computer Society, 2016.
\newblock Available at~\url{https://arxiv.org/abs/1608.06016}.

\bibitem[NS17]{NS17}
Danupon Nanongkai and Thatchaphol Saranurak.
\newblock Dynamic spanning forest with worst-case update time: adaptive, {L}as
  {V}egas, and ${O}(n^{1/2 - \epsilon})$-time.
\newblock In {\em Proceedings of the 49th Annual {ACM} {SIGACT} Symposium on
  Theory of Computing, {STOC} 2017, Montreal, QC, Canada, June 19-23, 2017},
  pages 1122--1129, 2017.
\newblock Available at~\url{https://arxiv.org/abs/1611.03745}.

\bibitem[NSW17]{NSW17}
Danupon Nanongkai, Thatchaphol Saranurak, and Christian Wulff{-}Nilsen.
\newblock Dynamic minimum spanning forest with subpolynomial worst-case update
  time.
\newblock In {\em 58th {IEEE} Annual Symposium on Foundations of Computer
  Science, {FOCS} 2017, Berkeley, CA, USA, October 15-17, 2017}, pages
  950--961. {IEEE} Computer Society, 2017.
\newblock Available at~\url{https://arxiv.org/abs/1708.03962}.

\bibitem[OR10]{OR10}
Krzysztof Onak and Ronitt Rubinfeld.
\newblock Maintaining a large matching and a small vertex cover.
\newblock In {\em Proceedings of the 42nd {ACM} Symposium on Theory of
  Computing, {STOC} 2010, Cambridge, MA, USA, June 5-8 2010}, pages 457--464.
  {ACM}, 2010.
\newblock Available
  at~\url{http://people.csail.mit.edu/ronitt/papers/01-maintaining.pdf}.

\bibitem[Orl13]{Orlin13}
James~B. Orlin.
\newblock Max flows in {$O(nm)$} time, or better.
\newblock In {\em Proceedings of the 45th Annual {ACM} Symposium on Theory of
  Computing, {STOC} 2013, Palo Alto, CA, USA, June 1-4, 2013}, pages 765--774.
  {ACM}, 2013.

\bibitem[Pen16]{P16}
Richard Peng.
\newblock Approximate undirected maximum flows in
  \emph{O}(\emph{m}polylog(\emph{n})) time.
\newblock In {\em Proceedings of the Twenty-Seventh Annual {ACM-SIAM} Symposium
  on Discrete Algorithms, {SODA} 2016, Arlington, VA, USA, January 10-12,
  2016}, pages 1862--1867. {SIAM}, 2016.

\bibitem[PSS19]{PSS19}
Richard Peng, Bryce Sandlund, and Daniel~Dominic Sleator.
\newblock Optimal offline dynamic 2, 3-edge/vertex connectivity.
\newblock In {\em Algorithms and Data Structures - 16th International
  Symposium, {WADS} 2019, Edmonton, AB, Canada, August 5-7, 2019}, volume 11646
  of {\em Lecture Notes in Computer Science}, pages 553--565. Springer, 2019.
\newblock Available at~\url{https://arxiv.org/abs/1708.03812}.

\bibitem[Ren88]{Ren88}
James Renegar.
\newblock A polynomial-time algorithm, based on newton's method, for linear
  programming.
\newblock {\em Mathematical Programming}, 40(1-3):59--93, 1988.

\bibitem[Sch18]{S18}
Aaron Schild.
\newblock An almost-linear time algorithm for uniform random spanning tree
  generation.
\newblock In {\em Proceedings of the 50th Annual {ACM} {SIGACT} Symposium on
  Theory of Computing, {STOC} 2018, Los Angeles, CA, USA, June 25-29, 2018},
  pages 214--227. {ACM}, 2018.
\newblock Available at~\url{https://arxiv.org/abs/1711.06455}.

\bibitem[She13]{S13}
Jonah Sherman.
\newblock Nearly maximum flows in nearly linear time.
\newblock In {\em 54th {IEEE} Annual Symposium on Foundations of Computer
  Science, {FOCS} 2013, Berkeley, CA, {USA}, October 26-29, 2013}, pages
  263--269, 2013.
\newblock Available at~\url{https://arxiv.org/abs/1304.2077}.

\bibitem[She17]{S17}
Jonah Sherman.
\newblock Generalized preconditioning and undirected minimum-cost flow.
\newblock In {\em Proceedings of the Twenty-Eighth Annual {ACM-SIAM} Symposium
  on Discrete Algorithms, {SODA} 2017, Barcelona, Spain, Hotel Porta Fira,
  January 16-19}, pages 772--780. {SIAM}, 2017.
\newblock Available at~\url{https://arxiv.org/abs/1606.07425}.

\bibitem[SRS18]{SRS18}
Aaron Schild, Satish Rao, and Nikhil Srivastava.
\newblock Localization of electrical flows.
\newblock In {\em Proceedings of the Twenty-Ninth Annual {ACM-SIAM} Symposium
  on Discrete Algorithms, {SODA} 2018, New Orleans, LA, USA, January 7-10,
  2018}, pages 1577--1584. {SIAM}, 2018.
\newblock Available at~\url{https://arxiv.org/abs/1708.01632}.

\bibitem[ST83]{ST83}
Daniel~D Sleator and Robert~Endre Tarjan.
\newblock A data structure for dynamic trees.
\newblock {\em Journal of Computer and System Sciences}, 26(3):362--391, 1983.
\newblock Announced at STOC'81.

\bibitem[ST04]{ST04}
Daniel~A. Spielman and Shang{-}Hua Teng.
\newblock Nearly-linear time algorithms for graph partitioning, graph
  sparsification, and solving linear systems.
\newblock In {\em Proceedings of the 36th Annual {ACM} Symposium on Theory of
  Computing, {STOC} 2004, Chicago, IL, USA, June 13-16, 2004}, pages 81--90,
  2004.
\newblock Available at \url{https://arxiv.org/abs/0809.3232},
  \url{https://arxiv.org/abs/0808.4134},
  \url{https://arxiv.org/abs/cs/0607105}.

\bibitem[ST11]{ST11}
Daniel~A Spielman and Shang-Hua Teng.
\newblock Spectral sparsification of graphs.
\newblock {\em SIAM Journal on Computing}, 40(4):981--1025, 2011.
\newblock Available at~\url{https://arxiv.org/abs/0808.4134}.

\bibitem[ST18]{ST18}
Aaron Sidford and Kevin Tian.
\newblock Coordinate methods for accelerating $\ell_\infty$ regression and
  faster approximate maximum flow.
\newblock In {\em 59th {IEEE} Annual Symposium on Foundations of Computer
  Science, {FOCS} 2018, Paris, France, October 7-9, 2018}, pages 922--933,
  2018.
\newblock Available at~\url{https://arxiv.org/abs/1808.01278}.

\bibitem[Sto05]{S05}
Arne Storjohann.
\newblock The shifted number system for fast linear algebra on integer
  matrices.
\newblock {\em Journal of Complexity}, 21(4):609--650, 2005.
\newblock Available at~\url{https://cs.uwaterloo.ca/\textasciitilde
  astorjoh/shifted.pdf}.

\bibitem[SW19]{SW19}
Thatchaphol Saranurak and Di~Wang.
\newblock Expander decomposition and pruning: Faster, stronger, and simpler.
\newblock In {\em Proceedings of the Thirtieth Annual {ACM-SIAM} Symposium on
  Discrete Algorithms, {SODA} 2019, San Diego, California, USA, January 6-9,
  2019}, pages 2616--2635. {SIAM}, 2019.
\newblock Available at~\url{https://arxiv.org/abs/1812.08958}.

\bibitem[Vai89]{Vaidya89}
Pravin~M. Vaidya.
\newblock Speeding-up linear programming using fast matrix multiplication
  (extended abstract).
\newblock In {\em 30th {IEEE} Annual Symposium on Foundations of Computer
  Science, {FOCS} 1989, Research Triangle Park, NC, USA, October 30 - November
  1, 1989}, pages 332--337. {IEEE} Computer Society, 1989.

\bibitem[Vai90]{V90}
Pravin~M Vaidya.
\newblock An algorithm for linear programming which requires ${O} (((m+ n) n^{
  2} +(m+ n)^{ 1.5} n) {L})$ arithmetic operations.
\newblock {\em Mathematical Programming}, 47(1-3):175--201, 1990.

\bibitem[vdB21]{B20b}
Jan van~den Brand.
\newblock Unifying matrix data structures: Simplifying and speeding up
  iterative algorithms.
\newblock In {\em 4th Symposium on Simplicity in Algorithms, {SOSA} 2021,
  Virtual Conference, January 11-12, 2021}, pages 1--13. {SIAM}, 2021.
\newblock Available at~\url{https://arxiv.org/abs/2010.13888}.

\bibitem[Wul17]{W17}
Christian Wulff{-}Nilsen.
\newblock Fully-dynamic minimum spanning forest with improved worst-case update
  time.
\newblock In {\em Proceedings of the 49th Annual {ACM} {SIGACT} Symposium on
  Theory of Computing, {STOC} 2017, Montreal, QC, Canada, June 19-23, 2017},
  pages 1130--1143, 2017.
\newblock Available at~\url{https://arxiv.org/abs/1611.02864}.

\end{thebibliography}

\appendix
\section{Omitted Proofs}
\label{sec:proofs}
\subsection{Proof of \texorpdfstring{\cref{lemma:reducecongestion}}{reducecongestion}}
\label{subsec:proofofreducecongestion}
\begin{proof}
Note that by definition, $\prhit{u}{v}{C\cup\left\{v\right\}}=\prhit{u}{v}{C}$ for any $C\subseteq V$.
We let $C$ be the set of endpoints of $\beta m$ edges chosen randomly. Then a set of $\Omega(\beta^{-1} \log m)$ distinct vertices intersects $C$ with high probability.
Let $F=V\setminus C$.

Let $\mathsf{dis}(u, C)$ be the expected number of distinct vertices visited by a random walk from $u$ to $C$. For any $C$, we have 
\begin{align}
\sum_{v\in F}\sum_{u\in F}\deg_u\prhit{u}{v}{C}
\le \sum_{u\in F} \mathsf{dis}(u, C) \cdot \deg_u \label{eq:doublecount}
\end{align} by double counting. Because of the distribution of $C$, $\mathsf{dis}(u, C)=\O(\beta^{-1})$ with high probability for every $u\in F$ by a union bound. Thus $ \sum_{u\in F} (\mathsf{dis}(u, C))\cdot \deg_u=\O(m\beta^{-1})$.

 We choose 
 \[
 D=C\cup\left\{v : v \in F \text{ and } \sum_{u\in F}\deg_u\prhit{u}{v}{C}\ge h\right\}
 \] where $h=\Omega(\beta^{-2}\polylog n)$ has a large enough $\polylog$ factor such that 
 \[
\left(\sum_{u\in F} \mathsf{dis}(u, C) \cdot \deg_u\right)/h\le \beta m.
 \] By Markov's inequality and \cref{eq:doublecount}, there are at most $\beta m$ vertices in $D\setminus C$. To see that $D$ is a valid choice of $\hat{C}$, note that for any $v\in V\setminus D$ we have
 \[
 \sum_{u\in F}\deg_u\prhit{u}{v}{D}\le \sum_{u\in F}\deg_u\prhit{u}{v}{C} = O(h) = \O(\beta^{-2})
 \]
 as desired. To make the proof algorithmic, we need to calculate $\sum_{u\in F}\deg_u\prhit{u}{v}{C}$ for every $v\not\in C$. This can be achieved by sampling: Let $\rho=O(\log n)$ be the sampling overhead. We sample $\deg_u\cdot \rho$ random walks from each $u\in F$ to $C$ by the method in \cref{lem:sample_exit_of_rw}.
 We estimate the value of $\sum_{u\in F}\deg_u\prhit{u}{v}{C}$ by
 \[
 \rrhotil_v
 \defeq
 \frac{\text{number of random walks that hit $v$}}{\rho}.
 \]
 We will prove that if $\sum_{u\in F}\deg_u\prhit{u}{v}{C}\ge h$,
 with high probability, 
\begin{equation}
\abs{\rrhotil_v
-
\sum_{u\in F}\deg_u\prhit{u}{v}{C}} \le \frac{1}{2}\sum_{u\in F}\deg_u\prhit{u}{v}{C}.
\label{eq:estimate_proj2}
\end{equation}
 
 Assuming the property above, we may choose  \[
 \Chat=C\cup\left\{v\middle| v\not\in C\wedge \rrhotil_v\ge h/2\right\}
 \] as the desired set with high probability because 
 \begin{enumerate}
     \item With high probability, $D\subseteq \Chat$.
     \item With high probability, $\sum_{v\in V\setminus C}\rrhotil_v=\O(m\beta^{-1})$. Then $\abs{\Chat} \le \O(m\beta^{-1})/(t/2)=O(m\beta)$.
 \end{enumerate}
Next we prove \cref{eq:estimate_proj2}. We fix some $v\in V\setminus C$. For each random walk we sample, we create one random variable that is $\frac{1}{\rho}$ if it hits $v$ and $0$ otherwise. Then $\rrhotil_v$ is the sum of the random variables and the expected value of $\rrhotil_v$ is exactly $\sum_{u\in F}\deg_u\prhit{u}{v}{C}$. 

By a corollary of Bernstein's inequality (\cref{lem:hoeffding}), we have 
\begin{align*}
&\pr{\text{random walks}}{\abs{\rrhotil_v- \sum_{u\in F}\deg_u\prhit{u}{v}{C}} > \frac{1}{2}\sum_{u\in F}\deg_u\prhit{u}{v}{C}} \\
&\le 2\exp\left(-\frac{\left(\left(\sum_{u\in F}\deg_u \cdot \prhit{u}{v}{C}\right)/2\right)^2}{6\left(\sum_{u\in F}\deg_u\cdot \prhit{u}{v}{C}\right)/\rho}\right)
=n^{-\Omega(1)}
\end{align*}
for any vertex $v$ satisfying $\sum_{u\in F}\deg_u\prhit{u}{v}{C}\ge h$.
\end{proof}
\begin{algorithm}[!ht]
\caption{\textsc{CongestionReductionSubset}. \label{algo:reduce_congestion2}}
\SetKwProg{Globals}{global variables}{}{}
\SetKwProg{Proc}{procedure}{}{}
\Globals{}{
    $\beta$: size of terminals. \\
    $C$: the terminal set. \\
    $\rho$: the sampling overhead.\\
    $h$: threshold of congestion.\\
}
\Proc{$\textsc{CongestionReductionSubset}(G, \beta)$}{
    $P\assign $ a random subset of $E$ with size $\beta m$.\\
    $C=\bigcup_{e\in P}e$.\\
    Assign $\pp_v\assign 0$ for every $v\in V$.\\
    \For{$u\in V\setminus C$} {
        \For{$1\le i\le \deg_u\cdot \rho$} {
            Sample a random walk $w$ from $u$ until it hits $C$ by \cref{lem:sample_exit_of_rw}.\\
            \For{each vertex $v$ visited by $w$} {
                $\pp_v\assign \pp_v+\frac{1}{\rho}$.\\
            }
        }
    }
    \Return $C\cup\{u\mid \pp_u\ge h/2\}$.
}
\end{algorithm}

\subsection{Proof of \texorpdfstring{\cref{lem:heavyhitter}}{heavyhitter}}
\label{proofs:heavyhitter}
\begin{proof}
We follow the proof of \cite{KNPW11}.

Form an $O(\log{m})$ depth binary tree over $[m]$. These correspond to intervals on $[m]$, and the children of every node/interval have approximately half its length.

In each level, sample $O(\eps^{-2}\log{m})$ random subsets of the nodes/intervals at that level, where each node is picked with probability $\eps^{2}/100$. This way, every subset contains on average $\eps^2$ intervals and $\Theta(\eps^2 m)$ elements, and each interval is in $\O(1)$ sets with high probability.

For each such subset $S$ sampled, create an $0.1$-error $\ell_2$ sketch by using a variant of the Johnson-Lindenstrauss Lemma \cite{Ach03} using only $\pm 1$ entries. Precisely, for $N^\mathrm{JL} = O(\log m)$, we build a $\pm 1$ matrix $\mQ_S \in {\pm 1}^{N^\mathrm{JL} \times S}$ such that for any vector $\yy \in \R^S$ with high probability we have
\begin{align*} 0.9\|\yy\|_2^2 \le (N^\mathrm{JL})^{-1} \|\mQ_S \yy\|_2^2 \le 1.1\|\yy\|_2^2. \end{align*}
Let the matrix $\mQ$ be the concatentation of these sketch matrices. Accounting for the $O(\log m)$ levels of the tree, and $O(\eps^{-2} \log m)$ subsets per level gives $\mQ \in \{-1, 0, 1\}^{O(\eps^{-2}\log^{3}{m} \times m)}.$

Now we implement \textsc{Recover}.
Given $\vv$ with $\|\vv - \mQ \xx\|_\infty \le \eps/100$, we show that for any subset $S$ above, we can estimate $\|\xx_S\|_2^2$ to multiplicative accuracy $0.1$, plus additive error $\eps/100$. Indeed, by the guarantee on $\vv$, consider the vector $\vv_S \in \R^{N^\mathrm{JL}}$ that corresponds to rows for $\mQ_S \xx_S$ in $\mQ\xx$, we know that
\begin{align}
(N^\mathrm{JL})^{-1/2}\|\vv_S\|_2 &\le (N^\mathrm{JL})^{-1/2}\|\mQ_S \xx_S - \vv_S\|_2 + (N^\mathrm{JL})^{-1/2}\|\mQ_S \xx_S\| \nonumber \\
&\le \|\mQ_S \xx_S - \vv_S\|_\infty + 1.1\|\xx_S\|_2 \le \eps/100 + 1.1\|\xx_S\|_2. \label{eq:estimate}
\end{align}
Similarly, we can show that
\begin{align} (N^\mathrm{JL})^{-1/2}\|\vv_S\|_2 \ge 0.9\|\xx_S\|_2 - \eps/100. \label{eq:estimate2} \end{align}

Our algorithm for \textsc{Recover} is as follows: we walk down the binary tree, only keeping intervals such that for
\emph{every} subset $S$ containing it, the $\ell_2$-norm estimate given by $(N^\mathrm{JL})^{-1/2}\|\vv_S\|_2$ is at least $\eps/2.$

Note that if an interval has $\ell_2$-norm at least $\eps,$ then \eqref{eq:estimate2} proves that it will be kept with high probability. This shows the correctness of the algorithm.

On the other hand, if an interval has $\ell_2$-norm at most $\eps/10$, then a random subset $S$ containing that interval has $\ell_2$ norm at most $\eps/10 + \eps/10 = \eps/5$ in expectation, as each interval is kept in the subset with probability $\eps^2/100,$ and the total $\ell_2$ norm is $1$. Therefore, by \eqref{eq:estimate}, with high probability some subset containing that interval will have $\ell_2$ norm less than $\eps/2.$

To analyze the runtime, note that at most $O(\eps^{-2})$ intervals at each level have $\ell_2$ norm at least $\eps/10.$ Hence, while walking down the binary tree, we process only $O(\eps^{-2})$ intervals per level (i.e. children of intervals we kept at the previous level), and the final list contains $O(\eps^{-2})$ intervals. The running time is $O(\log^2 m)$ per interval, times $O(\log m)$ levels, times $O(\eps^{-2})$ intervals per level, for $O(\eps^{-2} \log^3 m)$ in total.
\end{proof}

\subsection{Proof of \texorpdfstring{\cref{lemma:projrandwalk}}{projrandwalk}}
\label{proofs:projrandwalk}
\begin{proof}
Expressing $\mL_{FF}$ as its diagonal minus adjacency:
\[
\mD_{FF} - \mA_{FF}
\]
where $\mD$ denotes degrees in the entire graph,
we get
\[
\mL_{FF}^{-1}
= 
\mD_{FF}^{-1} + \mD_{FF}^{-1} \mA_{FF} \mD_{FF}^{-1}
+ \mD_{FF}^{-1} \mA_{FF} \mD_{FF}^{-1} \mA_{FF} \mD_{FF}^{-1} +
\ldots
\]
which substituted into the formula for $\ppi^C(\dd)$ above gives
\[
\left[\ppi^C\left(\dd\right)\right]_v
=
\dd_v - \sum_{t = 0}^{\infty} 
\mL_{CF} \mD_{FF}^{-1} \left( \mA_{FF} \mD_{FF}^{-1} \right)^t- \dd_{F}
=\dd_v+
\sum_{t = 0}^{\infty}
\mA_{CF} \mD_{FF}^{-1} \left( \mA_{FF} \mD_{FF} ^{-1}\right)^t \dd_{F}
\]
where the last equality uses $\mA_{CF} = -\mL_{CF}$.
By inspection, the sum evaluates to
\[
\sum_u \dd_u
\cdot \prhit{u}{v}{C}.
\]
\end{proof}

\subsection{Proof of \texorpdfstring{\cref{lem:hoeffding}}{hoeffding}}
\label{subsec:bernstein}
\begin{proof}[Proof of \cref{lem:hoeffding}]
Let $\sigma_i^2$ be the variance of $X_i$. Define $E_i=\expec{}{X_i}$.
\[
\sigma_i^2 =E_i^2\left(1-\frac{E_i}{a_i}\right)+\left(a_i-E_i\right)^2\frac{E_i}{a_i}\le |E_i|^2+|E_i||a_i|.
\]
Then 
\[
\sum_{i=1}^n \sigma_i^2 =\sum_{i=1}^n \abs{E_i}^2 + \sum_{i=1}^n \abs{E_i}\abs{a_i}\le 2EM.
\]
The result follows by applying Bernstein's inequality for zero-mean random variables
\[
\pr{}{S-\E[S]>t}\le \exp\left(-\frac{t^2/2}{\sum_{i=1}^n\sigma_i^2+tM/3}\right)
\] to both $X_1-E_1, \ldots, X_n-E_n$ and $-X_1+E_1, \ldots, -X_n+E_n$.
\end{proof}

\subsection{Proof of \texorpdfstring{\cref{lemma:recenter}}{recenter2}}
\begin{proof}
We provide a proof sketch, as our approach is identical to that in \cite{LS20:arxiv} Section 4.2.

Let
\[
V_e\left(\ff_e\right)
\defeq
-\log\left( \uu_e-\ff_e\right)
-\log\left( \uu_e+\ff_e\right)
\]
be the logarithmic barrier on an edge $e$,
so that the total barrier is
\[
V\left(\ff\right)
\defeq
\sum_e V_e\left(\ff_e\right).
\]

Now define a smoothed function $\widetilde{V_e}(\ff_e)$ that agrees with $V_e$ for
\[
\ff_e
\in
\left[
\fftil_e - \frac{1}{5} \rr\left(\fftil\right)_e^{1/2},
\fftil_e + \frac{1}{5} \rr\left(\fftil\right)_e^{1/2}
\right]
\]
and has constant second derivative outside this interval.
Define the smoothed barrier
\[
\widetilde{V}\left(\ff\right)
\defeq
\sum_e \widetilde{V_e}\left(\ff_e\right).
\]

Note that
\[
\ff\left(\mu\right)
=
\argmin_{\mB^\top\ff = \left(F^*-\mu\right)\cchi_{st}}
\widetilde{V}\left(\ff\right),
\]
because $\widetilde{V}(\ff)$ agrees with $V(\ff)$ on a neighborhood of $\ff(\mu)$ by the hypothesis of \cref{lemma:recenter}.
Additionally, the Hessian of $\widetilde{V}(\ff)$,
i.e. $\nabla^2 \widetilde{V}(\ff)$ is within a factor of $2$
everywhere because each $\widetilde{V_e}(\ff_e)$ has constant 
second derivative outside a small interval.
Therefore, the minimizer can be computed to high accuracy in
$\O(1)$ iterations of Newton's method.
Each iteration can be implemented in $\O(m)$ by solving
a Laplacian system (\cref{thm:lap}).
\end{proof}

\subsection{Proof of \texorpdfstring{\cref{lemma:recentererror}}{recentererror}}
We will require a result of the proof of Lemma 4.4 in \cite{LS20_STOC}.
\begin{lemma}[Implicit in \cite{LS20_STOC} Lemma 4.4]
\label{lemma:correction}
Let $\gg$ be a $\gamma$-centered flow.
Then there is a flow $\gghat$ such that
$\mB^\top\gghat = \mB^\top\gg$,
$\gghat$ is $5\gamma^2$-centered,
and $\|\mR(\gg)^{-1}(\gg - \gghat)\|_\infty \le \gamma$.
\end{lemma}

We can now show \cref{lemma:recentererror}.
\begin{proof}[Proof of \cref{lemma:recentererror}]
We first resolve the demand error by projecting onto the subspace of
flows that route $(F^{*} - \mu) \cchi_{st}$,
that is
\[
\left\{ \gg : \mB^\top \gg = \left( F^{*} - \mu \right) \cchi_{st} \right\}.
\]
Specifically, consider the flow $\gg = \ff - \Delta$ for 
\[
\Delta
=
\ff - \mR\left(\ff\right)^{-1}
\mB\mL\left(\rr\left( \ff \right)\right)^\dagger
\left(\mB^\tomato \ff - \left(F^{*} - \mu \right) \cchi_{st} \right).
\]
By construction, we have $\mB^\top \gg = (F^{*} - \mu) \cchi_{st}$.
Note that
\[
\norm{
\mR\left(\ff\right)^{1/2}\Delta
}_2
=
\norm{
\mB^\top\ff - \left( F^{*} - \mu \right) \cchi_{st}
}_{\mL\left(\rr\left( \ff \right) \right)^\dagger}
\le
1/1000.
\]
This tells us that $\mR(\gg) \approx_{0.01} \mR(\ff)$,
which means the centrality error of $\gg$ is bounded by
\begin{align*}
&\norm{
\mB\pphi - \left(\frac{\oone}{\uu-\ff+\Delta} - \frac{\oone}{\uu+\ff-\Delta}\right)}_{\mR\left(\gg\right)^{-1}}
\\ &\le
1.01\left(\norm{\mB\pphi - \left(\frac{\oone}{\uu-\ff} - \frac{\oone}{\uu+\ff}\right)}_{\mR\left(\gg\right)^{-1}}
+
3
\norm{\left(\frac{1}{\left(\uu-\ff\right)^2}
+
\frac{1}{\left(\uu+\ff\right)^2}\right) \circ \Delta
}_{\mR\left(\ff\right)^{-1}}\right) \\
&\le
1.01\left(\frac{1}{1000} + \frac{3}{1000}\right)
\le
\frac{1}{200}.
\end{align*}

Now, we will iteratively apply \cref{lemma:correction} to $\gg$ and take a limit to the $0$ centered point. Formally, define $\gg^{(0)} = \gg$ and $\gg^{(i+1)} = \gghat^{(i)}$ as in \cref{lemma:correction} for $i \ge 0$, and
\[
\gghat = \lim_{i \to \infty} \gg^{\left(i\right)}.
\]
Additionally, set
\[
\gamma^{\left(0\right)}
=
\frac{1}{200},
\]
and inductive
\[
\gamma^{\left(i+1\right)}
=
5\left(\gamma^{\left(i\right)}\right)^2
\qquad \text{for $i \ge 0$}.
\]
We know that $\gg^{(i)}$ is $5\gamma^{(i)}$-centered by \cref{lemma:correction}, hence $\gghat$ is $0$-centered.
Note that $\gamma^{(i)} \le \frac{2^{-i}}{200}$

We claim by induction that $\mR(\gg) \approx_{0.1} \mR(\gg^{(i)})$ for all $i$.
Assuming this and using \cref{lemma:correction} gives us
\[
\norm{\mR\left(\gg\right)^{1/2}\left(\gg - \gghat\right)}_\infty
\le
1.1\sum_{i \ge 0}
\norm{\mR\left(\gg^{\left(i\right)}\right)^{1/2}
\left(\gg^{\left(i+1\right)} - \gg^{\left(i\right)}\right)}_\infty
\le
1.1 \sum_{i \ge 0} \frac{2^{-i}}{200} \le \frac{1}{50}.
\]
This completes the induction. Combining this with the above proves that
\[
\norm{\mR\left(\ff\right)(\ff - \gghat)}_\infty
\le
2\left(\frac{1}{200} + \frac{1}{50}\right) \le \frac{1}{10}
\]
as desired.
\end{proof}

\subsection{Proof of \texorpdfstring{\cref{lemma:l2lemma}}{l2lemma}}
\begin{proof}
By centrality conditions, we know that for all edges
\begin{align}
\mB\pphi\left(\mu\right)
=
\frac{\oone}{\uu-\ff\left( \mu \right)}
- \frac{\oone}{\uu + \ff\left(\mu \right)}
\enspace \text{ and } \enspace
\mB\pphi\left(\muhat\right)
=
\frac{\oone}{\uu-\ff\left(\muhat\right)}
-
\frac{\oone}{\uu+\ff\left(\muhat\right)} \label{eq:cp}.
\end{align}

We first bound the change in $s$-$t$ voltage drop
from $\mu$ to $\muhat$, and specifically show
\[
\abs{\cchi_{st}^\tomato
\left( \pphi\left(\muhat\right) - \pphi\left(\mu \right) \right)}
\le
\frac{500k\sqrt{m}}{\mu}.
\]
We prove this by integrating along the central path.
By differentiating the centrality condition with respect to $\mu$ (and replacing $\mu$ with $\nu$ for notational purposes) we get
\[
\mB\left(d\pphi\left(\nu \right)\right)
=
\left(\frac{1}{\left(\uu-\ff\left(\nu\right)\right)^2}
+
\frac{1}{\left(\uu+\ff\left(\nu \right)\right)^2} \right)
d\ff\left(\nu\right),
\]
where $\mB^\top d\ff(\nu) = d\cchi_{st}.$
Therefore, we know that for resistances
$\rr(\ff(\nu)) = (\uu-\ff(\nu))^{-2} + (\uu + \ff(\nu))^{-2}$
and its associated graph Laplacian, we have
\[
d\pphi\left(\nu\right)
=
\mL\left( \rr \left( \ff \left( \nu \right) \right) \right)^{\dag} \cchi_{st},
\]
which when dotted against $\cchi_{st}$ again gives
\[
\cchi_{st}^\tomato d\pphi\left(\nu\right)
=
\cchi_{st}^\top \mL\left( \rr \left( \ff \left( \nu \right) \right) \right) ^\dagger \cchi_{st}.
\]
The right hand side denotes the electric energy of routing one unit of $a$-$b$ flow.
By \cref{lemma:elecenergy}, the $s$-$t$ electric energy is at most
\[
\frac{400m}{\muhat^2}
\le
\frac{500m}{\mu^2}
\]
by the assumption of $k \le \sqrt{m}/10.$
Therefore, integrating gives us
\[
\abs{\cchi_{st}^\tomato
\left(\pphi\left( \muhat\right) - \pphi\left(\mu\right)\right)}
=
\abs{\cchi_{st}^\tomato
\int_{\muhat}^{\mu}
d\pphi\left( \nu \right) d\nu}
\le
\frac{500m}{\mu^2} \cdot \left( \mu - \muhat \right)
\le
\frac{500k\sqrt{m}}{\mu}.
\]

We finish the proof by substituting this in.
Let the change in $\ff(\cdot)$ be $\Delta \ff$, and $\pphi(\cdot)$ be $\Delta \pphi$:
\[
\Delta \ff
\defeq
\ff\left( \muhat \right)
-
\ff\left( \mu \right)
\enspace \text{ and } \enspace
\Delta \pphi
\defeq
\pphi\left( \muhat \right)
-
\pphi\left( \mu \right).
\]
Subtracting the equations in \eqref{eq:cp} using the scalar identities
\begin{align*}
\frac{1}{u - (f + \Delta)}
-
\frac{1}{u - f}
=
\frac{\left(u - f\right) - \left( u - f - \Delta \right)}{\left( u - f - \Delta \right) \left(u - f\right)}
&=
\frac{\Delta}{\left( u - f - \Delta \right) \left(u - f\right)},\\
\frac{1}{u + (f + \Delta)}
-
\frac{1}{u + f}
=
\frac{\left(u + f\right) - \left( u + f + \Delta \right)}{\left( u + f + \Delta \right) \left(u + f\right)}
&=
\frac{-\Delta}{\left( u + f + \Delta \right) \left(u + f\right)},
\end{align*}
applied per vector entry, gives us
\[
\mB \Delta \pphi
=
\Delta \ff
\circ
\left[
\frac{\oone}{\left(\uu-\ff\left( \mu \right) \right)
\circ \left(\uu-\ff\left( \muhat \right) \right)}
+
\frac{\oone}{\left(\uu+\ff\left( \mu \right) \right)
\circ \left(\uu+\ff\left( \muhat \right) \right)}
\right] \label{eq:voltagedrop}
\]
Multiplying $\Delta\ff^\top$ by both sides of \eqref{eq:voltagedrop} gives us
\[
\Delta \ff^{\tomato}
\left[\Delta\ff \circ \left[ 
\frac{\oone}{\left(\uu-\ff\left( \mu \right) \right)
\circ \left(\uu-\ff\left( \muhat \right) \right)}
+
\frac{\oone}{\left(\uu+\ff\left( \mu \right) \right)
\circ \left(\uu+\ff\left( \muhat \right) \right)}
\right]
\right]
=
\Delta \ff^{\tomato}
\left( \mB \Delta \pphi \right).
\]
Applying the operators in the other order, specifically grouping $\mB$ with $\Delta \ff$ gives
\[
=
\Delta \pphi^{\tomato}
\left(\mB^{\tomato} \Delta \ff \right)
=
\Delta \pphi^{\tomato}
\left(\mB^{\tomato} \ff\left( \muhat \right) - \mB^{\tomato} \ff\left( \mu \right)\right)
=
\Delta \pphi^{\tomato}
\left( - \frac{k \mu}{\sqrt{m}} \cchi_{st} \right)
\]
where the last equality is by definition of the central
path solutions, specifically the demand they meet.
So the magnitude is at most
\[
\abs{\Delta \pphi^{\tomato}
\left( \frac{\mu k}{\sqrt{m}} \cchi_{st} \right)}
=
\frac{\mu k}{\sqrt{m}}
\cdot
\abs{\cchi_{st}^{\tomato} 
\Delta \pphi}
\leq
\frac{\mu k}{\sqrt{m}} \cdot 
\frac{500 k\sqrt{m}}{\mu}
\leq
500 k^2.
\]
Combining these equations gives the desired result.
\end{proof}

\subsection{Proof of \texorpdfstring{\cref{lemma:l2change}}{l2change}}
\begin{proof}
For $\nu \in [\muhat, \mu]$ let $d\ff(\nu)$ denote the differential change in the flow $\ff(\nu)$ with respect to the path parameter. As with the proof of \cref{lemma:l2lemma} above, we have that
\[ d\ff(\nu) = \mR(\ff(\nu))^{-1}\mB\mL(\ff(\nu))^\dagger\cchi_{st}. \]
Therefore, we get a total $\ell_2$ change bound over time of
\begin{align}
    \int_{\muhat}^\mu \left\|d\ff(\nu)\right\|_{\mR(\ff(\nu))}^2 &= \int_{\muhat}^\mu \left\|\cchi_{st}\right\|_{\mL(\ff(\nu))^\dagger}^2 \nonumber
    \\ &\le \int_{\muhat}^\mu \frac{400m}{\nu^2} d\nu \le O(m\mu^{-2}(\mu-\muhat)) = O(\mu^{-1}T\sqrt{m}).
    \label{eq:l2intbound}
\end{align}
where the first inequality follows from the energy bound in \cref{lemma:elecenergy}, and the second from $\muhat \ge 9\mu/10.$ Next, we lower bound the total $\ell_2$ contributed by an edge $e$ in terms of $\change(\gamma, e).$ Consider an edge $e$, and let $\tau = \change(e, \gamma).$

For any $i \in [\tau]$, we consider two cases. The first case is if $\rr(\ff(\nu)) \ge \rr(\ff(\mu^{(i)}))/2$ for all $\nu \in [\mu^{(i+1)}, \mu^{(i)}]$, i.e. $\rr(\ff(\nu))$ was multiplicatively stable on the range. In this case, we have
\begin{align*}
    \int_{\mu^{(i+1)}}^{\mu^{(i)}} \rr(\ff(\nu))_e[d\ff(\nu)_e]^2 &\ge \frac14\rr(\ff(\mu^{(i)})) \int_{\mu^{(i+1)}}^{\mu^{(i)}} [d\ff(\nu)_e]^2 \\
    &\ge \frac{\rr(\ff(\mu^{(i)}))}{4(\mu^{(i)}-\mu^{(i+1)})} \left(\int_{\mu^{(i+1)}}^{\mu^{(i)}} d\ff(\nu)_e \right)^2 \\
    &= \frac{\rr(\ff(\mu^{(i)}))}{4(\mu^{(i)}-\mu^{(i+1)})}\left|\ff(\mu^{(i)})_e - \ff(\mu^{(i+1)})_e\right|^2 \\
    &\ge \frac{\gamma^2}{4(\mu^{(i)}-\mu^{(i+1)})} \ge \Omega\left(\frac{\gamma^2\sqrt{m}}{\mu T}\right).
\end{align*}
Here, the first inequality follows from the assumption that $\rr(\ff(\nu)) \ge \rr(\ff(\mu^{(i)}))/2$ for all $\nu \in [\mu^{(i+1)}, \mu^{(i)}]$, the second inequality follows from Cauchy-Schwarz, and the third inequality follows from the definition of $\mu^{(i)}$.

In the other case where $\rr(\ff(\nu)) < \rr(\ff(\mu^{(i)}))/2$ for some $\nu \in [\mu^{(i+1)}, \mu^{(i)}]$, the same bound holds. Indeed, we can use an identical proof but replace $\mu^{(i+1)}$ with the largest $\nu$ satisfying $\rr(\ff(\nu)) < \rr(\ff(\mu^{(i)}))/2$. In all cases, this gives us
\begin{align}
\int_{\muhat}^{\mu} \rr(\ff(\nu))_e[d\ff(\nu)_e]^2 \ge \sum_{i=1}^{\change(e, \gamma)} \int_{\mu^{(i+1)}}^{\mu^{(i)}} \rr(\ff(\nu))_e[d\ff(\nu)_e]^2 \ge \Omega\left(\frac{\gamma^2\sqrt{m}}{\mu T}\right) \cdot \change(e, \gamma). \label{eq:singleedgebound}
\end{align}
Combining \cref{eq:singleedgebound,eq:l2intbound} using Markov's inequality gives that
\[ \sum_e \Omega\left(\frac{\gamma^2\sqrt{m}}{\mu T}\right) \cdot \change(e, \gamma) \le O(\mu^{-1}T\sqrt{m}), \]
so $\sum_{e \in E(G)} \change(e, \gamma) \le O(T^2\gamma^{-2})$ as desired.
\end{proof}
\newpage
\section{Table of Variables}
\label{appendix:VAR}

\begin{table}[H]
    \centering
    \begin{tabular}{|l|l|} \hline
        {\bf{Variable}} & {\bf{Definition}} \\ \hline
        $\uu$ & Edge capacities (assuming undirected graph by standard reductions)\\ \hline
        $\ff$ & Flow \\ \hline
        $\uu^{+}(\ff), \uu^{-}(\ff)$ & Upper/lower residual capacities for the flow $\ff$, $\uu^{+}(\ff) = \uu - \ff$, $\uu^{-}(\ff) = \ff + \uu$\\ \hline
        $\uu(\ff)$ & Smaller residual capacity of $\ff$,
        $\uu(\ff)_e = \min\{ \uu^{+}(\ff)_e, \uu^{-}(\ff)_e\} = \uu_e - | \ff_e|$\\ \hline
        $\rr(\ff)$ & Resistance induced by residual capacities of $\ff$,
        $\rr(\ff)_e  = \uu^{+}(\ff)_e^{-2} + \uu^{-}(\ff)_e^{-2}$\\ \hline
        $F^*$ & Optimal flow value (assume known via binary search) \\ \hline
        $\mu$ & Central path parameter that corresponds to the amount of residual flow left \\ \hline
        $\ff(\mu)$ & Central path flow routing $F^*-\mu$ units, aka. $\mB^{\tomato} \ff(\mu) = (F^{*} - \mu) \cchi_{st}$\\ \hline
        $\pphi(\mu)$ & Central path dual variable corresponding to $\ff(\mu)$\\ \hline
        $\mL(\ff)$ & Laplacian with resistances $\rr(\ff)$ \\ \hline
        $\Delta \ff$ & (Electric) flow we augment by \\ \hline
        $s, t$ & Source/sink vertices\\ \hline
        $\cchi_{st}$ & Indicator vector with $-1$ at $s$, $1$ at $t$. \\ \hline
        $\mB$ & Edge vertex incidence matrix \\ \hline
        $\zzero, \oone$ & All-$0$s / $1$s vector \\ \hline
        $\circ$ & Hadamard product, $(\xx \circ \yy)_{i} = \xx_i \yy_i$\\ \hline
        $\frac{\xx}{\yy}$ & (Overloaded) entry-wise division,
            $(\frac{\xx}{\yy})_{i} = \frac{\xx_i}{\yy_i}$\\ \hline
        $\err$ & Demand error vector, $\err = \mB^\tomato \ff - (F^*-\mu) \cchi_{st}$ for some approximately central $\ff$ \\ \hline
        $\cerr$ & Centrality error vector \\ \hline
        $\dd[i]$ & The $i$-th demand vector tracked by \textsc{Locator} \\ \hline
        $\ppi^C(\dd)$ & The projection of a demand vector $\dd\in \mathbb{R}^V$ onto a terminal set $C\subseteq V$ \\ \hline
        $\ppitil^C(\dd)$ & A variable that stores the estimated value of $\ppi_C(\dd)$ \\ \hline
        $S$ & Subset of edges $e$ with $\rr_e \in [\eps^2/20, 20\eps^{-2}]$. \\ \hline
    \end{tabular}
    \caption{Summary of recurring variables}
    \label{tab:variable}
\end{table}

\begin{table}[ht]
    \centering
    \begin{tabular}{|m{0.15\textwidth}|m{0.4\textwidth}|l|} \hline
        {\bf{Parameter}} & {\bf{Definition}} & {\bf{Value}} \\ \hline
        $k$ & Frequency of recentering & $m^{1/328}$ \\ \hline
        $\eps_\step$ & Step size is $\eps_\step/\sqrt{m}$ & $k^{-3}$ \\ \hline
        $\eps_\solve$ & Solve accuracy of Laplacian & $k^{-3}$ \\ \hline
        $\eps$ & Update edges with $\eps^2$ fraction of energy & $k^{-6}$ \\ \hline
        $\beta_{\Locator}$ & Size of terminals of \Locator~is $\beta m$ & $k^{-16}$ \\ \hline
        $\beta_{\Checker}$ & Size of terminals of \Checker~is $\beta m$ & $k^{-28}$ \\ \hline
        $\delta$ & Frequency to exactly repair projections inside \Locator & $k^{-38}$ \\ \hline
        $\overline{count}$ & Maximum number of terminals to add after a \textsc{Projection.InitProjections} & $\eps\delta^{-1} = k^{32}$ \\ \hline
    \end{tabular}
    \caption{Key parameters and their value in terms of $k$ (up to $\polylog$s),
    the amount of progress along the central path that we make
    in $\O(m)$ amortized time.}
    \label{tab:parameters}
\end{table}

\end{document}